\documentclass[english,11pt]{article}
\usepackage[T1]{fontenc}
\usepackage[utf8]{inputenc}
\usepackage{comment}
\usepackage[english]{babel}
\usepackage{a4wide}

\usepackage{amsmath,amssymb,amsopn,dsfont,marvosym,stmaryrd}
\usepackage{threeparttable,multirow,adjustbox}
\usepackage[]{placeins}
\usepackage{flafter}
\usepackage{booktabs}
\usepackage{pdflscape}
\usepackage{longtable}
\usepackage{booktabs}
\usepackage{fancyhdr}
\usepackage{float} 
\usepackage{rotating}
\usepackage{tikz}
\usepackage{pgfplots}
\usepgfplotslibrary{groupplots}
\usepgfplotslibrary{fillbetween}
\pgfplotsset{compat=1.18}
\usepackage{verbatim}
\usepackage{eurosym}
\usepackage{array}
\usepackage{xr}
\usepackage{graphicx}
\usepackage{subcaption}
\usepackage{natbib}
\usepackage{makecell}

\usepackage[normalem]{ulem}
\usepackage{xcolor}

\usepackage{chngcntr}

\definecolor{addblue}{RGB}{0,70,180}
\definecolor{delred}{RGB}{180,30,30}

\newtheorem{thm}{Theorem}[section]

\newtheorem{lem}{Lemma}[section]
\newtheorem{prop}[thm]{Proposition}
\newtheorem{hyp}{Assumption}

\definecolor{darkblue}{RGB}{0,55,160}
\definecolor{darkgreen}{RGB}{0,110,60}
\definecolor{burgundy}{RGB}{140,20,40}
\definecolor{lightgray}{RGB}{245,245,245}
\definecolor{medgray}{RGB}{200,200,200}
\definecolor{boxblue}{RGB}{230,240,255}
\definecolor{boxgreen}{RGB}{230,248,235}
\definecolor{boxamber}{RGB}{255,248,225}
\definecolor{boxcoral}{RGB}{255,235,230}

\usepackage{layout}
\usepackage[colorlinks=true,	
linkcolor=red,
urlcolor=blue,
citecolor=blue]{hyperref}

\usepackage[margin=1in]{geometry}

\usepackage{setspace}
\title{\vspace{-1.5cm} 
Targeting Support Using Job Seekers’ Biases: \\ A Randomized Experiment\thanks{\footnotesize{Cr\'epon: CREST, crepon@ensae.fr; Frot: France-Travail, CREST, aurelien.frot@ensae.fr; Gaillac: University of Geneva, GSEM: christophe.gaillac@unige.ch. This paper is the result of a partnership with France Travail, the Public Employment Service in France, and does not represent their views. We thank Anita Bonnet, Cyril Nouveau, and Chantal Vessereau, for their support and feedback throughout this project. This research was supported by the Chaire Sécurisation des Parcours Professionels. Authors retained full intellectual freedom throughout this process, all errors are our own.  We thank  Roland Rathelot, Arne Uhlendorff, as well as seminar participants at INSEE, the chaire travail seminar of PSE for useful comments and suggestions.}}}

\vspace{0.5cm}

\author{Bruno Crépon, Aurélien Frot, Christophe Gaillac}

\begin{document}

\maketitle

\begin{abstract}
\noindent Most digital job-search assistance encourages unemployed workers to broaden their search toward related occupations, targeting one important source of search inefficiency: insufficient occupational diversification. Our analysis suggests that the relevant margin of adjustment depends on the underlying search problem. Building on a detailed analysis of job seekers' beliefs and search behavior, we identify a large group of pessimistic workers for whom the main constraints are low search effort and low aspirations, rather than insufficient occupational diversification. This diagnosis points to an unexpected intervention: rather than encouraging these workers to search in new occupations, we encourage them to search more intensively and apply for better-paying jobs within the occupations they already consider. We evaluate this diagnosis-based intervention, alongside a standard occupational recommendation, in a large-scale randomized experiment conducted with the French Public Employment Service. The motivational intervention increases search effort, raises reservation wages, and improves reemployment outcomes along the predicted margins. Occupational recommendations, by contrast, primarily benefit workers whose search problem lies in the allocation of attention across occupations and operate by activating existing perceptions rather than correcting beliefs. More broadly, our findings show how digital platforms can combine subjective expectations, behavioral data, targeted interventions, and randomized experimentation to diagnose job seekers' needs and iteratively improve intervention design.

\end{abstract}

\newpage

\section{Introduction}

The rapid development of online job search platforms and the increasing availability of detailed data on job seekers have profoundly transformed the way job search assistance can be designed and delivered. Public Employment Services (PES) and private platforms can now provide personalized interventions at scale, drawing on increasingly rich digital traces and labor market data to guide job seekers through increasingly complex labor markets. A prominent and highly successful class of such interventions consists in recommending alternative occupations or opportunities related to individuals' initial targets, with the aim of improving the allocation of search effort. A growing body of experimental evidence shows that these recommendations can substantially affect job search behavior and, in some settings, improve employment outcomes \citep{belot2019providing, belot2025advising, bandiera2025search, altmann2022direct}. These interventions have established occupational diversification as a powerful way of guiding job seekers through digital platforms. Beyond their immediate policy value, digital job search platforms also create new opportunities to learn about job search behavior and improve the design of job search assistance. As emphasized by \citet{kircher2022schumpeter}, they allow researchers and practitioners not only to evaluate interventions through randomized experimentation, but also to iteratively refine personalized job search assistance by learning how individuals respond to different forms of guidance.

\medskip

To date, most digital job search interventions have personalized the information delivered to job seekers while relying on a common intervention strategy. Occupational recommendations and other information-based interventions may provide different forms of guidance, but they are typically delivered uniformly across broad groups of job seekers. Despite their differences, these interventions largely seek to address the same search problem: helping workers identify opportunities they would not otherwise consider. At the same time, a growing body of evidence has revealed substantial heterogeneity in both labor market prospects and subjective beliefs. Differences in job-finding probabilities are large, persistent, and increasingly predictable \citep{mueller2025nature}, while job seekers hold systematically biased and heterogeneous beliefs about key features of the search environment \citep{manski2004measuring, spinnewijn2015unemployed, mueller2021job}. Building on this literature, our companion paper \citep{crepon2025_1} shows that these biases are not isolated misperceptions but form coherent belief profiles: biases about wage distributions, job offer arrival rates, and reemployment prospects are strongly correlated and jointly predict distinct search behaviors. These findings suggest that designing effective job search assistance requires a prior step: understanding what actually prevents different job seekers from returning to work. If different workers face different search problems, then a common intervention strategy is unlikely to be optimal. The relevant question is no longer only what information to provide, but which search friction the intervention is intended to address.

\medskip

Building on this diagnostic, our companion paper \citet{crepon2025_1} shows that different belief profiles are associated with different search problems. In particular, we identify a large group of job seekers, whom we label pessimists, who underestimate their reemployment prospects, perceive low returns to search effort, and set low reservation wages. This diagnosis leads to a striking implication. For these individuals, the binding constraint is not where to search but how. Rather than encouraging them to broaden their search toward alternative occupations, the diagnosis suggests helping them search more intensively, raise their aspirations, and target higher-paying jobs within the occupations they already consider. This points to a need-based approach to job search assistance, in which interventions are tailored not only in terms of the information they provide but also to individuals' underlying beliefs and behavioral constraints.%

\medskip

This paper implements and tests interventions derived from the behavioral diagnosis developed in our companion paper in the context of a large-scale randomized experiment ($n=52{,}465$) conducted in partnership with France Travail. We leverage prior survey data to construct predictive models of individuals' belief profiles, estimated outside the experimental sample, and use these predictions, together with administrative information, to classify job seekers into groups characterized by different search problems identified in our companion paper. We then compare two intervention strategies that address different search frictions. The first is a standard \textit{occupational} recommendation, suggesting alternative occupations close to the individual's preferred one. The second is a novel \textit{motivational} intervention, specifically designed for pessimistic job seekers: rather than encouraging search diversification, it directly targets search intensity and selectivity by encouraging higher effort, raising aspirations, and providing tailored reference points on attainable wages. We interpret both interventions in a simple behavioral search framework that distinguishes two margins of adjustment: an attention margin, which governs the set of occupations actively considered, and an aspiration margin, which shapes effort and reservation wages within that set.

\medskip

The randomized experiment strongly supports the behavioral diagnosis developed in our companion paper. The motivational intervention, designed to address the search problems identified for pessimistic job seekers, generates behavioral responses that closely match the predicted margins of adjustment. It increases search effort by 0.9 hours per week (concentrated among low-effort workers) and raises reservation wages by 1.9\% (concentrated among workers below the disclosed peer benchmark). Consistent with this diagnosis, these behavioral changes translate into improved labor-market outcomes along the same dimensions. For individuals with initially low search effort, the intervention raises reemployment probabilities by 2.9 percentage points; for those with initially low reservation wages, it improves the quality of jobs obtained, raising the probability of accepting a wage above the baseline reservation wage by 5.3 percentage points. Rather than encouraging these workers to lower their ambitions or broaden their occupational search, the intervention proves effective by increasing search effort while encouraging them to target higher-paying jobs within the occupations they already consider.

\medskip

The benchmark occupational recommendation reveals a different pattern. Among job seekers with favorable or accurate beliefs, occupational recommendations increase the number of applications submitted, particularly toward the suggested occupations, raising applications to posted vacancies by 11\% and to suggested occupations by 21\%. However, these changes in search behavior do not translate into measurable improvements in reemployment outcomes.

\medskip

A central contribution of the paper is to clarify the mechanisms underlying these effects. A longstanding question in the literature  is whether information interventions operate by updating workers' beliefs or by changing how existing perceptions are translated into search behavior. We find only limited evidence of belief updating: although behavioral responses are substantial, measured beliefs change only marginally. The effects of occupational recommendations are concentrated among individuals who already perceive opportunities outside their target occupation, suggesting that these interventions activate existing perceptions rather than generate new beliefs. For the motivational arm, reservation wages rise only for workers below the disclosed peer benchmark, and search effort rises only for those below the disclosed effort norm. This is consistent with a change in aspiration references, rather than a correction of beliefs about the wage distribution or the returns to search. Taken together, these findings suggest that the main friction does not lie in a lack of information per se, but in the difficulty of translating existing perceptions into effective search strategies. Beyond explaining the effects of the interventions, these findings illustrate how randomized experimentation can be used to distinguish among competing behavioral mechanisms and, in doing so, inform the design of future job search assistance.

\medskip

Interestingly, although pessimistic job seekers largely ignore the suggested occupational recommendations, this intervention nevertheless improves their reemployment prospects. Their response takes the form of an increase in spontaneous applications, revealing that occupational recommendations can induce behavioral adjustments beyond the margin they were designed to target. These gains in employment are also accompanied by a shift in the type of jobs obtained, with a higher incidence of fixed-term contracts. This suggests that interventions may improve employment outcomes even when they do not operate through the behavioral margin they were designed to affect, highlighting the importance of understanding not only whether interventions work, but also how they shape search behavior and labor-market outcomes.

\medskip

More broadly, these findings illustrate how digital platforms can combine administrative records, targeted online surveys eliciting subjective expectations, and randomized experimentation to diagnose job seekers' needs, design interventions around these diagnoses, and refine them by revealing the behavioral pathways through which they affect job search and reemployment.

\subsection*{Related Literature}

\textbf{Job search recommendations.} We build on a literature that provides job seekers with information to redirect their search. These interventions address different informational frictions: expanding the set of occupations considered \citep{belot2019providing, belot2025advising,belot2026long,altmann2022direct}, reorganizing attention across vacancies \citep{le2023can}, or making visible firms outside the initial choice set \citep{behaghel2024potential}. This literature shares an implicit assumption: once information is accessible or salient, workers adjust relatively smoothly. We show that this assumption fails for a substantial share of job seekers. Occupational recommendations work when the binding margin is allocation across occupations, but not when constraints lie on effort or selectivity. A recommendation must match the margin along which workers can actually adjust.

\medskip

\textbf{Biases in subjective beliefs and information provision.} Our paper builds on the literature that documents biases in job seekers' subjective beliefs about reemployment outcomes \cite{mueller2021job,spinnewijn2015unemployed,van2024predicting} or perceptions of the wage distribution of job offers or postings \citep{conlon2018labor, caliendo2023accuracy,altmann2025}. The closest comparison is \citet{altmann2025}, who provide German job seekers with reference-wage information and find higher reservation wages and slower reemployment. We differ in two ways: our disclosure uses machine learning to target below-benchmark pessimists workers, avoiding demotivation of those already above, and it is embedded in a broader motivational intervention that also provides effort norms.

\medskip

\textbf{Heterogeneity in job search.} A growing literature documents that labor market trajectories are highly heterogeneous and largely predictable \citep{mueller2025nature}, and that this heterogeneity is reflected in biased subjective expectations \citep{manski2004measuring,spinnewijn2015unemployed,mueller2021job}. Our companion paper \citep{crepon2025_1} interprets this heterogeneity structurally: expectation biases are multivariate differences in perceived arrival rates, wage distributions, and returns to effort that translate into differentiated behavior. The present paper adds that beliefs do not mechanically become behavior: workers can hold favorable perceptions without acting on them, and interventions can change behavior without revising beliefs. This resonates with evidence that policy signals are not behaviorally neutral: \citet{bandiera2025search} find that employer contact can discourage workers when silence is read as a negative signal, and \citet{banerjee2023learning} show that labor-market exposure can lead the initially optimistic to revise ambitions downward.\footnote{This is consistent with \citet{he2023updating}, who show job seekers' beliefs respond strongly to aggregate conditions.} We document the symmetric case: a targeted intervention reactivates effort and selectivity among the initially pessimistic.

\medskip

\textbf{Nudges and aspiration gap.} Both treatments are nudges, \emph{i.e.}, low-cost, choice-preserving information, and instantiate two well-studied mechanisms: social-norm provision \citep{allcott2011social} and aspiration-setting \citep{ray2006aspirations, genicot2020aspirations}. Our results speak to the ``voltage drop'' at scale \citep{dellavigna2022rcts}: \citet{bo2025scaling} find nudge efficacy falls among low-motivation individuals. We provide a structural reason: low motivation reflects pessimistic beliefs, and such workers receive the wrong nudge rather than being inherently unresponsive; a motivational message recovers positive effects. The pessimist profile mirrors mechanisms from behavioral development economics: aspirations traps \citep{ray2006aspirations}, information-driven aspiration shifts \citep{jensen2010perceived, beaman2012female}, cognitive bandwidth constraints \citep{mullainathan2013scarcity, mani2013poverty}, implementation frictions \citep{duflo2011nudging}, and locus of control \citep{caliendo2015locus, cobbclark2015locus}. The attention channel parallels \citet{hanna2014learning}'s ``learning through noticing'': agents change behavior on an already-valued dimension when attention costs fall, without new information or belief revision. Our identification complements \citet{dellavigna2017reference} and \citet{marinescu2021}, who exploit benefit exhaustion as a downward shift in reference income and show that job seekers react to perceived wage prospects. In contrast, we identify the effects of an upward shift in reference income through randomized disclosure.

\medskip

\textbf{Information design.} The paper connects to information design \citep{kamenica2011bayesian}: the PES acts as a sender who observes job seekers' belief profiles and tailors disclosure accordingly. The key departure from standard Bayesian persuasion is that behavioral responses do not require belief updating: occupational recommendations operate through attention activation, motivational recommendations through aspiration shifts. Optimal disclosure in labor markets must account not only for heterogeneous beliefs but for heterogeneous behavioral frictions.

\medskip 

The paper is organized as follow. Section~\ref{sec:beliefs} documents belief heterogeneity and the typology of job seekers. Section~\ref{sec:model} introduces the conceptual framework. Section~\ref{sec:treatment} describes the treatments. Section~\ref{sec:experiment} presents the experimental design and data. Section~\ref{sec:occ} reports results for the occupational arm. Section~\ref{sec:motivational_impact} reports results for the motivational arm. Section~\ref{sec:employment} examines reemployment outcomes. Section~\ref{subsec:interpretation} maps model predictions to the evidence. We conclude with policy implications.

\medskip

The experiment was pre-registered at the AEA RCT Registry (AEARCTR-0012218).\footnote{See \url{https://www.socialscienceregistry.org/trials/12218}.} The study protocol received ethical approval from the Institutional Review Board of the Paris School of Economics.

\section{Beliefs and Typology of Search Behavior}
\label{sec:beliefs}

This section documents job seekers’ beliefs about the labor market, the extent to which these beliefs are distorted relative to realized outcomes, and how such distortions relate to job search behavior. We leverage a series of repeated surveys conducted by France Travail, linked to administrative data, which are described in detail in \citet{crepon2025_1}. That study establishes the magnitude, heterogeneity, and persistence of belief distortions regarding reemployment prospects and key labor market parameters.

\medskip

Building on this framework, we use earlier survey waves combined with realized reemployment outcomes to construct a typology of job seekers based on systematic forecast errors. The prediction algorithm underlying this typology is trained on pre-2023 survey waves and subsequently applied to the October 2023 survey.\footnote{Our study also relies on rich administrative data from the PES. Appendix Section \ref{app:admin_data} provides additional details.}

\medskip

Throughout this section, we focus on the October 2023 wave to characterize belief heterogeneity and forecast errors in the current labor market, abstracting from the experimental interventions introduced later. The analysis is descriptive and diagnostic: it motivates the targeted informational interventions presented in Section~\ref{sec:treatment} and provides the empirical foundation for the experimental design described in Section~\ref{sec:experiment}.

\subsection{Survey on job seekers' beliefs and expectations}
\label{subsec:surveys}

France Travail conducts repeated surveys to measure job seekers’ beliefs, expectations, and search behavior. The survey design and question wording are stable across waves, ensuring consistent measurement over time and enabling linkage to subsequent labor market outcomes observed in administrative data. We describe here the main variables used in our analysis.

\paragraph{Reemployment expectations}

A central feature of the survey is the elicitation of subjective reemployment probabilities at multiple horizons. Job seekers report their perceived probability of finding a job within different time frames, conditional on standard employment criteria, through the following question, for $t \in \{1, 3, 6, 12, 24\}$:

\medskip

\textit{``What are your chances of finding a job within $t$ month(s), on a permanent or fixed-term contract of at least one month?''}.

\medskip

These responses provide a direct measure of perceived reemployment prospects across horizons. Combined with administrative data on realized outcomes, they allow us to quantify forecast errors at the individual level.

\paragraph{Beliefs about labor market parameters}

The surveys also elicit beliefs about key labor market parameters. First, the survey measures beliefs about application success across occupations. Job seekers report the perceived probability of receiving an offer when applying to vacancies in their main target occupation and in alternative occupations. The precise framing of these questions is the following:

\begin{quote}
    ``\textit{If you apply for a position as a "job position", what are the chances that this will lead to a job offer?}``
\end{quote}

In addition to the target occupation the same question is asked for two alternative occupations that are close to the target. These are defined based on occupations of interest declared by job seekers who share the same target occupation.

\medskip

Second, respondents report their perceived probability of receiving at least one job offer within the next three months under alternative levels of weekly search effort (10, 20, or 30 hours). These responses allow us to recover perceived job offer arrival rates and perceived returns to search effort. We define the perceived return to search as the average increase in this subjective probability associated with an additional hour of weekly search effort.

\medskip

Finally, respondents report beliefs about the wage distribution of job offers. For each individual, we compute the median wage among vacancies matching their target occupation and region, and ask respondents to assess the probability that a job offer would exceed this benchmark.

\paragraph{Declared search parameters}

The surveys further collect information on key job search parameters. Search intensity is measured as self-reported weekly time spent searching for a job, using a slider ranging from 0 to 35 hours.\footnote{With an option to indicate in writing if more than 35 hours. This measure is positively correlated with the number of applications submitted on the France Travail platform.} The exact framing of this questions is the following: 

\begin{quote}
    ``\textit{On average, how many hours a week do you spend looking for a job?}``
\end{quote}

Respondents also report their reservation wage, defined as the minimum gross wage they would be willing to accept for their next job: 

\begin{quote}
    ``\textit{What is the minimum gross monthly salary you would accept for the job you are willing to
take?}``.
\end{quote}

\subsection{Reemployment expectations, forecast errors, and prediction}
\label{subsec:forecast_errors}

\paragraph{Measuring reemployment forecast errors}

Subjective reemployment expectations are elicited at multiple horizons. We focus on the three-month horizon, which is both salient for job seekers and closely aligned with the administrative definition of reemployment used by France Travail. Realized reemployment is measured using a three-month indicator that combines two administrative sources (the \textit{Indicateur de retour à l’emploi} and the \textit{Déclarations préalables à l’embauche}, which record firms’ hiring declarations); it closely tracks the institutional definition of reemployment used by the French Public Employment Service, and its construction is detailed in Appendix~\ref{app:admin_data}. Because individuals are surveyed in October 2023, reemployment spells starting in October are included in the three-month window.

\medskip

We define a naive estimate of the individual reemployment forecast error as the difference between the subjective probability of reemployment within three months and the corresponding realized outcome. At the individual level, this object is a noisy estimate of the true individual reemployment forecast error, as realizations are binary while expectations are probabilistic.
\medskip

Our object of interest is instead the gap between subjective expectations and the underlying (objective) reemployment probability. While this objective probability is not directly observed, it can be consistently estimated at the subgroup level: under standard assumptions, the average realized reemployment rate within a group provides an unbiased estimator of the average objective probability for that group.

\medskip

It follows that averaging individual forecast errors within sufficiently homogeneous subgroups recovers average belief distortions. In this sense, positive (negative) average forecast errors correspond to over-optimism (pessimism).

\paragraph{A prediction approach to forecast errors}

We recast the measurement of belief distortions as an ex ante prediction problem. Specifically, we seek to estimate conditional (``local'') forecast errors, defined as the difference between subjective expectations and objective reemployment probabilities, using only information available prior to realization.

\medskip

To do so, we exploit earlier survey waves in which both expectations and realized outcomes are observed. Since realized outcomes provide an unbiased signal of objective probabilities conditional on observables, we can use these data to train a prediction algorithm that maps individual characteristics and stated beliefs into expected forecast errors. The resulting predictions can be interpreted as estimates of conditional belief distortions.

\medskip

The set of predictors includes the subjective three-month reemployment probability and a rich set of administrative variables drawn from PES records. These cover: (i) demographics (age, gender, education, qualifications, marital and parental status); (ii) characteristics of the current unemployment spell (duration, reason for unemployment, and type of contract sought); (iii) past labor market history (number and cumulative duration of previous unemployment spells, as well as past applications through the PES and their outcomes); and (iv) geographic and occupational characteristics (region of residence and target occupation).

\medskip

We then apply the trained algorithm to the October 2023 survey wave, for which realized reemployment outcomes are not yet observed at the time beliefs are elicited. This procedure allows us to characterize belief distortions using only ex ante information, without relying on contemporaneous realizations.

\subsection{Classifying job seekers}
\label{subsec:typology}

We classify job seekers into bias groups based on their predicted individual three-month reemployment bias, using only pre-treatment information. The prediction algorithm is trained on earlier survey waves (2021--2023) and applied to baseline survey responses collected before treatment assignment; it therefore cannot be affected by the intervention. Using the auxiliary dataset described above, we determine a threshold that maximizes the separation in average reemployment bias across groups. Specifically, we choose the cutoff that maximizes the difference in average forecast errors between the resulting subgroups. 

\medskip

When a job seeker completes the baseline survey in the experiment, we predict their individual three-month reemployment bias using the trained algorithm. Individuals are then classified as \textit{pessimists} if their predicted bias falls below the threshold, and as \textit{optimists} otherwise.\footnote{Because the algorithm predicts the gap between subjective reemployment probabilities and realized outcomes, higher predicted values correspond to more optimistic beliefs about reemployment prospects.}$^,$\footnote{This procedure cannot guarantee that every individual classified as a pessimist actually underestimates their true reemployment probability. The true reemployment probability of a given individual is unobservable, and prediction error at the individual level is unavoidable. However, as discussed above, average forecast errors within groups provide consistent estimates of average belief distortions. Our procedure therefore partitions individuals into groups that differ systematically in their expected bias, rather than perfectly identifying individual-level misperceptions.} Appendix Table~\ref{tab:bias_grp_diff} reports average demographics characteristics by bias group. Pessimists are, on average, two years older and more likely to be women (67\% vs.\ 61\%), while educational attainment is similar across groups.

\medskip 

Table~\ref{tab:bias_grp_diff1} reports subjective expectations, search behaviors, and realized reemployment outcomes for the two bias groups, focusing on individuals assigned to the control group only.\footnote{Restricting the sample to untreated individuals ensures that realized outcomes are not mechanically affected by the intervention.} The two groups exhibit substantial differences in forecast errors. Individuals classified as pessimists report an average three-month reemployment probability of 15.8\%, compared with a realized reemployment rate of 18.8\%, implying an average underestimation of 3 percentage points. By contrast, individuals classified as optimists report an average subjective probability of 59.0\%, while their realized reemployment rate is only 31.0\%, corresponding to an average overestimation of 28 percentage points.

\paragraph{Heterogeneous subjective beliefs and behaviors across bias groups.}

Reemployment biases are strongly associated with systematic differences in beliefs about labor market conditions and returns to search. Table~\ref{tab:bias_grp_diff1} shows that pessimistic job seekers hold substantially more negative expectations along nearly all dimensions. They perceive a lower probability that job offers exceed the median wage in their occupation (19\% versus 31\% for optimists) and a substantially lower job offer arrival rate, as reflected in their perceived probability of receiving at least one offer within three months (23.1\% versus 54.2\%). Similar gaps emerge for the perceived probability of receiving multiple offers.

\medskip

Pessimists also report lower perceived hiring probabilities conditional on applying, both in their target occupation (20.5\% versus 30.9\%) and in alternative occupations. More generally, they perceive lower returns to search effort. Using hypothetical scenarios that vary weekly search intensity, we estimate the perceived increase in the probability of receiving at least one offer associated with an additional hour of weekly search effort. This perceived return equals 0.57 percentage points for pessimists, compared with 0.88 percentage points for optimists, suggesting that pessimists perceive additional effort as substantially less productive.

\medskip

These differences in beliefs are mirrored in observed search behavior. Relative to optimists, pessimists exert lower weekly search effort, with a gap of approximately two hours per week, and set slightly lower reservation wages (2{,}055\EUR{} versus 2{,}091\EUR{}). At the same time, they submit more applications, both overall and outside their main target occupation. Taken together, these patterns suggest that pessimistic job seekers are not simply less active in their search. Rather, they combine lower expectations and lower perceived returns to search with broader, but potentially less targeted, application strategies.

\medskip

\begin{table}[htbp]
\centering
\caption{Differences in Beliefs, Behaviors and Outcomes by Bias Group (Control Group)}
\label{tab:bias_grp_diff1}
\scalebox{0.85}{
\begin{threeparttable}
\begin{tabular}{lccc}
\toprule
& \textbf{Pessimists} & \textbf{Optimists} & \textbf{Difference} \\
\midrule

\textit{Panel A. Reemployment expectations and outcomes} \\
Perceived reemployment probability (3 months) & 0.158 & 0.590 & -0.419*** \\
Realized reemployment rate (3 months) & 0.188 & 0.310 & -0.107*** \\
Reemployment forecast error (3 months) & -0.030 & 0.280 & -0.312*** \\
  \\
\textit{Panel B. Beliefs about wages} \\
$\mathbb{P}$(job offer exceeds median wage) & 0.190 & 0.330 & -0.122*** \\
  \\
\textit{Panel C. Beliefs about job offers} \\
$\mathbb{P}$($\geq$ 1 job offer) & 0.231 & 0.542 & -0.277*** \\
$\mathbb{P}$($\geq$ 3 job offer) & 0.173 & 0.438 & -0.235*** \\
  \\
\textit{Panel D. Beliefs about application success} \\
Hiring probability (target occupation) & 0.205 & 0.309 & -0.094*** \\
Hiring probability (alternative occupation 1) & 0.164 & 0.257 & -0.083*** \\
Hiring probability (alternative occupation 2) & 0.149 & 0.231 & -0.072*** \\
  \\
\textit{Panel E. Beliefs about returns to search effort}\\
$\mathbb{P}$($\geq$ 1 offer, 10h/week) & 0.190 & 0.352 & -0.174*** \\
$\mathbb{P}$($\geq$ 1 offer, 20h/week) & 0.245 & 0.444 & -0.192*** \\
$\mathbb{P}$($\geq$ 1 offer, 30h/week) & 0.303 & 0.528 & -0.215*** \\
Return to search effort (pp) & 0.567 & 0.881 & -0.157*** \\
  \\
\textit{Panel F. Search behaviors} \\
Weekly search effort (hours) & 11.39 & 13.81 & -2.159*** \\
Reservation wage (\euro, difference in log) & 2054.65 & 2091.33 & -0.019*** \\
Nb. of applications & 1.39 & 1.22 & 0.215** \\
Nb. of applications outside target occ. & 1.05 & 0.93 & 0.161** \\
\midrule
Number of observations & 3,443 & 7,048 & \\
\bottomrule
\end{tabular}

\begin{tablenotes}[flushleft]
\footnotesize
\item \textit{Notes:} This table reports mean beliefs and outcomes by bias group. The sample is restricted to individuals in the control group. Beliefs are measured prior to the intervention using the baseline survey. All variables correspond to survey responses, except for realized three-month reemployment rates, forecast errors, and the return to search effort. The latter is defined as the average increase in the perceived probability of receiving at least one job offer associated with a one-hour increase in weekly search effort. $\mathbb{P}$ denotes probability.
\end{tablenotes}
\end{threeparttable}
}
\end{table}

\medskip

\paragraph{Heterogeneous psychological characteristics across bias groups.} To further characterize this group, Figure \ref{fig:psychological_profiles} examines a broader set of psychological characteristics using earlier waves of the panel for which such measures are available. Although these variables are not observed in the current experimental sample, pessimistic job seekers are identified using the same prediction algorithm as in the main experiment. Figure \ref{fig:psychological_profiles} reveals substantial differences in psychological traits between pessimistic and non-pessimistic job seekers. Pessimistic job seekers exhibit significantly lower levels of internal locus of control and substantially higher levels of external locus of control. They also score lower on several dimensions of the Big Five personality traits, including emotional stability, conscientiousness, extraversion, and openness to experience. Consistent with these patterns, they also display substantially higher levels of neuroticism.

\medskip

Overall, these patterns suggest that pessimistic job seekers differ not only in their labor market beliefs, but also in broader psychological characteristics associated with self-efficacy, confidence, and the ability to translate intentions into action. While these correlations remain purely descriptive, they are consistent with the idea that pessimistic job seekers may face behavioral constraints that go beyond inaccurate beliefs about labor market prospects. In particular, lower search effort or weak responses to occupational recommendations may partly reflect difficulties in activation, motivation, or perceived control over labor market outcomes.

\medskip

In summary, pessimistic job seekers both anticipate and experience substantially lower reemployment rates than other job seekers. They search less intensively, submit more applications, and target lower-wage job opportunities.

\section{Conceptual Framework}\label{sec:model}

This section sets out a behavioral search model that organizes the empirical findings and supports the discussion of welfare implications for personalized versus uniform recommendations.

\medskip

The model has a two-stage structure. In the first stage, the worker decides which occupations to actively search, namely her consideration set. In the second stage, conditional on this set, she chooses search effort and a reservation wage. The first stage is governed by attention: occupations require cognitive and implementation costs to evaluate, so only a subset is actively searched. The second stage is governed by aspirations: socially determined reference points for wages and effort that create additional incentives for workers whose current choices fall short of what comparable peers target. 

\medskip

The occupational recommendation acts on the first stage by expanding the consideration set. The motivational recommendation acts on the second by raising aspirations. Thus, the two channels are complementary rather than competing.

\subsection{Stage 1: Attention and the consideration set}

The labor market contains occupations $k \in \{1, \ldots, K\}$. The worker holds subjective beliefs $p_k$ about the probability of receiving an acceptable offer when searching in occupation $k$. Searching an occupation requires attention, identifying relevant vacancies, adapting applications, preparing for interviews, so the worker allocates attention weights $a_k \in [0,1]$ subject to a capacity constraint $\overline{a}< K$ satisfying \citep{caplin2015attention}: $\sum_{k=1}^{K} a_k \leq \overline{a}$.
An occupation enters the worker's active consideration set $\mathcal{A}$ only if it receives enough attention to cover the fixed implementation cost:
\begin{equation*}
	\mathcal{A} = \bigl\{k : a_k \geq \underline{a}\bigr\},
%	\label{eq:consideration_set}
\end{equation*}
where $\underline{a} > 0$ is a threshold reflecting these fixed costs. The threshold $\underline{a}$ captures the combined cognitive, implementation, and salience costs of actively searching an  occupation: evaluating fit with one's skills, identifying relevant vacancies, adapting applications, and overcoming the default focus on the target occupation.  The worker searches only in occupations in $\mathcal{A}$, while the rest are effectively invisible regardless of their objective return.

\medskip

The optimal allocation of attention depends on priors: the worker directs attention toward occupations with higher expected returns $p_k$. Two features drive the model's predictions. First, attention is prior-dependent: an occupation with $p_k = 0$ receives no attention, so a recommendation can only activate occupations the worker already weakly values. Second, the occupational recommendation operates by lowering $\underline{a}$ for a small number of nearby occupations rather than by transmitting new information about $p_k$. When $\underline{a}$ drops for occupation $k$, occupations with $a_k$ just below the old threshold cross it and enter $\mathcal{A}$, expanding the consideration set without changing beliefs. This is the ``learning through noticing'' channel of \citet{hanna2014learning}.

\medskip

Our attention-allocation model differs from the framework in  \citet{belot2019providing}, where the consideration set is  determined directly by subjective job-finding probabilities  $p_k$: workers do not search in occupations where they  perceive low returns, and the recommendation expands the  consideration set by providing information that raises  $p_k$ for alternative occupations. In their model, a  change in the consideration set requires a change in  beliefs, i.e., the mechanism is Bayesian updating applied to the occupational margin.

\medskip

Our model separates the two objects. Workers may hold  positive beliefs about alternative occupations ($p_k > 0$) yet fail to include them in $\mathcal{A}$ because the fixed attention cost $\underline{a}$ prevents it. 

\subsection{Stage 2: Search with aspirations}

Conditional on the consideration set $\mathcal{A}$, the worker faces an effective offer arrival rate $\lambda_{\mathcal{A}}(e)$ and an effective wage-offer distribution $F_{\mathcal{A}}(w)$, both of which depend on the set of occupations actively searched.\footnote{The precise aggregation is not essential for the comparative statics that follow. What matters is that expanding $\mathcal{A}$ changes the effective distribution the worker faces.} The worker holds subjective beliefs $(\tilde\lambda_{\mathcal{A}}, \tilde F_{\mathcal{A}})$ that may differ from the true objects. To ease notation, we drop the $\mathcal{A}$ subscript and write $(\tilde\lambda, \tilde F)$ throughout, keeping in mind that these are conditional on the active consideration set.

\medskip

The worker chooses search effort $e \geq 0$ and a reservation wage $w^* \geq 0$. She accepts an offer iff $w \geq w^*$. Writing $\overline{U}$ for the value of continued search, the problem is
\begin{equation}
	\max_{e,\, w^*}\; \tilde\lambda(e)\int_{w^*}^{\bar w}\bigl[U(w;\, r_w) - \overline{U}\bigr]\,d\tilde F(w) - C(e;\, r_e) + \overline{U},
	\label{eq:searchproblem}
\end{equation}
where $U(w; r_w)$ captures aspiration-dependent utility from wages and $C(e; r_e)$ captures aspiration-dependent cost of effort. Both are defined below.

\paragraph{Aspiration-dependent wage utility.} Following the aspirations literature \citep{ray2006aspirations,genicot2020aspirations}, the worker evaluates wages relative to a socially determined aspiration $r_w$. Falling below the aspiration inflicts a utility penalty proportional to the shortfall:
\begin{equation}
	U(w;\, r_w) = v(w) + \kappa_w \bigl(v(w) - v(r_w)\bigr)\,\mathbf{1}\{w < r_w\}, \qquad \kappa_w > 0,
	\label{eq:utility}
\end{equation}
with $v(\cdot)$ smooth, increasing, and concave. In the aspiration-gap region ($w < r_w$), the marginal utility of the wage is $(1 + \kappa_w)\,v'(w)$, steeper than the baseline $v'(w)$: each euro of shortfall below the aspiration hurts $\kappa_w$ times more than a euro gained above it. Above the aspiration, utility reduces to $v(w)$.

\paragraph{Aspiration-dependent effort cost.} Symmetrically, the worker evaluates her effort relative to a socially determined effort aspiration $r_e$. Falling below the effort norm creates a psychological cost including shame or social-comparison disutility which is proportional to the squared gap:
\begin{equation}
	C(e;\, r_e) = c(e) + \frac{\psi}{2}\bigl(\max(r_e - e,\, 0)\bigr)^2, \qquad \psi > 0,
	\label{eq:cost}
\end{equation}
where $c(e)$ is the material cost with $c' > 0$, $c'' > 0$, $c(0) = 0$. The marginal cost of effort is
\begin{equation}
	C'(e;\, r_e) = \begin{cases}
		c'(e) - \psi(r_e - e) & \text{if } e < r_e, \\[2pt]
		c'(e) & \text{if } e \geq r_e.
	\end{cases}
	\label{eq:mc_effort}
\end{equation}
When effort falls below the aspiration, the norm-gap penalty acts as a marginal subsidy: each additional unit of effort reduces the psychological cost of falling short, lowering the effective marginal cost and pushing effort toward $r_e$. The quadratic form ensures that the marginal effect of the aspiration gap increases with the distance to the norm.

\medskip

Both aspiration functions share the same structure: a status quo payoff ($v(w)$ or $c(e)$) augmented by a penalty when status falls short of the aspiration. The key economic insight from \citet{ray2006aspirations} is that aspirations that are attainable but not yet reached create the strongest incentive to invest. This is precisely the regime the motivational intervention targets: it reveals that comparable peers achieve higher wages and exert more effort, placing pessimistic workers in the aspiration-gap region on both margins.

\paragraph{Aspiration formation.} Both aspirations are socially determined based on the worker's perception (denoted by $\mathbb{E}_i$) of what comparable peers target:
\begin{align}
	r_{w,i} &= \mathbb{E}_i\bigl[w^* \mid \text{occupation}_i,\, \text{region}_i,\, \text{comparable workers}\bigr], \label{eq:rw} \\[2pt]
	r_{e,i} &= \mathbb{E}_i\bigl[e^* \mid \text{occupation}_i,\, \text{region}_i,\, \text{comparable workers}\bigr]. \label{eq:re}
\end{align}
These are perceived peer aspirations or norms, what comparable workers aim for, not forecasts of realized labor-market outcomes. This distinction is what allows the motivational message to shift behavior without requiring changes in beliefs about market fundamentals $(\tilde\lambda, \tilde F)$: the message discloses an attainable target anchored on the median peer reservation wage and effort level, which shifts $r_w$ and $r_e$ independently of beliefs about the offer distribution or the returns to search.

\paragraph{Optimality conditions.} The reservation wage solves the 
indifference condition
\begin{equation}
    U(w^*;\, r_w) = \overline{U}.
    \label{eq:reservation}
\end{equation}
Since $\overline{U}$ depends on the entire problem, $w^*$ and $\overline{U}$ are determined jointly. Because the quadratic penalty ensures that $C(e; r_e)$ is continuously differentiable, the first-order condition for effort is
\begin{equation}
    \tilde\lambda'(e^*)\int_{w^*}^{\bar w}\bigl[U(w;\, r_w) 
    - \overline{U}\bigr]\,d\tilde F(w) = C'(e^*;\, r_e),
    \label{eq:foc_effort}
\end{equation}
with $C'$ given by \eqref{eq:mc_effort}.

\medskip

The model delivers two  comparative statics, one for each margin targeted by the motivational  intervention.  Both are 
stated as properties of the optimality conditions at a given 
continuation value~$\overline{U}$. In equilibrium, 
$\overline{U}$ adjusts, but the direct effects derived below determine the sign of the equilibrium response. The first (Lemma~\ref{lem:rw}) governs the wage margin: for a worker in the aspiration-gap region, where the reservation wage lies below the perceived peer wage aspiration ($w^*<r_w$), a rise in that aspiration raises the reservation wage, $\partial w^*/\partial r_w>0$. The second (Lemma~\ref{lem:effort}) governs the effort margin: a rise in the perceived peer effort norm $r_e$ raises search effort for any worker whose effort falls below the norm, while leaving workers already above it unaffected.

\subsection{Two worker types}

Two friction profiles, drawn from the typology of Section~\ref{subsec:typology}, motivate the two treatment arms.

\medskip

A pessimist underestimates returns on two margins: a downward-biased arrival rate, $\tilde\lambda(e) < \lambda(e)$, and a wage-offer distribution $\tilde F^P$ first-order stochastically dominated by $F$. Both distortions depress effort through the FOC \eqref{eq:foc_effort}. Where pessimism co-moves with low perceived peer norms, the aspirations $r_w$ and $r_e$ are compressed, further depressing $w^*$ (see Lemma~\ref{lem:rw}) and $e^*$ (see Lemma~\ref{lem:effort}).

\medskip

An  optimist holds biased beliefs in the opposite direction: $\tilde\lambda^O(e) > \lambda(e)$ and $\tilde F^O$ stochastically dominates $F$. Optimists are also attention-constrained: they hold $\tilde p_k > 0$ for alternative occupations but do not include them in $\mathcal{A}$, because the attention threshold $\underline{a}$ 
prevents these occupations from entering the consideration set.

\subsection{Candidate mechanisms}\label{subsec:mechanisms}

Three candidate mechanisms could drive the behavioral responses. Each leaves distinct empirical signatures.

\paragraph{Mechanism BU: Bayesian Belief Updating.} Under this mechanism \citep[in the spirit of][]{kamenica2011bayesian}, the recommendation is a signal about market fundamentals: occupational recommendations transmit information about $p_k$; motivational recommendations transmit information about $\tilde\lambda$ or $\tilde F$. This predicts: (i)~post-treatment beliefs shift; (ii)~the behavioral response scales with the size of the update; (iii)~effects are largest for workers with the most biased or uncertain priors.

\paragraph{Mechanism AA: Attention Activation.} Following \citet{caplin2015attention,hanna2014learning}, the recommendation lowers $\underline{a}$ for recommended occupations, expanding $\mathcal{A}$ without changing beliefs. This predicts: (i) no detectable shift in beliefs; (ii)~effects concentrated among workers with positive but low-attention priors (spread beliefs, low $h$); (iii)~no effect for workers with narrow beliefs ($h$ high), since attention activation cannot operate where $p_k$ is close to zero.

\paragraph{Mechanism AS: Aspiration Shift.} Following the aspirations and reference-point literature \citep{ray2006aspirations,genicot2020aspirations,koszegi2006model,dellavigna2017reference}, the motivational recommendation raises the perceived peer aspirations $r_w$ and $r_e$. Through Lemma~\ref{lem:rw}, this raises $w^*$ for workers in the aspiration-gap region; through Lemma~\ref{lem:effort}, it lowers the effective marginal cost of effort for workers below the norm. This predicts: (i)~no shift in beliefs about $\tilde\lambda$ or $\tilde F$; (ii)~reservation wages rise only for workers below the disclosed benchmark (threshold heterogeneity); (iii)~effort effects concentrated among low-effort workers, with the largest gains for those furthest below the norm.

\medskip

These candidate mechanisms generate testable predictions for the two treatments that are respectively described in the next section.

\section{Recommendation and motivational treatments and their allocation}\label{sec:treatment}

We consider an ``occupational'' recommendation intervention inspired by the seminal work of \citet{belot2019providing}. However, the previous section has shown that pessimistic job seekers have specific needs. We thus designed a tailored ``motivational'' intervention to match these specific needs.

\subsection{Occupational recommendation}

The first recommendation is inspired by the work of \cite{belot2019providing} and aims to encourage search diversification. The underlying premise is that job seekers may have imperfect information about their reemployment prospects in occupations other than their main target.

\medskip

We construct two distinct proximity rankings of alternative occupations. The first ranking is based on declared occupational interests. When registering with the French PES, job seekers may report several occupations of interest, which provides information on alternative occupations considered relevant by individuals sharing the same target occupation.

\medskip

The second ranking is based on observed application behavior. We rank alternative occupations according to the number of applications they receive from job seekers who mainly search in the given target occupation. This ranking captures which alternative occupations job seekers actively consider worth applying to. We then select the top two occupations from each ranking, ensuring that the four recommended alternatives do not overlap.\footnote{The occupational recommendation procedure differs from that of \cite{belot2019providing}. In their setting, recommended occupations are constructed by combining two sources of information: (i) longitudinal data on occupational transitions, identifying occupations to which workers in a given occupation most frequently move, and (ii) measures of skill similarity based on O*NET. These recommendations are presented alongside descriptive indicators of local labor-market tightness. Empirically, application-based rankings are closely aligned with hiring-based alternatives: among the 150 most common target occupations, more than two thirds have the top hiring-based alternative occupation included among the two top alternatives identified using application-based rankings, based on administrative hiring outcomes observed in an auxiliary dataset.}

\medskip

The recommendation is delivered as a short message displayed at the end of the baseline survey. It suggests a set of alternative occupations tailored to the job seeker’s profile and current target occupation. For example, job seekers are told:

\begin{quote}
``\textit{You could broaden your job search to other sectors of activity. Our studies show that you have hiring opportunities in several other sectors. If you are searching for a position as ``lib\_metier'', the following positions may also be of interest to you: ``lib\_prop\_metier1'', ``lib\_prop\_metier2''. You could also apply to positions such as ``lib\_prop\_metier3'' and ``lib\_prop\_metier4''. You could increase your chances of finding a job by diversifying your job search and your applications.}''
\end{quote}

The message emphasizes that diversifying applications can increase job-finding prospects and provides concrete alternatives derived from administrative data.

\medskip

In the context of the model of Section~\ref{sec:model}, the 
occupational recommendation targets the consideration margin: it 
lowers the attention threshold~$\underline{a}$ for nearby 
occupations, expanding~$\mathcal{A}$ without changing 
beliefs~$\tilde p_k$.  The behavioral response here 
requires no belief revision. This distinction generates sharp 
predictions.

\begin{enumerate}
    \item[P.1] \textbf{No belief updating.} The intervention does not shift perceived success probabilities $p_1$, $p_2$, $p_3$, because it expands $\mathcal{A}$ by reducing attention 
    costs rather than by revising posteriors. Inconsistent with Mechanism~BU.
    
    \item[P.2] \textbf{Heterogeneity by initial belief 
    narrowness.} Under belief updating, effects should be largest for workers with the most biased priors (low~$p_k$, high~$h$), who have the most room to update. Under attention activation, effects should concentrate among workers who already hold favorable priors 
    (high~$p_k$, low~$h$) but do not act on them, because attention activation can only operate where priors are positive. The two mechanisms make opposite predictions about which workers respond. 
    
    \item[P.3] \textbf{No wording effect on beliefs.} Even wordings explicitly designed to correct beliefs, such as the raise-awareness framing, which invites job seekers to  reconsider their priors, should generate no stronger belief revision than direct advice. Consistent with~AA, inconsistent with~BU.
\end{enumerate}

\subsection{Motivational recommendation}\label{sec:motivational}

The motivational recommendation is designed to address the specific needs of pessimistic job seekers. This treatment directly targets the search parameters of individuals classified in the pessimistic group. In addition to holding depressed beliefs about key labor-market parameters, particularly wage distributions and the returns to search effort in terms of job-offer arrival rates, these individuals tend to exert lower search effort and to set lower reservation wages.

\medskip

A primary objective of the motivational recommendation is therefore to restore job seekers’ confidence in the effectiveness of search effort. The first message explicitly encourages job seekers not to underestimate the productivity of search activity. In particular, we recommend increasing search effort and provide information on the median level of weekly search effort exerted by other job seekers. This reference level is estimated using the auxiliary dataset collected between October 2021 and July 2023. In the treatment, this is conveyed through statements such as:

\begin{quote}
``\textit{Searching for job offers is an effective way to find a job. Our studies show that searching for job offers greatly improves the chances of finding a job. One out of two job seekers spends at least 10 hours per week searching for job offers.}'' 
\end{quote}

We expect that providing information on typical search intensity may induce individuals who currently search less to increase their own effort.

\medskip

The recommendation also emphasizes that higher search effort may enable job seekers to target better job opportunities. Although pessimistic job seekers exert less search effort, they apply to more vacancies on average, suggesting lower selectivity in their application behavior. This pattern is consistent with their lower reservation wages. Accordingly, the motivational message encourages job seekers to reallocate search effort toward higher-quality vacancies. This is expressed in the treatment through statements such as:

\begin{quote}
``\textit{Increasing your search effort allows you to focus on job offers with higher wages.}'' 
\end{quote}

and

\begin{quote}
``\textit{Given your profile and the job offers available on the market, you could also target and apply for better-paid jobs than those you are currently applying to. For a position as ``lib\_metier'', you could apply to jobs offering a monthly gross wage of ``wage\_target'' euros.}'' 
\end{quote}

This reference wage is computed as the median reservation wage set at PES registration among individuals sharing the same target occupation and region. This information is provided only to job seekers whose own reservation wage falls below this median level, so as to avoid potential demotivation effects that could arise from delivering a negative signal about labor market prospects.

\medskip

We hypothesize that framing increased effort as a way to access better opportunities may further enhance motivation and improve search efficiency.

\medskip
Overall, the motivational recommendation aims to stimulate higher  search intensity by restoring perceived returns to search effort and by raising aspirations regarding the quality of attainable job offers. We do not consider this type of recommendation to be well suited to individuals classified as optimistic. Providing motivational content to job seekers who already hold overly optimistic beliefs would create a risk of iatrogenic treatment, as it may amplify existing biases and potentially worsen reemployment outcomes.

\medskip

In the framework of Section~\ref{sec:model}, the motivational recommendation operates by raising two aspirations simultaneously:  the peer effort norm~$r_e$ (through the disclosed median search 
effort) and the peer wage benchmark~$r_w$ (through the disclosed median reservation wage). Both shifts place pessimistic workers in the aspiration-gap region, activating the mechanisms of 
Lemmas~\ref{lem:rw} and~\ref{lem:effort}. Crucially, neither shift requires the worker to revise her beliefs about market fundamentals. The message changes what she perceives as 
attainable, not what she believes about the wage distribution or the returns to search. This yields four predictions that collectively discriminate the aspiration channel from belief updating.

\begin{enumerate}
    \item[P.4] \textbf{Threshold heterogeneity in reservation wages.} By Lemma~\ref{lem:rw}, the aspiration penalty applies only in the gap region. The treatment should therefore raise~$w^*$ only for workers whose pre-intervention reservation wage lies below the disclosed peer benchmark, and have no effect on workers already above it. Under belief updating, the effect should not exhibit this sharp threshold. Consistent with Mechanism~AS.
    
    \item[P.5] \textbf{No shift in beliefs about $\tilde F$.} 
    The treatment raises~$w^*$ without shifting subjective beliefs about the wage-offer distribution. This isolates the aspiration channel, which operates through~$r_w$, not 
    through~$\tilde F$, from a belief-correction channel, which would require $\tilde F$ to move.  Inconsistent with BU.
    
    \item[P.6] \textbf{Effort increase concentrated among low-effort workers.} By Lemma~\ref{lem:effort}, the marginal cost reduction is $\psi(r_e - e^*)$, which is largest for workers furthest below the norm. The effort effect should therefore be largest for low-effort pessimists and absent for workers already at or above the disclosed benchmark of 10 hours per week.
    
    \item[P.7] \textbf{Upward reallocation of applications.} Higher~$w^*$ makes the worker more selective; higher effort expands the set of offers she evaluates. The net effect on total applications is ambiguous, but the composition should shift toward higher-wage vacancies.
\end{enumerate}

Finally, the last prediction concerns employment outcomes, where the effect of such a motivational treatment is ambiguous.

\begin{enumerate}
	\item[P.8] \textbf{Ambiguous employment effects.} Correcting the effort distortion speeds reemployment; raising the reservation wage delays it by inducing rejection of offers that would previously have been accepted. The net effect depends on which margin dominates for each worker, which in turn depends on how much of her underperformance reflects depressed aspirations versus structural market frictions.
\end{enumerate}

\subsection{Treatment wordings variations}\label{sec:wording_treatment}

Each recommendation was implemented using three alternative framings: a direct-advice framing, a peer-behavior framing, and an awareness framing emphasizing potential behavioral biases during job search. These alternative wordings were designed to explore whether recommendations operate primarily through direct advice, social comparison, or increased awareness of potential search biases. Results reported in the main text pool the three wording variants. Detailed wording examples and heterogeneous effects by framing are reported in Appendix \ref{sec:app-wordings}.

\section{Experimental Design and Data}
\label{sec:experiment}

This section describes the experimental procedure implemented in partnership with France Travail. We first detail the enrollment process and the baseline survey used to register job seekers into the experiment. We then describe the treatment arms and the randomization scheme. Finally, we discuss take-up, compliance, and the follow-up survey.

\medskip

Table~\ref{tab:design} summarizes the different steps of the experimental design, from the initial survey invitation to enrollment and treatment assignment. Figure~\ref{fig:rct_schema} provides a graphical overview of the experimental timeline, including enrollment, classification into belief types, randomization, treatment delivery, and follow-up.

\subsection{Enrollment and baseline survey}

The experiment was conducted in October 2023 in collaboration with France Travail. A random sample of 375{,}000 registered job seekers received an email invitation to participate in an online survey on their job search behavior and perceptions of the labor market. The survey corresponds to the baseline survey described above and builds on the survey infrastructure implemented in previous waves.

\medskip

Participation in the survey was voluntary. Respondents were asked whether they were actively searching for a job at the time of the survey and to report their subjective reemployment expectations. Individuals who reported not actively searching for a job or did not provide reemployment expectations are excluded from the experimental sample. A total of 52{,}465 job seekers satisfy these conditions and constitute the experimental sample.

\medskip

Upon completion of the baseline survey, job seekers were informed that they might receive personalized information or recommendations related to their job search. The survey responses collected at this stage are used to construct the bias groups described in Section~\ref{subsec:typology}. Among the 52{,}465 enrolled job seekers, 17{,}073 are classified as pessimistic and 35{,}392 as optimistic. 

\subsection{Treatment allocation}

The allocation of job seekers across treatment arms proceeds in two steps. The first step occurs at a very early stage. The initial sample is assigned to a treatment group and a control group, with respective proportions of 80\% and 20\%. This assignment remains hidden from job seekers.

\medskip

Job seekers are then invited to complete the survey and, conditional on satisfying the eligibility criteria described earlier, are enrolled in the experiment. Based on their survey responses, individuals are classified into two bias groups, pessimistic and optimistic job seekers. This classification is observed for individuals in both the treatment and control groups.

\medskip

The second step of treatment assignment is conditional on the predicted bias group. Job seekers in the treatment group who are classified as pessimists are randomly assigned to one of six experimental arms. These include three wording variations of the occupational recommendation and three wording variations of the motivational recommendation. Within the group of pessimistic individuals, one fifth are randomly assigned to the control group, while the remaining four fifths are evenly distributed across the six treatment arms.

\medskip

We follow a similar procedure for optimistic job seekers. However, because the motivational recommendation is not intended for individuals classified as optimistic, job seekers in this group who are initially assigned to the treatment group in the first step are randomly allocated to one of the three wording variations of the occupational recommendation. Overall, among optimistic job seekers, one fifth are allocated to the control group, and the remaining four fifths are evenly split across the three occupational treatment arms.

\medskip

Table~\ref{tab:balance} presents balance tests, which indicate that, conditional on bias group, individual characteristics are comparable across treatment arms and the control group (each arm is compared to the corresponding control group; test details are given in the table notes).

\medskip

Overall, we find very few statistically significant imbalances. The main exception concerns the perceived reemployment probability, which slightly differs in two treatment arms (the occupational recommendations with the raise awareness wording). For this reason, we systematically control for baseline subjective beliefs in our main estimations of treatment effects.

\subsection{Take-up of treatments}

Upon completing the survey, job seekers assigned to the treatment group are offered the opportunity to receive a personalized recommendation based on their responses. For those who consent, the informational content displayed corresponds to their assigned treatment arm. Table~\ref{tab:design} reports the number and share of job seekers who consented, separately for pessimistic and optimistic individuals. All reported estimates in this paper are intent-to-treat (ITT).\footnote{Take-up  rates are approximately 63\% for pessimists and 66\%  for optimists (Table~\ref{tab:design}). Thus, under standard  exclusion restrictions, implied treatment-on-the-treated effects are roughly 50\% larger than the reported ITT.} 
\medskip

The message sent to individuals in the control group contains only a link to the PES platform and a brief note indicating that they may contact their caseworker if needed. The same message is displayed to individuals in the treatment group who did not consent to receive a personalized recommendation.

\subsection{Outcome data}

 \subsubsection{Administrative outcome data}

The two main outcomes used to assess the impact of our recommendations, namely job seekers’ applications and reemployment outcomes, are drawn from administrative databases provided by the French PES.

\medskip

The first administrative outcome of interest is the set of applications submitted by job seekers on the PES platform. We focus on applications made in the two months following the intervention.\footnote{Since job seekers completed the baseline survey until October 16th, we consider applications made between October 2nd and December 15th.} Although job seekers may also rely on other job boards, the France Travail platform plays a central role in the matching process and accounts for a substantial share of vacancies in France. As documented by \cite{le2019gender}, it represents approximately 60\% of all job vacancies nationwide, and a large fraction of job seekers use it as a primary search and application channel.

\medskip 

These data are particularly rich and allow us to observe detailed characteristics of each application. We distinguish applications by their source. Standard applications correspond to applications submitted by job seekers to vacancies posted on the PES platform. Suggested applications are also directed toward posted vacancies but originate from recommendations made by caseworkers. Finally, spontaneous applications are not linked to posted vacancies but are directed toward specific firms; they are initiated by job seekers through a dedicated platform that enables direct contact with employers.\footnote{This corresponds to applications made through a dedicated tool available at France Travail called \textit{La Bonne Boîte} (see \citet{behaghel2024potential}).}

\medskip 

For applications to posted vacancies, we observe several key characteristics of the job ads, including the offered wage, the type of contract (permanent or fixed-term), and whether the position is full-time or part-time. Importantly, we also observe the occupation associated with each vacancy. This allows us to classify applications according to whether they target (i) the job seeker’s preferred occupation, (ii) one of the four occupations suggested in the occupational recommendation, or (iii) any other occupation.

\medskip

To evaluate the impact of the intervention on reemployment outcomes, we rely on a second administrative data source that records all hirings starting in October 2023.\footnote{These data correspond to the \textit{Déclarations sociales nominatives} (DSN), which are mandatory monthly wage reports submitted by all employers. These data are distinct from those used to construct the 3-month reemployment bias measure; see Appendix \ref{app:admin_data} for further details.} Unlike the monthly employment indicators typically used by the PES, this dataset captures the universe of hirings, including temporary and short-term contracts, and provides detailed information on post-hire wages. It is therefore particularly well suited to assess the effects of the intervention on reemployment outcomes.

\medskip

We focus on the first job held by each job seeker after the intervention date.

\subsubsection{Follow-up survey and measurement of beliefs and search parameters}

To measure post-intervention search behaviors and subjective beliefs, job seekers were invited by email one month after the baseline survey to complete a follow-up questionnaire. The survey replicates key baseline questions and allows us to construct outcomes related to search effort, reservation wages, and perceived labor market prospects. Invitations were conditioned on engagement in the baseline survey. In the treatment group, they were sent to individuals who accepted the recommendation, while in the control group they were sent to individuals who completed the final question of the survey. Overall, 64\% of the 52,465 individuals in the experimental sample were invited to participate (see Table \ref{tab:design}). Among those invited, 44\% participated in the follow-up, and 38\% ultimately provided usable observations.

\medskip 

Importantly, the share of individuals with usable follow-up data is very similar across groups, ranging from 23\% to 24\% of the experimental sample, with a maximum difference of 1.7 percentage points across columns. However, despite these similar aggregate rates, the underlying selection processes differ across treatment arms, reflecting differences in both invitation and participation mechanisms. This leaves open the possibility of differential selection into the estimation sample, which we address using the reweighting and machine learning methods described below.

\medskip 

Because follow-up participation reflects engagement with the survey, the resulting sample differs somewhat from the full experimental population (see Table \ref{tab:balance_follow} for balance tests). While observable baseline characteristics remain broadly similar across treatment arms among respondents, we implement several estimation procedures to address potential selection concerns. These procedures are described in detail in Appendix \ref{sec:selectedsurvey}. First, we implement inverse probability weighting (IPW) based on predicted response probabilities estimated separately for the treatment and control groups \citep{wooldridge2007inverse}. Second, we estimate treatment effects using double/debiased machine learning (DML), which flexibly adjusts for observable determinants of survey response \citep{chernozhukov2018double}.

\medskip 

We consider two approaches to the propensity score (see Appendix \ref{sec:nuisance} for details). First, we estimate the probability of being treated among respondents. Second, we use an alternative specification that exploits the randomized assignment of treatment and reconstructs the propensity score from treatment-specific participation probabilities. This latter approach explicitly leverages the experimental design and isolates the role of differential follow-up.

\medskip 

 To check the validity of these estimation methods, we use administrative outcomes that are observed for the full experimental sample: we verify that applying our estimation procedures on the subsample that replicates the same selection yields results very close to those obtained using the full data. This provides further reassurance that our findings are not driven by selective follow-up.

%%%%%%%%%%%%%%%%%%%%%%%%%%%%%%%%%%%%%%%%%%%%%%%%%%%%%%%%%%%%%%%%%%%%%%%%%%%%%%%%%%%%%%%%%%%%%%%%%%%%%%%%%
\section{Results: Occupational Arm}\label{sec:occ}

Job seekers differ substantially in their subjective labor market expectations, and these differences may shape their responses to occupational recommendations. This section compares the effects of providing recommendations for alternative occupations across the two bias groups. We first present the empirical specification and then examine how the intervention affects applications across search channels and occupations.

\medskip

Our main finding is that occupational recommendations substantially increase applications to posted vacancies among optimistic job seekers, while generating much weaker responses \emph{along the targeted margins} among pessimistic job seekers, who instead respond by increasing spontaneous, direct-to-firm applications. The contrast is especially pronounced along the primary search margin, applications to posted vacancies, and for applications directed toward the occupations explicitly recommended by the intervention. As we show in Section~\ref{sec:employment}, this weaker response along the recommended margin does not, however, imply that the occupational arm is inconsequential for pessimists: it nonetheless produces a positive effect on their reemployment.

\subsection{Econometric specification}

We estimate the effect of the intervention using the following specification:

\begin{equation}
\label{eq:mod1}
    Y_i = \alpha P_i +\beta_P P_i T_i + \beta_O (1-P_i) T_i +  X_i'\gamma + \varepsilon_i,
\end{equation}

where $Y_i$ denotes the outcome of interest, such as the number of applications submitted by job seeker $i$ after the intervention. $P_i$ is an indicator equal to one if individual $i$ is classified as pessimistic, and $T_i$ is an indicator equal to one if the job seeker received an occupational recommendation. 

\medskip 
The coefficients $\beta_P$ and $\beta_O$ capture the treatment effects for the two groups.\footnote{In the main specification, we pool all occupational-treatment wordings. Robustness checks that estimate effects separately by wording are reported in the appendix.} Although randomization is sufficient for identification, we include controls for baseline perceived probability of reemployment within three months and for the number of applications submitted on the PES platform in the month preceding the intervention in order to improve precision. These variables capture baseline beliefs and pre-intervention search intensity.

\medskip 

Our main analysis focuses on \textit{platform users}, defined as job seekers who submitted at least one application on the PES platform in the month prior to the intervention. To limit the influence of extreme values, we winsorize the number of applications at the 98th percentile of strictly positive observations.

\subsection{Main findings}

Table \ref{tab:ate_app} reports the estimated treatment effects. The left panel reports results by application channel and the right panel decompose applications on posted vacancies in three groups: the job seeker's main target occupation, the occupations suggested by the intervention, and all other occupations. Figure \ref{fig:ATE_main} offers a visual of these findings. 

\medskip

Overall, occupational recommendations increase application activity for optimistic job seekers, but have much weaker effects for pessimistic job seekers. The difference is most pronounced for applications to posted vacancies, which account for about 90\% of all applications on the platform and therefore constitute the main application margin.\footnote{The same qualitative pattern also appears in the full sample, when we do not restrict the sample to platform users, although the estimated effects are attenuated (Table \ref{tab:ate_full_sample}). This pattern is also robust when the treatment is split by wording (Table \ref{tab:ate_by_wording}). Across all wordings, applications to posted vacancies increase for optimistic individuals but not for pessimists.}

\medskip

For optimistic job seekers, the intervention raises applications to posted vacancies by 0.457, corresponding to a 11\% increase relative to the control group. It also increases spontaneous applications by 0.104, or 43.7\%. In addition, the increase in applications to posted vacancies is disproportionately concentrated on suggested occupations. The number of applications to suggested occupations increases by 0.174, corresponding to a 20.8\% increase relative to the control group.

\medskip 

By contrast, for pessimistic job seekers, the intervention does not increase applications to posted vacancies. In addition, the estimated effects are close to zero across all occupation categories. There is no impact as well on caseworker-suggested vacancies. The only significant response is for spontaneous applications, which rise by 0.242, a substantial increase of over 90\%.

\begin{table}[!htbp]
    \caption{Occupational treatment effect on online applications by bias group - platform users only}
    \label{tab:ate_app}
\begin{adjustbox}{width = \textwidth}
\begin{threeparttable}
\begin{tabular}{lllllll}
            \hline
            \hline
            & \multicolumn{6}{c}{\textit{Number of applications in the two months following intervention}} \\
            & \multicolumn{3}{c}{Across streams} & \multicolumn{3}{c}{Across occupations} \\  
            \cmidrule(lr){2-4}\cmidrule(lr){5-7} \\
            & On a posted vacancy & Spontaneous & Caseworker suggestion  & Preferred & Suggested & Others \\
            \hline
            & (1) & (2) & (3) & (4) & (5) & (6) \\
            \textbf{Pessimists} & & & & & & \\
            Occupational treatment ($\widehat{\beta}_P$) & 0.154 & 0.242$^{***}$ & -0.063 & 0.027 & -0.012 & 0.138\\   
                                      & (0.300) & (0.082) & (0.077) & (0.115) & (0.099)                & (0.222)\\
            \cline{2-7}
            Control mean  ($\widehat{\alpha}_P$) & 4.95  & 0.262 & 0.538   & 1.19     & 1.01         & 2.76\\   
            & & & & & & \\
            \hline
            \textbf{Optimists} & & & & & & \\
            Occupational treatment ($\widehat{\beta}_O$) & 0.457$^{**}$  & 0.104$^{**}$  & 0.083$^{**}$       & 0.098             & 0.174$^{***}$ & 0.185\\   
                                      & (0.201) & (0.048) & (0.037) & (0.076) & (0.062)                & (0.139)\\ 
            \cline{2-7} 
            Control mean ($\widehat{\alpha}_O$) & 4.31 & 0.238  & 0.455& 1.01    & 0.829  & 2.47\\
            & & & & & & \\
            \hline
            $\beta_P = \beta_O$ pvalue  & 0.40 & 0.16 & 0.09 & 0.59 & 0.11 & 0.86 \\
            $\alpha_P = \alpha_O$ pvalue & 0.03 & 0.71 & 0.28 & 0.11 & 0.07 & 0.20\\
            Number Obs. & 8,212 & 8,212 & 8,212 & 8,212 & 8,212 & 8,212 \\  
            \hline
        \end{tabular}
\begin{tablenotes}[flushleft]
    \small
    \item  \textit{Notes:} This table reports estimated average treatment effects of the occupational recommendation on the number of applications submitted in the two months following the intervention, separately by bias group. Wording variations are pooled together. Several application outcomes are considered: (1) the number of applications to posted vacancies, (2) the number of spontaneous applications (i.e., applications not linked to a posted vacancy), and (3) the number of applications submitted in response to vacancies suggested by caseworkers. Applications to posted vacancies are further disaggregated by the occupation of the vacancy: (4) preferred (target) occupation, (5) suggested occupations, and (6) other occupations. The sample is restricted to job seekers who submitted at least one application on the PES platform in the month prior to the intervention. All regressions control for perceived three-month reemployment probability and the number of applications made in the month before the intervention. Application counts are winsorized at the 98th percentile of strictly positive values. P-values of Fisher tests for equal control mean and equal effect of the treatment are also reported. Significance levels: $^{*}$p$<$0.1; $^{**}$p$<$0.05; $^{***}$p$<$0.01.
\end{tablenotes}
\end{threeparttable}
\end{adjustbox}
\end{table}

\subsection{Interpretation}

Occupational recommendations affect optimistic and pessimistic job seekers through different adjustment margins. This contrast does not stem from large differences in baseline diversification: in the control group, the two groups display broadly similar distributions of applications across occupation categories, although pessimists submit slightly more applications overall.

\medskip

Pessimistic job seekers already operate relatively extensive search strategies, leaving limited scope for further increases in applications to posted vacancies. The occupational recommendation does trigger a behavioral response among pessimists, but this response takes the form of a strong increase in spontaneous applications directed to firms rather than a reallocation toward the recommended occupations. Because the occupational content of spontaneous applications is unobserved, we cannot determine whether these additional applications target the suggested occupations.

\medskip

Optimistic job seekers have greater scope for expanding applications to posted vacancies, and the additional applications are disproportionately directed toward the suggested occupations: applications to suggested occupations increase by 21\%, compared with 10\% for preferred and 8\% for other occupations. The recommendation redirects part of their search effort toward the highlighted occupations.

\subsection{Beliefs and mechanisms}

A natural interpretation of occupational recommendations is that they operate through belief updating. Job seekers may underestimate their chances of success outside their target occupation, and recommendations may increase applications by correcting these misperceptions.

\medskip 

Our baseline survey data reveal substantial heterogeneity in perceived opportunities across occupations. Some job seekers perceive relatively favorable opportunities outside their target occupation, whereas others perceive sharply lower chances of success outside their preferred occupation. However, despite this heterogeneity, our results provide little support for a standard belief-updating mechanism. First, the intervention does not significantly modify subjective beliefs regarding success probabilities across occupations (see prediction P.1). Second, treatment effects are concentrated among individuals who already perceive relatively favorable opportunities outside their target occupation before the intervention (see prediction P.2). The recommendation therefore appears to activate search behavior among job seekers who are already aware of outside opportunities, rather than to generate new beliefs about these opportunities.

\medskip 

To summarize how individuals allocate perceived chances across occupations, we construct the  following individual measure of belief narrowness:
\[
h = \frac{p_t}{p_t + p_1 + p_2},
\]
where $p_t$ (respectively $p_1$ and $p_2$) denotes the perceived probability of success in the target occupation (respectively in suggested occupations 1 and 2). Higher values of $h$ indicate that perceived opportunities are more strongly concentrated on the target occupation, while lower values correspond to individuals with relatively high perceived outside options.

\medskip

Figure \ref{fig:beliefs_relocc} reports the empirical distribution of this index.\footnote{Appendix Figure \ref{fig:beliefs_raw} presents the distributions of $p_t$, $p_1$, and $p_2$.} Lower values correspond to individuals with relatively balanced perceived opportunities across occupations, whereas higher values indicate stronger concentration of perceived opportunities on the target occupation.\footnote{These perceptions are strongly related to pre-intervention search behavior: individuals who perceive more favorable outside opportunities already diversify their applications more across occupations (see Appendix Table \ref{tab:desc_model})}

\medskip

We classify individuals according to whether they perceive relatively high chances of success outside their target occupation. Specifically, we define
\[
S = \mathbf{1}\{h < \text{median}(h)\},
\]
so that $S=1$ identifies individuals with \textit{high perceived outside options}, that is, individuals whose perceived probabilities are relatively more balanced across occupations and less concentrated on the target occupation.

\paragraph{Heterogeneity by initial perceptions.}

We first examine whether treatment effects depend on initial perceptions of opportunities in alternative occupations. We estimate the following specification:
\begin{equation*}
    Y_i = \alpha + \gamma S_i + \beta_T T_i + \beta_S S_iT_i + X_i'\lambda + \varepsilon_i,
\end{equation*}
where $T_i$ is the occupational-treatment indicator, $S_i$ the high perceived outside option dummy and $X_i$ includes the same controls as in the main specification.

\medskip 

Table \ref{tab:ate_by_initial_beliefs} reports the results. For optimistic job seekers, treatment effects are concentrated among individuals with high perceived outside options. In this group, applications to suggested occupations increase by about 0.33, which corresponds to roughly 45\% of the control mean, and applications to other occupations also rise significantly. By contrast, there is no significant effect for individuals whose perceived chances are more concentrated on the target occupation.

\medskip 

For pessimistic job seekers, we find no significant treatment effects regardless of initial perceptions. The absence of response is therefore not explained by differences in perceived opportunities across occupations.

\begin{table}[htbp]
\centering
\caption{Occupational Treatment Effects on Applications: Heterogeneity by Initial Perceived Outside Options}
\label{tab:ate_by_initial_beliefs}

\scalebox{0.8}{
\begin{tabular}{lcccc}
\toprule
& \multicolumn{4}{c}{Applications in the Two Months Following the Intervention} \\
\cmidrule(lr){2-5}
& To a Vacancy & Target Occ. & Suggested Occ. & Other Occ. \\
\midrule

\multicolumn{5}{l}{\textbf{Optimists}} \\

High perceived outside options
& \begin{tabular}[c]{@{}c@{}}-0.357\\(0.505)\end{tabular}
& \begin{tabular}[c]{@{}c@{}}-0.171\\(0.185)\end{tabular}
& \begin{tabular}[c]{@{}c@{}}0.188\\(0.138)\end{tabular}
& \begin{tabular}[c]{@{}c@{}}-0.374\\(0.333)\end{tabular}
\\[1.5ex]

Occ. treatment
& \begin{tabular}[c]{@{}c@{}}-0.249\\(0.413)\end{tabular}
& \begin{tabular}[c]{@{}c@{}}-0.060\\(0.147)\end{tabular}
& \begin{tabular}[c]{@{}c@{}}0.036\\(0.106)\end{tabular}
& \begin{tabular}[c]{@{}c@{}}-0.224\\(0.285)\end{tabular}
\\[1.5ex]

High perceived outside options $\times$ Occ. treatment
& \begin{tabular}[c]{@{}c@{}}1.240**\\(0.567)\end{tabular}
& \begin{tabular}[c]{@{}c@{}}0.136\\(0.206)\end{tabular}
& \begin{tabular}[c]{@{}c@{}}0.331**\\(0.159)\end{tabular}
& \begin{tabular}[c]{@{}c@{}}0.768**\\(0.371)\end{tabular}
\\

\midrule

Control mean
& 4.630 & 1.190 & 0.749 & 2.690 \\

Observations
& 4,522 & 4,522 & 4,522 & 4,522 \\

\midrule

\multicolumn{5}{l}{\textbf{Pessimists}} \\

High perceived outside options
& \begin{tabular}[c]{@{}c@{}}1.020\\(0.704)\end{tabular}
& \begin{tabular}[c]{@{}c@{}}-0.044\\(0.213)\end{tabular}
& \begin{tabular}[c]{@{}c@{}}0.463**\\(0.215)\end{tabular}
& \begin{tabular}[c]{@{}c@{}}0.603\\(0.493)\end{tabular}
\\[1.5ex]

Occ. treatment
& \begin{tabular}[c]{@{}c@{}}0.178\\(0.549)\end{tabular}
& \begin{tabular}[c]{@{}c@{}}0.133\\(0.203)\end{tabular}
& \begin{tabular}[c]{@{}c@{}}0.121\\(0.158)\end{tabular}
& \begin{tabular}[c]{@{}c@{}}-0.076\\(0.387)\end{tabular}
\\[1.5ex]

High perceived outside options $\times$ Occ. treatment
& \begin{tabular}[c]{@{}c@{}}-0.276\\(0.863)\end{tabular}
& \begin{tabular}[c]{@{}c@{}}-0.037\\(0.286)\end{tabular}
& \begin{tabular}[c]{@{}c@{}}-0.271\\(0.259)\end{tabular}
& \begin{tabular}[c]{@{}c@{}}0.032\\(0.596)\end{tabular}
\\

\midrule

Control mean
& 4.390 & 1.130 & 0.780 & 2.470 \\

Observations
& 1,303 & 1,303 & 1,303 & 1,303 \\

\bottomrule
\end{tabular}
}

\vspace{0.75em}

\parbox{\textwidth}{\footnotesize
\textit{Notes:} ``Occ.'' denotes occupational. This table reports treatment effects of the occupational recommendation on applications submitted during the two months following the intervention, allowing for heterogeneity by initial perceived outside options. The first column reports effects on the total number of applications submitted to posted vacancies. The remaining columns decompose applications by occupation type: target occupation, suggested occupation, and all other occupations. All specifications control for perceived 3-month reemployment probability and the number of applications submitted in the month prior to the intervention. Treatment wording variants are pooled. The sample is restricted to individuals who submitted at least one application on the PES platform before the intervention. Robust standard errors are reported in parentheses. $^{*}p<0.10$, $^{**}p<0.05$, and $^{***}p<0.01$.
}

\end{table}

\paragraph{Does the intervention affect beliefs?}

We then examine whether occupational recommendations directly shift subjective beliefs. Using follow-up survey data, we reconstruct the post-intervention belief narrowness index $h_{post}$ and estimate the effect of the intervention on this measure, controlling for baseline beliefs and allowing for heterogeneity by initial perceived outside options.

\medskip 

Table \ref{tab:ate_by_candidature_avant_full} shows little evidence that the intervention modifies subjective beliefs regarding the relative concentration of perceived opportunities across occupations as measured by the post-intervention belief narrowness index, neither for pessimists or optimists.\footnote{Table \ref{tab:ate_by_candidature_avant_full} details impacts on the perceived probabilities of success in the target occupation $p_{t,post}$ and the two suggested occupations $p_{1,post}$ and $p_{2,post}$. We do not detect systematic changes in perceived success probabilities in either the target or suggested occupations, although a few isolated coefficients are statistically significant, their magnitude remains small relative to baseline belief levels and no coherent pattern of belief updating emerges across groups or occupations.}

This result is important because it speaks directly to the standard theory of change underlying such recommendations. Under a belief-updating mechanism, one would expect individuals to revise their perceptions of alternative occupations after receiving the recommendation, especially those for whom the recommendation reveals attractive opportunities outside their target occupation. We do not observe such a pattern.

Overall, the occupational recommendation does not seem to shift beliefs, and its behavioral effects concentrate among workers who already perceive favorable outside options. This pattern points to attention activation rather than belief updating, a reading we formalize and test against additional evidence in Section~\ref{subsec:interpretation}.

\section{Results: Motivational Arm}\label{sec:motivational_impact}

This section examines the motivational recommendation, which targets the specific constraints of pessimistic job seekers: low search effort and low reservation wages (Table~\ref{tab:bias_grp_diff1}). The intervention encourages higher search intensity and provides reference wages for better-paying jobs (Section~\ref{sec:motivational}). We focus on effort, reservation wages, and underlying beliefs measured three weeks after the intervention.

\subsection{Econometric specification}

As in the analysis of occupational recommendations, identification relies on the random assignment of treatments. We estimate specifications that either pool all wording treatments or allow for heterogeneous effects across wording variants.

\medskip 

To estimate average treatment effects pooling across all wording treatments, we use the following specification:
\begin{equation}
\label{eq:mod2}
Y_i = \alpha + T^o_i \beta^o + T^m_i \beta^m + X_i'\lambda + \varepsilon_i,
\end{equation}
where $Y_i$ denotes either search behavior or subjective beliefs for individual $i$, measured in the follow-up survey. $T^o_i$ and $T^m_i$ are indicators for assignment to the occupational and motivational treatments, respectively, such that $\beta^o$ and $\beta^m$ capture the corresponding average treatment effects. The vector $X_i$ includes baseline measures of search behavior, related subjective beliefs, and target-occupation fixed effects, all of which are strong predictors of the outcomes.

\medskip 

We focus on search parameters reported in the follow-up survey to assess how each recommendation affects job search behavior. Because these outcomes are measured post-intervention, the estimation sample is restricted to individuals who completed the follow-up survey, representing approximately 24\% of enrolled participants (Table~\ref{tab:design}).\footnote{This attrition raises the concern that differential selection into the follow-up could bias the estimates. We address this in two ways. First, response rates conditional on invitation are virtually identical across treatment arms (Table~\ref{tab:perm_test_results}): permutation tests fail to reject equal response rates both unconditionally and after covariate residualization, supporting the response-invariance assumption formalized in Appendix~\ref{sec:selectedsurvey}. Second, we replicate the main results using inverse probability weighting (IPW) and double-debiased machine learning (DML) estimators that flexibly control for observable predictors of follow-up participation (Tables~\ref{tab:dml_eff_selected} and~\ref{tab:dml_wage_selected}); estimates are very close to the naive specification throughout, indicating that selection on observables is of limited magnitude.}

\subsection{Search effort}

To analyze the impact on search effort, we rely exclusively on responses from the follow-up survey (see Section~\ref{subsec:surveys}).

\paragraph{Average treatment effects.}

We estimate average treatment effects using equation \eqref{eq:mod2}. Results are reported in Table \ref{tab:ate_search_behavior}. The motivational recommendation increases weekly search effort by 0.89 hours  on average. By contrast, the occupational recommendation has no statistically significant effect on search effort: the estimated coefficient is small (0.14 hours/week) and not significant. This pattern reinforces the idea that encouraging diversification is not an effective margin of adjustment for pessimistic job seekers, who instead respond along the effort margin.\footnote{Table \ref{tab:ate_search_behavior_wording} reports average treatment effects of the two different combinations of our two suggestions and the three possible wordings. In addition to the ``naive'' specification it also provides the estimated ATE when using the DML correction and the IPW weighting. We find all specifications to yield similar estimates, suggesting that the bias introduced by relying on the follow-up subsample is of limited magnitude. }

\paragraph{Heterogeneity by baseline search effort.}

The effect of the motivational recommendation should depend on initial search intensity: the intervention explicitly encourages individuals to reach a benchmark of approximately 10 hours of weekly search effort (three half-days), so treatment effects should concentrate among individuals below this threshold.

\medskip 

We test this hypothesis by splitting the sample according to pre-intervention job-search effort (below versus above 10 hours per week). Results are reported in columns (2) and (3) of Table \ref{tab:ate_search_behavior} for individuals with low and high baseline effort, respectively. The effects of the motivational recommendation are entirely driven by individuals with low baseline effort. Among those searching fewer than 10 hours per week, the treatment increases search effort by approximately 1.6 hours per week (20\% relative to baseline), and the estimates are highly statistically significant. In contrast, column (3) shows no statistically significant effect among individuals already searching more than 10 hours per week.

\medskip 

This pattern is consistent with the hypothesis that the motivational recommendation relaxes a binding effort constraint for low-effort job seekers, while having limited scope to affect individuals who are already searching intensively.

\begin{table}[htbp]
\centering
\caption{Intervention Effects on Search Parameters One Month After the Intervention}
\label{tab:ate_search_behavior}

\scalebox{0.78}{
\begin{tabular}{lcccccc}
\toprule
& \multicolumn{6}{c}{Declared Search Parameters} \\
\cmidrule(lr){2-4} \cmidrule(lr){5-7}
& \multicolumn{3}{c}{Search Effort (hours/week)} & \multicolumn{3}{c}{Log Reservation Wage} \\
\cmidrule(lr){2-4} \cmidrule(lr){5-7}
& (1) & (2) & (3) & (4) & (5) & (6) \\
& Full Sample & Low Effort Before & High Effort Before & Full Sample & Low Wage & High Wage \\
\midrule

Occ. treatment
& \begin{tabular}[c]{@{}c@{}}0.138\\(0.306)\end{tabular}
& \begin{tabular}[c]{@{}c@{}}0.524\\(0.368)\end{tabular}
& \begin{tabular}[c]{@{}c@{}}-0.320\\(0.513)\end{tabular}
& \begin{tabular}[c]{@{}c@{}}0.006\\(0.009)\end{tabular}
& \begin{tabular}[c]{@{}c@{}}0.000\\(0.011)\end{tabular}
& \begin{tabular}[c]{@{}c@{}}0.009\\(0.013)\end{tabular}
\\[1.5ex]

Motivational treatment
& \begin{tabular}[c]{@{}c@{}}0.896***\\(0.310)\end{tabular}
& \begin{tabular}[c]{@{}c@{}}1.590***\\(0.382)\end{tabular}
& \begin{tabular}[c]{@{}c@{}}0.156\\(0.509)\end{tabular}
& \begin{tabular}[c]{@{}c@{}}0.019**\\(0.009)\end{tabular}
& \begin{tabular}[c]{@{}c@{}}0.024**\\(0.011)\end{tabular}
& \begin{tabular}[c]{@{}c@{}}0.016\\(0.014)\end{tabular}
\\

\midrule

Control Mean
& 12.100 & 8.200 & 17.200
& 2,034\,\euro & 1,853\,\euro & 2,172\,\euro \\

Equality p-value ($\beta^o = \beta^m$)
& 0.006 & 0.002 & 0.271
& 0.059 & 0.008 & 0.480 \\

Subgroup Equality p-value
& \multicolumn{3}{c}{0.043}
& \multicolumn{3}{c}{0.651} \\

Observations
& 3,485 & 1,882 & 1,603
& 2,366 & 937 & 1,429 \\

$R^2$
& 0.342 & 0.111 & 0.175
& 0.712 & 0.679 & 0.702 \\

\bottomrule
\end{tabular}
}

\vspace{0.75em}

\parbox{\textwidth}{\footnotesize
\textit{Notes:} ``Occ.'' denotes occupation. This table reports estimated average treatment effects of the occupational and motivational recommendations on search parameters reported in the follow-up survey, pooling all wording variants. Outcomes include declared search effort and the logarithm of the reservation wage. Heterogeneity analyses split the sample according to pre-intervention search intensity (below versus above 10 hours per week) and according to reservation wage declared at PES registration (below versus above the median reservation wage among individuals with the same target occupation). The row ``Equality p-value'' reports tests of the null hypothesis that the occupational and motivational treatment effects are equal. The row ``Subgroup Equality p-value'' reports tests of equality of treatment effects across the low and high subgroups for each outcome. The sample is restricted to job seekers who completed the follow-up survey. Robust standard errors are reported in parentheses. $^{*}p<0.10$, $^{**}p<0.05$, and $^{***}p<0.01$.
}

\end{table}

\subsection{Reservation wage.}

Job seekers in the pessimist group set lower reservation wages than their optimistic counterparts, even after controlling for demographic characteristics and target occupations. The motivational recommendation therefore aims not only to increase search effort, but also to raise job seekers’ aspirations by encouraging them to target higher wages and higher-quality jobs.

\paragraph{Average treatment effects.}

We estimate the average treatment effects of the recommendations on post-intervention reservation wages using equation \eqref{eq:mod2}. Similar to the analysis of search effort, the naïve specification controls for pre-intervention reservation wages and subjective beliefs about the wage-offer distribution.\footnote{Specifically, we include job seekers’ perceived probability that a job offer would propose a wage above the median wage.} We also include fixed effects for target occupations which is strongly predictive of the reservation wage level.\footnote{Table \ref{tab:dml_wage_selected} reports the corresponding wording-specific estimates together with the DML and IPW corrections; as for search effort, all specifications yield very similar estimates.}

\medskip

The motivational recommendation increases reservation wages by 1.9\% on average. In contrast, the occupational recommendation, which does not target this search parameter, has no statistically significant effect. Table \ref{tab:ate_search_behavior} reports coefficient estimates and the p-value for the test of equality between the two treatment effects. Despite relatively small sample sizes, the p-value is below 0.10, leading us to reject the null hypothesis of equal effects at the 10\%  level. 

\medskip

\paragraph{Heterogeneity.}

Once again, the content of the motivational recommendation suggests the presence of heterogeneous treatment effects. While all individuals in this group are encouraged to search for higher-paying job opportunities, some job seekers are additionally provided with information on the median reservation wage among individuals sharing the same target occupation and living in the same region.\footnote{The reference reservation wage is computed using reservation wages declared at PES registration.} To avoid discouraging job seekers, this reference wage is disclosed only to those whose pre-intervention reservation wage falls below this benchmark. As a result, the informational content of the treatment varies with individuals’ initial reservation wages.

\medskip

We estimate treatment effects separately according to whether job seekers received additional information about the reservation wage of comparable individuals. Columns (5) and (6) of Table \ref{tab:ate_search_behavior} report the corresponding average treatment effects for individuals whose pre-intervention reservation wage was below and above the reference level, respectively. Among the former, the motivational recommendation increases reservation wages by 2.4\%. By contrast, we find no statistically significant effect among individuals whose reservation wages were already above the reference level prior to the intervention (column (6)). This heterogeneity likely reflects both differences in treatment intensity (the additional wage information was provided only to the low-wage group) and differences in the scope for adjustment.

\medskip

Overall, the motivational recommendation effectively raises reservation wages among pessimistic job seekers. Its effects are concentrated among individuals who initially set low reservation wages, which is the group explicitly targeted by the intervention.

\subsection{Impact of motivational recommendations on applications}

We now examine whether the changes in search parameters induced by the motivational recommendation translate into changes in job seekers’ application behavior. Because the intervention explicitly encourages individuals to target better job opportunities, its effect on the total number of applications is theoretically ambiguous. Greater selectivity, combined with increased search intensity, may leave the total number of applications unchanged. We therefore turn to the characteristics of the vacancies to which job seekers apply to assess whether the intervention affects the quality of their applications.

\medskip

Table \ref{tab:ate_applications_pessimists_id_avg} reports estimates of equation \eqref{eq:mod2} using either the number of applications or the characteristics of the vacancies to which individuals apply as outcomes. We focus in particular on the normalized wage measure $\tilde w$, defined as the ratio of a vacancy's posted wage to the job seeker's baseline reservation wage.

\medskip 

The table is organized into three panels. The top panel A reports results for the full sample of pessimistic job seekers, while the bottom panels B and C  split the sample according to whether the baseline reservation wage is respectively below or above the median reservation wage within the same occupation. This distinction is central to the design of the intervention: information on the median reservation wage is only provided to individuals whose own reservation wage lies below this benchmark (see Section \ref{sec:motivational}).

\medskip 

We find no effect on the total number of applications in column (1), either in the full sample or within either subsample. In contrast, the motivational recommendation leads to a shift in the type of vacancies targeted. In the full sample, we observe an increase in the lower part of the wage distribution (25th percentile, column (4)), although average effects remain small.

\medskip 

This effect is entirely driven by job seekers with below-median reservation wages, i.e., those explicitly encouraged to target higher-paying jobs. For this group, we find a positive and statistically significant effect on the normalized wage $\tilde w$, as well as a 11.5 percentage point increase in the probability of applying to at least one vacancy offering a wage above their baseline reservation wage. The quantile treatment effects further indicate sizable and statistically significant shifts mostly located in the first half of the distribution.

\medskip 

By contrast, we find no meaningful effects for job seekers with above-median reservation wages, consistent with the fact that the intervention does not provide additional wage information to this group. Taken together, these results provide clear evidence that the motivational recommendation increases the selectivity of applications among individuals with initially low reservation wages. The increase in search effort documented in the previous section, combined with this higher selectivity, naturally results in a stable total number of applications.

\begin{table}[htbp]
\centering
\caption{Application Characteristics by Treatment Arm in the Two Months Following  the Intervention}
\label{tab:ate_applications_pessimists_id_avg}

\scalebox{0.9}{
\begin{tabular}{lcccccc}
\toprule
%& \multicolumn{6}{c}{Applications in the Two Months Following the Intervention} \\
%\cmidrule(lr){2-7}
& (1) & (2) & (3) & (4) & (5) & (6) \\
& & \multicolumn{5}{c}{$\tilde w = \text{Posted Wage} / \text{Reservation Wage (Baseline)}$} \\
\cmidrule(lr){3-7}
& Applications & \multicolumn{2}{c}{OLS} & \multicolumn{3}{c}{Quantile Treatment Effects} \\
\cmidrule(lr){3-4} \cmidrule(lr){5-7}
& Count & $\tilde w$ & $1(\tilde w \geq 1)$ & 25th Pct. & 50th Pct. & 75th Pct. \\
\midrule

\multicolumn{7}{l}{\textbf{Panel A. Full Sample}} \\

Motivational treatment
& \begin{tabular}[c]{@{}c@{}}0.007\\(0.077)\end{tabular}
& \begin{tabular}[c]{@{}c@{}}0.013\\(0.023)\end{tabular}
& \begin{tabular}[c]{@{}c@{}}0.047\\(0.036)\end{tabular}
& \begin{tabular}[c]{@{}c@{}}0.036***\\(0.017)\end{tabular}
& \begin{tabular}[c]{@{}c@{}}0.009\\(0.011)\end{tabular}
& \begin{tabular}[c]{@{}c@{}}0.033\\(0.030)\end{tabular}
\\[1.5ex]

Occ. treatment
& \begin{tabular}[c]{@{}c@{}}-0.021\\(0.078)\end{tabular}
& \begin{tabular}[c]{@{}c@{}}0.006\\(0.021)\end{tabular}
& \begin{tabular}[c]{@{}c@{}}0.041\\(0.035)\end{tabular}
& \begin{tabular}[c]{@{}c@{}}0.037***\\(0.016)\end{tabular}
& \begin{tabular}[c]{@{}c@{}}0.004\\(0.008)\end{tabular}
& \begin{tabular}[c]{@{}c@{}}0.002\\(0.021)\end{tabular}
\\

\midrule

\multicolumn{7}{l}{\textbf{Panel B. Low Wage Prior}} \\

Motivational treatment
& \begin{tabular}[c]{@{}c@{}}-0.035\\(0.119)\end{tabular}
& \begin{tabular}[c]{@{}c@{}}0.058*\\(0.031)\end{tabular}
& \begin{tabular}[c]{@{}c@{}}0.115**\\(0.055)\end{tabular}
& \begin{tabular}[c]{@{}c@{}}0.050***\\(0.025)\end{tabular}
& \begin{tabular}[c]{@{}c@{}}0.064***\\(0.021)\end{tabular}
& \begin{tabular}[c]{@{}c@{}}0.047\\(0.036)\end{tabular}
\\[1.5ex]

Occ. treatment
& \begin{tabular}[c]{@{}c@{}}-0.117\\(0.118)\end{tabular}
& \begin{tabular}[c]{@{}c@{}}0.021\\(0.025)\end{tabular}
& \begin{tabular}[c]{@{}c@{}}0.044\\(0.054)\end{tabular}
& \begin{tabular}[c]{@{}c@{}}0.021\\(0.024)\end{tabular}
& \begin{tabular}[c]{@{}c@{}}0.014\\(0.017)\end{tabular}
& \begin{tabular}[c]{@{}c@{}}0.010\\(0.039)\end{tabular}
\\

\midrule

\multicolumn{7}{l}{\textbf{Panel C. High Wage Prior}} \\

Motivational treatment
& \begin{tabular}[c]{@{}c@{}}0.043\\(0.101)\end{tabular}
& \begin{tabular}[c]{@{}c@{}}-0.028\\(0.034)\end{tabular}
& \begin{tabular}[c]{@{}c@{}}-0.006\\(0.041)\end{tabular}
& \begin{tabular}[c]{@{}c@{}}0.025\\(0.018)\end{tabular}
& \begin{tabular}[c]{@{}c@{}}-0.005\\(0.010)\end{tabular}
& \begin{tabular}[c]{@{}c@{}}-0.020\\(0.024)\end{tabular}
\\[1.5ex]

Occ. treatment
& \begin{tabular}[c]{@{}c@{}}0.062\\(0.103)\end{tabular}
& \begin{tabular}[c]{@{}c@{}}-0.005\\(0.033)\end{tabular}
& \begin{tabular}[c]{@{}c@{}}0.051\\(0.041)\end{tabular}
& \begin{tabular}[c]{@{}c@{}}0.034***\\(0.017)\end{tabular}
& \begin{tabular}[c]{@{}c@{}}0.017\\(0.010)\end{tabular}
& \begin{tabular}[c]{@{}c@{}}0.000\\(0.021)\end{tabular}
\\

\midrule

Observations
& 17,073 & 13,111 & 13,111 & 13,111 & 13,111 & 13,111 \\

Job Seekers
& 17,073 & 2,695 & 2,695 & 2,695 & 2,695 & 2,695 \\

\bottomrule
\end{tabular}
}

\vspace{0.75em}

\parbox{\textwidth}{\footnotesize
\textit{Notes:} ``Occ.'' denotes occupation. This table reports treatment effects on the number of applications submitted and on the characteristics of vacancies applied to during the two months following the intervention. Reservation wages are measured at baseline, and all outcomes are constructed from administrative data. Column 1 reports effects on the total number of applications. Columns 2--6 are conditional on having submitted at least one application and are based on the normalized wage measure $\tilde w$, defined as the ratio of the posted wage to the individual's baseline reservation wage. Column 2 reports effects on $\tilde w$, Column 3 on the indicator $1(\tilde w \geq 1)$, and Columns 4--6 report quantile treatment effects at the 25th, 50th, and 75th percentiles of $\tilde w$. The top panel reports results for all pessimistic job seekers. The lower panels split the sample according to whether the reservation wage declared at PES registration is below or above the median among comparable job seekers (defined within occupation and location). The sample is restricted to pessimistic individuals with a reported baseline reservation wage. Standard errors clustered at the job seeker level are reported in parentheses. $^{*}p<0.10$, $^{**}p<0.05$, and $^{***}p<0.01$.
}

\end{table}

\section{Reemployment Outcomes}\label{sec:employment}

We now examine the impact of the intervention on job seekers’ reemployment outcomes. We focus on the first employment spell beginning after the intervention date (October 2\textsuperscript{nd}, 2023). To do so, we rely on an exhaustive administrative database containing monthly wage declarations reported by employers for each worker. For each job seeker in our sample, we identify the first post-intervention job and consider three outcomes: whether an employment spell begins within twelve months of the intervention, whether this job is a permanent contract, and its associated hourly wage. Importantly, these administrative data are exhaustive, allowing us to include all individuals who participated in the baseline survey. As a result, our estimates are not restricted to job seekers who use the PES platform or who responded to the follow-up survey.

\subsection{Effect of occupational recommendation for the optimists}

We find that the occupational recommendation has a limited effect on the reemployment outcomes of job seekers classified as optimists. Table \ref{tab:ate_reemployment_optimists} reports the average treatment effects of the occupational recommendation on the three reemployment outcomes, restricting the sample to job seekers classified as optimists. In the full sample (Panel A), the recommendation does not significantly affect reemployment outcomes, despite having a statistically significant positive effect on the number of applications submitted. Treated individuals are no more likely than control individuals to find a job within twelve months of the intervention, and among those who do find a job, contract types and wages are similar across groups. These conclusions also hold when the sample is split according to initial beliefs about the narrowness of reemployment prospects across occupations. The only significant effect is a 1.5\% increase in the average normalized hiring wage (relative to the reservation wage) for those with a narrow initial belief about the wage dispersion (Panel B, column (4)), driven by those in the upper part of the wage distribution (column (8)).

\medskip

%%%%%%%%%%%%%%%%%%%%%%%%%%%%%%%%%%%%%%%%%%%%%%%%%
\subsection{Pessimists}

We now examine the impact of both occupational and motivational recommendations on the reemployment outcomes of job seekers classified as pessimists. We consider the same three outcomes as in the previous section and control for baseline beliefs about three-month reemployment probabilities, as well as for the number of applications submitted in the month preceding the intervention. Table \ref{tab:ate_reemployment_pessimists} reports the corresponding average treatment effects.

\medskip

The upper panel A of Table \ref{tab:ate_reemployment_pessimists} report the effects of the occupational recommendation, first for the full sample and then across subgroups defined by initial search effort and reservation wages. While these dimensions are primarily relevant for the motivational treatment, they provide a useful benchmark for assessing heterogeneity in the occupational treatment as well.

\medskip 

Overall, the occupational recommendation has modest effects on reemployment outcomes. We find a small increase in the probability of reemployment within 12 months, amounting to approximately 1.8 percentage points relative to a control mean of 40.7\%, although the estimate is only marginally statistically significant. This increase is accompanied by a slight decline in the probability of obtaining a permanent contract and a modest increase in wages. However, these effects are also limited in magnitude and statistical significance, and are primarily concentrated among job seekers with higher initial search effort. Taken together, these results suggest that the occupational recommendation may induce some degree of occupational mobility among pessimists, leading job seekers to reallocate across jobs with different characteristics. 

\medskip

The route to this reemployment gain differs from the one the recommendation was designed to activate. Pessimists did not take up the recommendation along its intended margin: their applications to the suggested occupations, and to posted vacancies more broadly, did not increase (Table~\ref{tab:ate_app}), and the occupational treatment left their search effort unchanged (Table~\ref{tab:ate_search_behavior}). Their only behavioral response was a rise in spontaneous, direct-to-firm applications (an increase of roughly 90\%). The occupational reemployment gain among pessimists therefore does not operate through diversification into the recommended occupations, nor through greater search intensity. Because reemployment is measured from exhaustive administrative hiring records, independently of the channel through which a job is found, this pattern is most consistent with a generic activation of search: the recommendation nudges pessimists toward an untargeted, spontaneous application channel whose occupational content we do not observe, and this shows up as a marginal reemployment gain. Consistent with a mobility rather than a selectivity response, the gain coincides with a decline in the probability of a permanent contract (Table~\ref{tab:ate_reemployment_pessimists}, column~(3)): pessimists appear to enter employment somewhat faster, but into less stable jobs. Given the marginal precision of these estimates, we read this channel as suggestive rather than established. The comparison with the motivational arm is instructive: both raise pessimists' reemployment by a comparable amount, but only the motivational arm does so along the margins it targets and while improving the quality of accepted jobs. Matching the recommendation to the underlying friction thus governs not merely whether reemployment rises, but the margin through which it rises and the quality of the jobs obtained.

\medskip

The lower panel B report the effects of the motivational recommendation. As discussed above, this intervention is designed to increase search effort, especially among individuals with initially low effort, and to encourage applications to higher-paying jobs for those with below-median reservation wages.

\medskip 

We find a modest but statistically significant increase in reemployment at 12 months, of 2.2 percentage points relative to a control mean of 40.7\%. This effect is primarily driven by job seekers with low initial search effort, consistent with the intended mechanism of the intervention.

\medskip 

In contrast, we do not observe a significant effect on reemployment among individuals with below-median reservation wages. However, for this group, the motivational recommendation leads to improvements in the quality of jobs obtained. In particular, we observe an increase in wages relative to baseline reservation wages: the average normalized wage increases by 0.040 (relative to a control mean of 1.17), and the probability of obtaining a job with a wage above the baseline reservation wage increases by 5.3 percentage points (from a control mean of 55.4\%). We also find significant quantile treatment effects, particularly at the upper end of the distribution.

\medskip

Taken together, the motivational recommendation affects both margins that the model identifies as constrained among pessimists. For workers with low baseline search effort, it increases the probability of reemployment by 2.9 percentage points. For workers with low baseline reservation wages, it improves the quality of accepted jobs, increasing the normalized wage by 0.040 and the probability that the accepted wage exceeds the reservation wage by 5.3 percentage points. The overall effect across these two subgroups is therefore stronger than the effect on either margin alone, and each margin responds precisely in the dimension targeted by the intervention. All results on job characteristics are conditional on reemployment and therefore pertain to a selected subpopulation. Accordingly, they should be interpreted as changes in the composition of accepted jobs rather than as causal effects for the full randomized sample.

%%%%%%%%%%%%%%%%%%%%%%%%%%%%%%%%%%%%%%%%%%%%%%%%

\begin{table}[!htbp]
\centering
\caption{Reemployment Outcomes of Pessimists}
\label{tab:ate_reemployment_pessimists}
\scalebox{0.8}{
\begin{threeparttable}
\begin{tabular}{lcccccccc}
\toprule
& (1) & (2) & (3) & (4) & (5) & (6) & (7) & (8) \\
\cmidrule(lr){2-9}
Dependent variable
& \multicolumn{2}{c}{Reemployment} 
& Contract 
& \multicolumn{5}{c}{Hiring wage relative to reservation wage} \\
\cmidrule(lr){2-3} \cmidrule(lr){4-4} \cmidrule(lr){5-9}
& 6 months & 12 months & Permanent & Mean & $\geq 1$ 
& QTE 25 & QTE 50 & QTE 75 \\
\midrule

\multicolumn{9}{l}{\textbf{Panel A. Occupational treatment}} \\

Full sample 
& 0.017$^{*}$ & 0.018$^{*}$ & -0.021$^{*}$ & -0.007 & 0.033 & 0.012 & 0.009 & 0.002 \\
& (0.009) & (0.010) & (0.013) & (0.011) & (0.022) & (0.011) & (0.007) & (0.016) \\

\addlinespace
\multicolumn{9}{l}{\textit{Heterogeneity by initial search effort}} \\

$\times$ Low effort prior 
& 0.010 & 0.013 & -0.014 & -0.017 & -0.004 & 0.010 & 0.000 & -0.012 \\
& (0.011) & (0.013) & (0.018) & (0.016) & (0.030) & (0.016) & (0.011) & (0.028) \\

$\times$ High effort prior 
& 0.023 & 0.022 & -0.033$^{*}$ & 0.004 & 0.075$^{**}$ & 0.011 & 0.018$^{**}$ & 0.026 \\
& (0.015) & (0.016) & (0.019) & (0.014) & (0.031) & (0.015) & (0.009) & (0.020) \\

\addlinespace
\multicolumn{9}{l}{\textit{Heterogeneity by initial reservation wage}} \\

$\times$ Low wage prior 
& 0.002 & 0.015 & -0.045$^{**}$ & 0.006 & 0.044 & 0.010 & 0.016 & 0.017 \\
& (0.014) & (0.015) & (0.018) & (0.015) & (0.031) & (0.013) & (0.012) & (0.022) \\

$\times$ High wage prior  
& 0.029$^{**}$ & 0.021 & 0.0008 & -0.017 & 0.029 & -0.004 & -0.000 & -0.008 \\
& (0.012) & (0.013) & (0.018) & (0.015) & (0.030) & (0.014) & (0.009) & (0.021) \\

\midrule
\multicolumn{9}{l}{\textbf{Panel B. Motivational treatment}} \\

Full sample 
& 0.009 & 0.022$^{**}$ & -0.011 & 0.009 & 0.030 & 0.010 & 0.010 & 0.030$^{*}$ \\
& (0.009) & (0.010) & (0.013) & (0.011) & (0.021) & (0.011) & (0.007) & (0.015) \\

\addlinespace
\multicolumn{9}{l}{\textit{Heterogeneity by initial search effort}} \\

$\times$ Low effort prior   
& 0.002 & 0.029$^{**}$ & -0.019 & -0.007 & 0.013 & -0.001 & 0.002 & 0.010 \\
& (0.011) & (0.013) & (0.017) & (0.017) & (0.030) & (0.017) & (0.011) & (0.021) \\

$\times$ High effort prior 
& 0.014 & 0.007 & -0.012 & 0.026$^{*}$ & 0.049 & 0.018 & 0.017$^{*}$ & 0.051$^{***}$ \\
& (0.015) & (0.016) & (0.020) & (0.015) & (0.031) & (0.013) & (0.009) & (0.020) \\

\addlinespace
\multicolumn{9}{l}{\textit{Heterogeneity by initial reservation wage}} \\

$\times$ Low wage prior 
& 0.0006 & 0.024 & -0.021 & 0.040$^{**}$ & 0.053$^{*}$ & 0.024$^{*}$ & 0.017 & 0.046$^{**}$ \\
& (0.013) & (0.015) & (0.019) & (0.016) & (0.031) & (0.012) & (0.012) & (0.019) \\

$\times$ High wage prior 
& 0.016 & 0.020 & -0.001 & -0.015 & 0.014 & -0.010 & -0.004 & 0.010 \\
& (0.012) & (0.013) & (0.018) & (0.016) & (0.030) & (0.015) & (0.009) & (0.022) \\

\midrule
Target occupation FE 
& \checkmark & \checkmark & \checkmark & & & & & \\
Initial reemployment beliefs 
& \checkmark & \checkmark & \checkmark & & & & & \\
Initial wage beliefs 
& & & & \checkmark & \checkmark & \checkmark & \checkmark & \checkmark \\
Applications prior 
& \checkmark & \checkmark & \checkmark & \checkmark & \checkmark & \checkmark & \checkmark & \checkmark \\

\midrule
Control mean 
& 0.271 & 0.407 & 0.199 & 1.06 & 0.554 & 0.93 & 1.02 & 1.17 \\
Observations 
& 17,066 & 17,066 & 7,107 & 4,062 & 4,062 & 4,062 & 4,062 & 4,062 \\

\bottomrule
\end{tabular}

\begin{tablenotes}[flushleft]
\footnotesize
\item \textit{Notes:} This table reports regression coefficients from estimates of reemployment outcomes on treatment status for job seekers classified as pessimists. Each column corresponds to a different dependent variable. Permanent contract status and hiring wages are observed conditional on reemployment within twelve months. Hiring wages are expressed relative to the reservation wage. Hourly wages are winsorized between \euro 11.52 and \euro 26. “QTE” denotes quantile treatment effects. “FE” denotes fixed effects. “Low” and “high” effort are defined relative to 10 hours of initial weekly search effort. “Low” and “high” wage are defined relative to the occupation-specific median reservation wage as declared at the PES registration. Standard errors, clustered at the job-seeker level, are reported in parentheses. Significance levels: $^{*}p<0.1$, $^{**}p<0.05$, $^{***}p<0.01$.
\end{tablenotes}

\end{threeparttable}
}
\end{table}

\section{Interpretation: Mapping Predictions to Evidence}\label{subsec:interpretation}

The model implies that a recommendation moves behavior only when it acts on a margin the worker can adjust, and that the three candidate channels introduced in Section~\ref{subsec:mechanisms} (Bayesian updating (BU), attention activation (AA), and aspiration shift (AS)) leave distinct empirical fingerprints. As set out there, BU would shift post-treatment beliefs, most for the most uncertain or biased priors; AA would leave beliefs about success probabilities unmoved, with effects only where priors are positive but inattentive; and AS would also leave beliefs unmoved, raising reservation wages only below the disclosed benchmark. We read the evidence against these signatures in turn. Throughout, we treat the absence of a measured belief response as evidence that updating is unlikely to be the operative channel, not as a decisive rejection: a small or imprecisely estimated update cannot be excluded with our sample, and we say so where it matters.

\medskip

Table~\ref{tab:pred_map} maps each of the eight predictions of Section~\ref{sec:treatment} to the test that discriminates it, the resulting finding, and the supporting evidence. The remainder of this section discusses the reading behind this map, arm by arm.

\paragraph{The occupational arm: the data point to attention rather than belief updating.} If the occupational recommendation worked by transmitting information, posterior beliefs about success in alternative occupations should move, most for workers with low priors. The data give little sign of this: neither the belief-narrowness index $h_{\text{post}}$ nor its interaction with high perceived outside options moves significantly, and the null extends to wage beliefs, perceived arrival rates, and perceived returns to search effort, in both average effects and heterogeneity (Predictions P.1 and P.3; Table~\ref{tab:pred_map}). We read this comprehensive absence of a belief response as evidence against an operative belief-updating channel.

\medskip 

What the data do show is the prior dependence that the attention model predicts. The effect of the occupational treatment on applications is concentrated among workers who already hold dispersed beliefs. Among optimists with high perceived outside options (low $h$), the treatment raises applications to the suggested occupations by roughly $45\%$ of the control mean, and increases applications to other occupations and total applications as well; among workers with narrow beliefs, the estimated effects are small and statistically insignificant throughout. Consistent with this pattern, the interaction coefficient $\beta_S$ is significant at the 5\% level for applications to suggested occupations, other occupations, and vacancies overall (Table~\ref{tab:ate_by_initial_beliefs}).

\medskip 

If the recommendation had provided genuinely new information about $p_k$, one would expect the low-prior group to respond most strongly. Instead, the response is concentrated among workers who already hold favorable priors. Among pessimists, the corresponding interaction coefficient is statistically insignificant ($\beta_S=-0.276$ for total applications, s.e.\ $0.863$), suggesting that limited attention is not their primary friction. Behavior seems to change where workers already perceive an opportunity but had not previously acted upon it.

\medskip 

The machine-learning-based heterogeneity analysis reinforces and refines this reading. Using the generic machine learning (GML) framework of \citet{chernozhukov2018generic}, we examine treatment-effect heterogeneity for the occupational arm on the full sample and separately for pessimists and non-pessimists (Tables~\ref{tab:mcomp_occ_all}--\ref{tab:clan_occ_opt_retour_12m}).\footnote{We use 5-fold cross-fitting, causal boosting as described in \citet{chernozhukov2018generic}, and the X-learner of \citet{kunzel2019metalearners}.}

\medskip 

A first key finding is the heterogeneity between pessimists and non-pessimists. The occupational recommendation has a significant positive effect on 12-month reemployment for pessimists (median ATE $= 2.1$ percentage points, $p = 0.035$), but essentially zero effect for non-pessimists (ATE $= -0.2$ pp, $p = 0.72$). This differential is the primary source of treatment-effect heterogeneity: on the full sample, the best-performing method (Elastic Net with causal boost) yields a heterogeneity loading $\hat{\beta}_2 = 0.224$ with a one-sided $p$-value of $0.079$, marginally significant at the 10\% level. By contrast, no method detects significant within-group heterogeneity when the GML is run separately on pessimists or non-pessimists alone. The treatment does not affect relative wages for any subgroup.

\medskip 

The CLAN analysis yields a consistent profile of treatment beneficiaries across all three samples. A particularly robust discriminator is labor-market tightness in the target occupation. Individuals who benefit most from the intervention face substantially lower labor-market tightness in their primary occupation, implying fewer vacancies relative to the number of job seekers (Tables~\ref{tab:clan_occ_opt_retour_12m} and \ref{tab:clan_occ_pess_retour_12m}). These workers search in labor-market segments where competition for available positions is particularly intense. For optimists, the occupational recommendation appears to benefit them by redirecting search toward occupations offering more abundant employment opportunities. For pessimists, this reallocation is not visible in their posted-vacancy applications, which do not increase; to the extent it occurs, it can only operate through the spontaneous, direct-to-firm applications whose occupational content we do not observe.

\medskip

A second notable characteristic among pessimists concerns beliefs about occupational prospects. Individuals who benefit most from the intervention place relatively greater weight on their primary occupation. This pattern is consistent with the occupational recommendation stimulating search among pessimists whose attention was narrowly concentrated on their primary occupation. Because we do not observe a corresponding broadening of their posted-vacancy applications, we read this as an activation of search rather than as a documented diversification of the occupational set: the recommendation appears most effective when it prompts workers facing weak local demand to search more, plausibly through the unobserved spontaneous channel, rather than when it provides fundamentally new information about hiring probabilities.

\paragraph{The motivational arm: the data support an aspiration and effort shift.} The reference-point mechanism is supported on each margin the model isolates. The motivational recommendation raises the log reservation wage by $1.9\%$ on average, while the occupational arm does not; the two effects differ at the $10\%$ level (Table~\ref{tab:ate_search_behavior}). The effect shows the threshold heterogeneity of Lemma~\ref{lem:rw}: it is concentrated among workers whose pre-intervention reservation wage lay below the peer benchmark ($2.4\%$), with no detectable effect above it. A belief-updating account does not predict this kink at the benchmark. It also accords with the wording analysis: the raise-awareness wording, which directly invites reconsideration of one's target, amplifies the reservation-wage effect, whereas the peers'-behavior wording produces an effect indistinguishable from zero.

\medskip 

The model offers two routes to a higher reservation wage, and the data discriminate between them. The wage could rise because the intervention corrects $\tilde F$-pessimism and lifts the continuation value, or because it shifts the perceived peer aspiration $r_w$ and triggers loss aversion below it (Lemma~\ref{lem:rw}). The elicited probability that an offer would exceed the median wage does not move (see Table~\ref{tab:ols_beliefs}), and, together with the broader finding that no beliefs about labor-market fundamentals shift under either arm, this points to the aspiration channel rather than belief correction as the operative one (prediction P.5). As above, we read an unchanged measured   belief as suggestive rather than dispositive, but it is the pattern the reference-point channel predicts and the $\tilde F$-correction channel does not.

\medskip 

On effort, the motivational recommendation adds about $0.9$ hours of weekly search on average, near-identical across wordings, while the occupational arm has no effect (Tables~\ref{tab:ate_search_behavior} and \ref{tab:ate_search_behavior_wording}). The gain is concentrated where effort was depressed: among workers below the $10$-hour benchmark stated in the message it is about $1.6$ hours, against essentially zero for those already above, with the difference significant ($p=0.043$), as the model predicts when the message shifts the peer effort norm $r_e$. The two responses combine into a change in the \emph{composition} of search rather than its volume: higher $w^*$ and higher effort leave total applications among pessimists ambiguous, and, as documented in Section~\ref{sec:motivational_impact} (Table~\ref{tab:ate_applications_pessimists_id_avg}), the data confirm a reallocation toward higher-wage vacancies rather than an expansion in count.

\medskip 

The GML heterogeneity analysis confirms that the motivational treatment has a significant average effect on 12-month reemployment among pessimists (median ATE $= 2.2$ percentage points, $p = 0.027$; Table~\ref{tab:mcomp_mot}), robust across all machine learning methods. While treatment-effect heterogeneity as measured by the BLP loading $\hat{\beta}_2$ does not reach conventional significance under honest median $p$-values (best method: $\hat{\beta}_2 = 0.102$, one-sided $p = 0.35$), the GATES analysis reveals a clear monotone gradient in the predicted treatment effects (Table~\ref{tab:gates_mot}). The top quintile (G5) gains $5.3$ percentage points in 12-month reemployment, with the 90\% confidence interval excluding zero $[0.2\text{pp},\; 10.3\text{pp}]$, while the bottom quintile (G1) shows a near-zero effect ($-0.9$ pp).

\medskip 

The CLAN analysis for 12-month reemployment (Table~\ref{tab:clan_mot_retour_12m}) shows that the most-affected pessimists are those with the lowest baseline search effort: $87.6\%$ of the top quintile have prior effort below 10 hours, compared with $33.3\%$ in the bottom quintile (difference $0.54$, $p < 0.001$). Average search hours are $7.7$ for the most-affected versus $15.1$ for the least-affected ($p < 0.001$). These low-effort pessimists also face substantially tighter labour markets and hold a less concentrated occupational focus, yet display higher subjective hiring probabilities across all three occupations.  This profile is more consistent with a motivational deficit than with an informational deficit. These individuals perceive abundant hiring opportunities but do not respond with commensurate search effort. The motivational intervention appears to overcome this inertia.

\medskip 

The beneficiary profile for relative wages is entirely different (Table~\ref{tab:clan_mot_rel_wage}). The pessimists who benefit most in terms of wage match quality are high-effort searchers ($19.8$ hours vs.\ $6.6$ hours, $p < 0.001$; only $10\%$ with effort $\leq 10$ vs.\ $95\%$) who hold a low wage prior: $90.8\%$ of the most-affected have a reservation wage below the recommended salary, compared with $3.7\%$ among the least-affected (difference $0.87$, $p < 0.001$). These are pessimists who search actively but whose low wage expectations lead them to target jobs below their potential. The motivational intervention helps them aim higher and secure positions that better match their skills, resulting in a higher ratio of hiring wage to reservation wage. They also face looser labor markets (tension $0.46$ vs.\ $0.02$, $p < 0.001$), where opportunities exist but go unexploited due to unrealistic wage expectations.
\medskip 

These two CLAN profiles reveal complementary channels through which motivational counselling operates among pessimists: it activates search among the low-effort discouraged, and it raises wage ambitions among the high-effort underconfident. Both groups are ``stuck'', but for different reasons that the same intervention addresses through distinct mechanisms.

\paragraph{Taking stock.} Across both treatment arms, the evidence is difficult to reconcile with a belief-updating mechanism and instead supports two distinct behavioral channels. For the occupational treatment, the evidence is consistent with an attention-activation mechanism, operating only among workers with dispersed beliefs and strongest where labor-market tightness is low in the recommended occupations. For the motivational treatment, the evidence is consistent with a reference-point and aspiration-shift mechanism, operating only among workers whose baseline reservation wages or search effort lie below the relevant benchmarks. The personalized design of the interventions therefore receives direct empirical support: each treatment affects behavior precisely where the underlying friction aligns with its intended mechanism. A subtler lesson concerns the occupational arm among pessimists. Both recommendations raise their twelve-month reemployment by a comparable amount (around two percentage points), but they do so very differently. The occupational recommendation acts through an untargeted margin: pessimists neither take up the suggested occupations nor search more, yet a rise in spontaneous applications coincides with faster, though less stable, job entry, as reflected in a lower probability of a permanent contract. The motivational recommendation, by contrast, moves exactly the margins it targets and improves the quality of accepted jobs. The case for need-based personalization is therefore not only that targeting can raise the reemployment response, but that it makes the response controllable: it determines whether reemployment is achieved with greater selectivity and better job quality, or through an untargeted channel with an ambiguous quality profile. One caveat carries over to the welfare discussion. Because the motivational treatment does not appear to correct pessimistic beliefs about $\tilde F$, and because we find no evidence of learning, its effects may be short-lived and could require repeated interventions to persist over time.

\paragraph{Alternatives considered.} Several alternative interpretations can account for parts of the empirical pattern, but none can explain the full set of results. A pure information channel operating through labor-market tightness would predict heterogeneity by perceived tightness rather than by belief narrowness $h$. The concentration of effects among low-$h$ workers, irrespective of their perceived tightness, is difficult to reconcile with such an explanation. Likewise, an experimenter-demand effect would be expected to raise reservation wages uniformly among workers exposed to the motivational framing, regardless of whether their baseline reservation wage lies above or below the benchmark. The observed threshold heterogeneity is therefore more consistent with a reference-point mechanism. A pure implementation-friction account, whereby the recommendation simply reduces the cost of applying, would predict effects regardless of prior beliefs, yet the effect disappears among workers with narrow beliefs. Finally, self-directed learning triggered by the recommendation would be expected to generate similar responses across treatment arms for comparable workers, whereas the observed responses are treatment-specific: diversification under the occupational treatment and increases in search effort and $w^*$ under the motivational treatment. Taken together, these findings leave the two-margin interpretation as a parsimonious explanation for the overall pattern of results.

\section*{Conclusion}

Support policies should not be designed as uniform information-delivery devices. The same recommendation produces different adjustments depending on a worker's binding constraint and can produce iatrogenic effects: where the constraint is under-exploration of alternative occupations, an occupational recommendation activates already-perceived opportunities and diversifies search; where it is low effort or low selectivity, the same message is misaligned and leaks onto untargeted margins. The data tools available to PESs can identify not only relevant opportunities but the constraints workers face, and those constraints can guide content: diversification for some, effort encouragement for others, support for selectivity or for targeting higher-quality jobs for others.

\medskip 

A rough welfare calculation underscores the potential gains from personalization. The motivational recommendation costs essentially nothing beyond the survey infrastructure, as it is a short automated message. Among the low-effort pessimists it targets, it raises 12-month reemployment by 2.9 percentage points (from a base of 40.7\%). With approximately 6,800 pessimists assigned to the motivational arm, this implies roughly 200 additional reemployments. Combining this with the wage-quality gains for low-wage pessimists (a 4\% increase in the normalized hiring wage) suggests that even modest personalization of automated messages can generate meaningful labor market improvements at near-zero marginal cost. 

\medskip 

More generally, effective recommendations are those that target the right behavioral margin, not merely those that convey the right information. For a substantial share of job seekers, the central challenge is translating perceived opportunities into concrete search.

\FloatBarrier
\bibliographystyle{myagsm}
\bibliography{ref.bib}
\newpage

\appendix

\renewcommand{\thesection}{\Alph{section}}

\counterwithin{equation}{section}
\counterwithin{figure}{section}
\counterwithin{table}{section}

\section{Model's Complementary results and Proofs}\label{sec:proof}

\begin{lem}[Wage aspiration and reservation wage]\label{lem:rw}
	Hold $\overline{U}$ fixed and suppose the worker is in the aspiration-gap region ($w^* < r_w$). Then $\partial w^* / \partial r_w |_{\overline{U}} > 0$, i.e., raising the perceived peer wage aspiration raises the worker's reservation wage. %For workers already above the aspiration ($w^* \geq r_w$), the comparative static does not apply.
\end{lem}

\textbf{Proof of Lemma \ref{lem:rw}.}
The reservation wage satisfies $U(w^*; r_w) = \overline{U}$. In the aspiration-gap region where $w^* < r_w$, this yields:
\[
v(w^*) + \kappa_w\bigl(v(w^*) - v(r_w)\bigr) = \overline{U},
\]
i.e., $(1 + \kappa_w)\,v(w^*) - \kappa_w\,v(r_w) = \overline{U}$. Differentiating implicitly at fixed $\overline{U}$ with respect to $r_w$:
\[
(1 + \kappa_w)\,v'(w^*)\,\frac{\partial w^*}{\partial r_w} - \kappa_w\,v'(r_w) = 0,
\]
so
\[
\frac{\partial w^*}{\partial r_w} = \frac{\kappa_w\,v'(r_w)}{(1 + \kappa_w)\,v'(w^*)} > 0.
\]
For $w^* \geq r_w$, $U(w^*; r_w) = v(w^*)$, which does not depend on $r_w$. \hfill $\square$

\bigskip

\begin{lem}[Effort aspiration and search effort]\label{lem:effort}
   Hold $\overline{U}$ fixed and suppose the effort aspiration increases from $r_e^0$ to $r_e^1 > r_e^0$.
    \begin{enumerate}
        \item[\emph{(i)}] Workers already below the old norm, with $e^* < r_e^0$, increase effort: the marginal cost curve shifts down by $\psi(r_e^1 - r_e^0)$ at every point, so the FOC requires higher $e^*$.
        \item[\emph{(ii)}] Workers previously at or above the old norm, now below the new one, with $e^* \in [r_e^0,\, r_e^1)$,  also increase effort: their marginal cost drops discretely from $c'(e^*)$ to $c'(e^*) - \psi(r_e^1 - e^*)$.
        \item[\emph{(iii)}] Workers above the new norm, with $e^* \geq r_e^1$  are unaffected.
    \end{enumerate}
\end{lem}

\textbf{Proof of Lemma \ref{lem:effort}.}
Define $$\mathrm{MB}(e) := \tilde\lambda'(e) \int_{w^*}^{\bar w} [U(w; r_w) - \overline{U}]\,d\tilde F(w),$$ which is decreasing ($\lambda'' \leq 0$), and $\mathrm{MC}(e; r_e) :=  C'(e; r_e)$ from \eqref{eq:mc_effort}, which is increasing ($C'' = c'' + \psi > 0$ below the norm, $C'' = c''$ above). Then, we have:
\begin{enumerate}
    \item[(i)] For $e^* < r_e^0 < r_e^1$: $\mathrm{MC}(e;\, r_e^1) = c'(e) - \psi(r_e^1 - e) = \mathrm{MC}(e;\, r_e^0) - \psi(r_e^1 - r_e^0)$ for all $e < r_e^0$. The MC curve shifts down uniformly while MB is unchanged, so the intersection moves right. If the new optimum exceeds $r_e^1$, the penalty deactivates and effort solves $\mathrm{MB}(e) = c'(e)$; in either case, $e^{**} > e^*$.

 \item[(ii)] For $e^* \in [r_e^0, r_e^1)$: before the shift, $\mathrm{MB}(e^*) = c'(e^*)$ (no penalty). After, $\mathrm{MC}(e^*;\, r_e^1) = c'(e^*) - \psi(r_e^1 - e^*) < c'(e^*) = \mathrm{MB}(e^*)$, so MB exceeds MC and effort increases.

 \item[(iii)] For $e^* \geq r_e^1$: MC $= c'(e^*)$ before and 
after. \hfill $\square$

\end{enumerate}

\section{Administrative Data}\label{app:admin_data}\label{sec:appendix1}

Our study is conducted in partnership with France Travail, the French PES, which provides a rich administrative data environment covering job seekers’ characteristics, job search activity, and labor market outcomes.

\medskip 

Upon registration with France Travail, job seekers are required to report a set of job search parameters that define the characteristics of acceptable job offers. These parameters include the reservation wage, the maximum commuting time, the target occupation, the desired type of labor contract (temporary or permanent), and preferred working time (full-time or part-time). These parameters are recorded in administrative data and can be updated over the course of the unemployment spell.

\medskip 

Registered job seekers can apply to vacancies directly through the France Travail online platform. Although job seekers may also rely on other job boards, the France Travail platform plays a central role in the matching process and accounts for a substantial share of vacancies in France. 

\medskip 

For each application submitted through the platform, we observe detailed information on the vacancy, including the occupation, the offered wage, and the type of contract. Importantly, we also observe the origin of each application: whether it corresponds to a direct response to a posted vacancy, follows a caseworker’s suggestion, or takes the form of a spontaneous application to a firm without a posted vacancy. This information allows us to distinguish between different application streams and to study how job seekers adjust their search behavior along these margins.

\medskip 

We construct reemployment indicators at different horizons by combining two administrative data sources, with a primary focus on a three-month horizon. Specifically, we merge the monthly Indicateur de retour à l’emploi (IRE) produced by France Travail with the Déclarations préalables à l’embauche (DPAE), which record firms’ hiring declarations. This indicator closely matches the definition of reemployment used by France Travail. Because job seekers were sampled in September 2023, reemployment spells beginning in September are included in the three-month reemployment window.

\medskip 

Overall, administrative data provide precise and comprehensive information on observed job search behavior and labor market outcomes. However, they do not capture job seekers’ subjective beliefs, expectations, or perceived returns to search effort. Measuring these dimensions, which play a central role in job search decisions, requires complementary survey data, described in the next section.

\section{Alternative Recommendation Framing}\label{sec:app-wordings}

The main text pools all framing variants. This appendix documents the design of the framing experiment and shows that treatment effects are broadly similar across framings.

\medskip

Each recommendation was implemented using three alternative framings:
\begin{itemize}
    \item[(i)] a direct-advice framing,
    \item[(ii)] a peer-behavior framing,
    \item[(iii)] an awareness framing emphasizing potential behavioral biases during job search.
\end{itemize}

The objective was to assess whether job seekers respond differently depending on how recommendations are presented. The informational content of the recommendation remained unchanged across framings.

\medskip

The direct-advice framing closely resembles standard recommendations provided by the PES. The peer framing justifies the recommendation using the behavior of similar job seekers. The awareness framing additionally highlights potential cognitive biases that may affect job-search decisions. 

Examples of the three framings are reported below:
\begin{itemize}
    \item[(i)] Direct advice  example (example from occupational intervention):

\begin{quote}
``\textit{You could broaden your job search to other sectors of activity [...]}'' 
\end{quote}

\item[(ii)] Peer wording example (example from motivational intervention):

\begin{quote}
``\textit{Half of job seekers with your level of experience, in your region, and searching for a ``lib\_metier'' position apply to jobs offering a monthly gross wage of ``montant\_salaire'' euros or more. They still have good chances of finding a job.}'' 
\end{quote}

  \item[(iii)] \textit{Raise awareness wording} example (example from motivational intervention):

\begin{quote}
``\textit{It is sometimes possible to underestimate the wages offered in available job opportunities. This can lead to choosing jobs that pay less than those one could obtain. This may also affect the wage obtained in the new job.}'' 
\end{quote}

\end{itemize}

Treatment effects are very similar across the three wording variants. Tables \ref{tab:ate_by_wording} and \ref{tab:ate_search_behavior_wording} report the corresponding estimates for occupational and motivational recommendations. We do not find statistically significant differences across framings.

\medskip

Additional analyses examining subjective beliefs, perceived returns to search effort, reservation wages, and related survey outcomes lead to the same conclusion. Across specifications, estimated effects are broadly similar across wording variants, and we do not detect systematic evidence that peer-based or awareness framings outperform direct advice. These additional results are not reported to conserve space but are available from the authors upon request.

\medskip

Effects are highly similar across framings. The content of the recommendation matters much more than its presentation.

\begin{table}[!htbp]
    \caption{Occupational treatment effect on online applications by bias group - Treatment wording}
    \label{tab:ate_by_wording}
\begin{adjustbox}{width = \textwidth}
\begin{threeparttable}
\begin{tabular}{lllllll}
            \hline
            \hline
            & \multicolumn{6}{c}{\textit{Number of applications in the two months following intervention}} \\
            & \multicolumn{3}{c}{Across streams} & \multicolumn{3}{c}{Across occupations} \\  
            \cmidrule(lr){2-4}\cmidrule(lr){5-7} \\
            & On a posted vacancy & Spontaneous & Caseworker suggestion  & Preferred & Proposed & Others \\
            \hline
            & (1) & (2) & (3) & (4) & (5) & (6) \\
            \textbf{Pessimists} & & & & & & \\
            Occ. treatment x Direct advice & 0.122         & 0.229$^{*}$   & -0.076             & 0.079             & -0.051                 & 0.095\\   
                                              & (0.375)       & (0.132)       & (0.089)            & (0.154)           & (0.122)                & (0.280)\\      
            Occ. treatment x Peers' behaviors& 0.242         & 0.312$^{***}$ & -0.115             & -0.021            & 0.078                  & 0.185\\   
                                              & (0.399)       & (0.120)       & (0.083)            & (0.148)           & (0.140)                & (0.285)\\ 
            Occ. treatment x Raise awareness & 0.100         & 0.183         & 0.004              & 0.025             & -0.063                 & 0.138\\   
                                              & (0.397)       & (0.118)       & (0.098)            & (0.154)           & (0.123)                & (0.284)\\   
            \cline{2-7}
            Control mean  & 4.95  & 0.262 & 0.538   & 1.19     & 1.01 & 2.76\\   
            & & & & & & \\
            \hline
            \textbf{Optimists} & & & & & & \\
            Occ. treatment x Direct advice & 0.333         & 0.135$^{**}$  & 0.040              & 0.066             & 0.176$^{**}$           & 0.091\\   
                                              & (0.240)       & (0.061)       & (0.047)            & (0.089)           & (0.074)                & (0.164)\\    
            Occ. treatment x Peers' behaviors & 0.546$^{**}$  & 0.113$^{*}$   & 0.084$^{*}$        & 0.134             & 0.203$^{***}$          & 0.209\\   
                                              & (0.239)       & (0.062)       & (0.045)            & (0.092)           & (0.077)                & (0.165)\\ 
            Occ. treatment x Raise awareness & 0.507$^{**}$  & 0.073         & 0.118$^{**}$       & 0.106             & 0.144$^{*}$            & 0.257\\   
                                              & (0.239)       & (0.059)       & (0.046)            & (0.092)           & (0.076)                & (0.164)\\
            \cline{2-7} 
            Control mean ($\widehat{\alpha}_O$) & 4.31 & 0.238  & 0.455& 1.01    & 0.829  & 2.47\\ 
            & & & & & & \\
            \hline
            Number Obs. & 8,210        & 8,210         & 8,210         & 8,210         & 8,210              & 8,210\\ 
            \hline
        \end{tabular}
\begin{tablenotes}[flushleft]
    \small
    \item \textit{Notes:} This table reports estimated average treatment effects of the occupational recommendation, by wording, on the number of applications submitted in the two months following the intervention, separately by bias group. Wording variations are pooled together. Several application outcomes are considered: (1) the number of applications to posted vacancies, (2) the number of spontaneous applications (i.e., applications not linked to a posted vacancy), and (3) the number of applications submitted in response to vacancies suggested by caseworkers. Applications to posted vacancies are further disaggregated by the occupation of the vacancy: (4) preferred (target) occupation, (5) suggested occupations, and (6) other occupations. The sample is restricted to job seekers who submitted at least one application on the PES platform in the month prior to the intervention. All regressions control for perceived three-month reemployment probability and the number of applications made in the month before the intervention. Application counts are winsorized at the 98th percentile of strictly positive values. Significance levels: $^{*}$p$<$0.1; $^{**}$p$<$0.05; $^{***}$p$<$0.01.
\end{tablenotes}
\end{threeparttable}
\end{adjustbox}
\caption*{\scriptsize }
\end{table}

\begin{table}[!htbp]
\centering 
\caption{Intervention impact on pessimists search parameters 1 month after the intervention by wording}
    \label{tab:ate_search_behavior_wording}
\resizebox{\textwidth}{!}{\begin{tabular}{lcccccc}
\\
\hline 
\hline \\
& \multicolumn{6}{c}{\textit{Declared search parameters}} \\
\cline{2-4} \cline{5-7}\\
& \multicolumn{3}{c}{Search effort} & \multicolumn{3}{c}{Log reservation wage} \\ 
\cmidrule(lr){2-4} \cmidrule(lr){5-7}& \makecell{Full\\ sample} & \makecell{Low eff.\\ before} & \makecell{High eff.\\ before} & \makecell{Full\\ sample} & \makecell{Low\\ wage} & \makecell{High\\ wage} \\
\hline
\multicolumn{7}{l}{\textit{Occupational Treatment Framing}}\\

x Direct advice & -0.348  & 0.274         & -1.10         & 0.011         & 0.017         & 0.017\\   
                             & (0.404)       & (0.471)       & (0.693)       & (0.011)       & (0.013)       & (0.016)\\ 
x Peers' behavior & 0.462         & 0.850         & -0.035        & 0.007         & 0.0001        & 0.008\\   
                             & (0.422)       & (0.534)       & (0.667)       & (0.012)       & (0.013)       & (0.017)\\  
x Raise awareness & 0.294         & 0.472         & 0.098         & -0.002        & -0.024        & 0.009\\   
                             & (0.400)       & (0.488)       & (0.661)       & (0.011)       & (0.015)       & (0.015)\\ 
\hline
\multicolumn{7}{l}{\textit{Motivational Treatment Framing}}\\
x Direct advice & 0.882$^{**}$  & 1.55$^{***}$  & 0.185         & 0.013         & 0.015         & 0.013\\   
                             & (0.426)       & (0.535)       & (0.679)       & (0.011)       & (0.013)       & (0.017)\\   
x Peers' behavior & 0.948$^{**}$  & 1.59$^{***}$  & 0.279         & 0.011         & 0.024$^{*}$   & 0.003\\   
                             & (0.403)       & (0.503)       & (0.648)       & (0.012)       & (0.014)       & (0.018)\\   
x Raise awareness & 0.857$^{**}$  & 1.63$^{***}$  & -0.003        & 0.034$^{***}$ & 0.033$^{**}$  & 0.034$^{**}$\\   
                             & (0.423)       & (0.543)       & (0.663)       & (0.012)       & (0.015)       & (0.016)\\
\\[-1.8ex]\hline 
Control Mean & 12.1 & 8.2 & 17.2 & 2,034\euro & 1,853\euro & 2,172\euro\\
$\beta^o = \beta^m$ & 0.006 & 0.002 & 0.271 & 0.059 & 0.008 & 0.48 \\
No difference between wordings & 0.980 & & & 0.20 & 0.49 & 0.12\\
No difference between subgroups & & 0.043 & 0.043 & & 0.651 & 0.651 \\
Number Obs. & 3,485         & 1,882         & 1,603         & 2,366         & 937           & 1,429\\  
   R$^2$                     & 0.342       & 0.111       & 0.175      & 0.712       & 0.679       & 0.702\\ 
\hline 
\end{tabular}}
\caption*{\scriptsize \textit{Notes:} This table reports estimated average treatment effects of the occupational and motivational recommendations on search parameters as reported in the follow-up survey. We first present ATEs pooling all wording variations, and then report estimates separately by wording. We further examine heterogeneity by splitting the sample according to pre-intervention search intensity (below vs. above 10 hours per week) and according to pre-intervention reservation wage (below vs. above the median reservation wage among individuals sharing the same target occupation). The sample is restricted to job seekers who completed the follow-up survey. Heteroskedasticity-robust standard errors are reported. Significance levels: $^{*}$p$<$0.1; $^{**}$p$<$0.05; $^{***}$p$<$0.01.}
\end{table}

\section{Tables and Figures}\label{sec:app-tables}

\begin{figure}[htpb]
    \centering
    \caption{Average treatment effect (occupational recommendations) on applications across bias groups, platform users only.}
    \label{fig:ATE_main}

    \begin{subfigure}{0.48\linewidth}
        \centering
        \includegraphics[width=\linewidth]{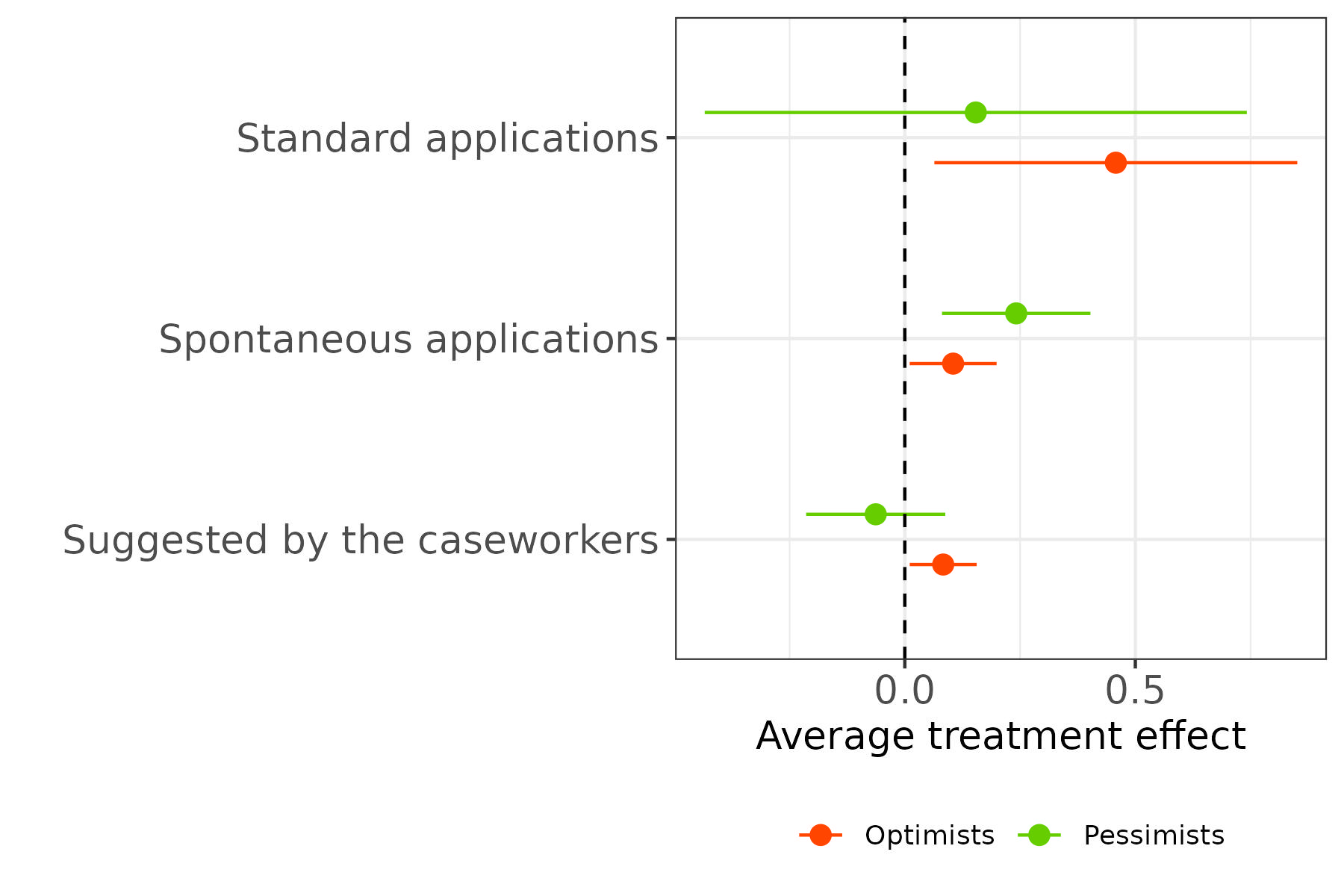}
        \caption{\scriptsize Number of applications}
        \label{fig:ATE_streams}
    \end{subfigure}
    \quad
    \begin{subfigure}{0.48\linewidth}
        \centering
        \includegraphics[width=\linewidth]{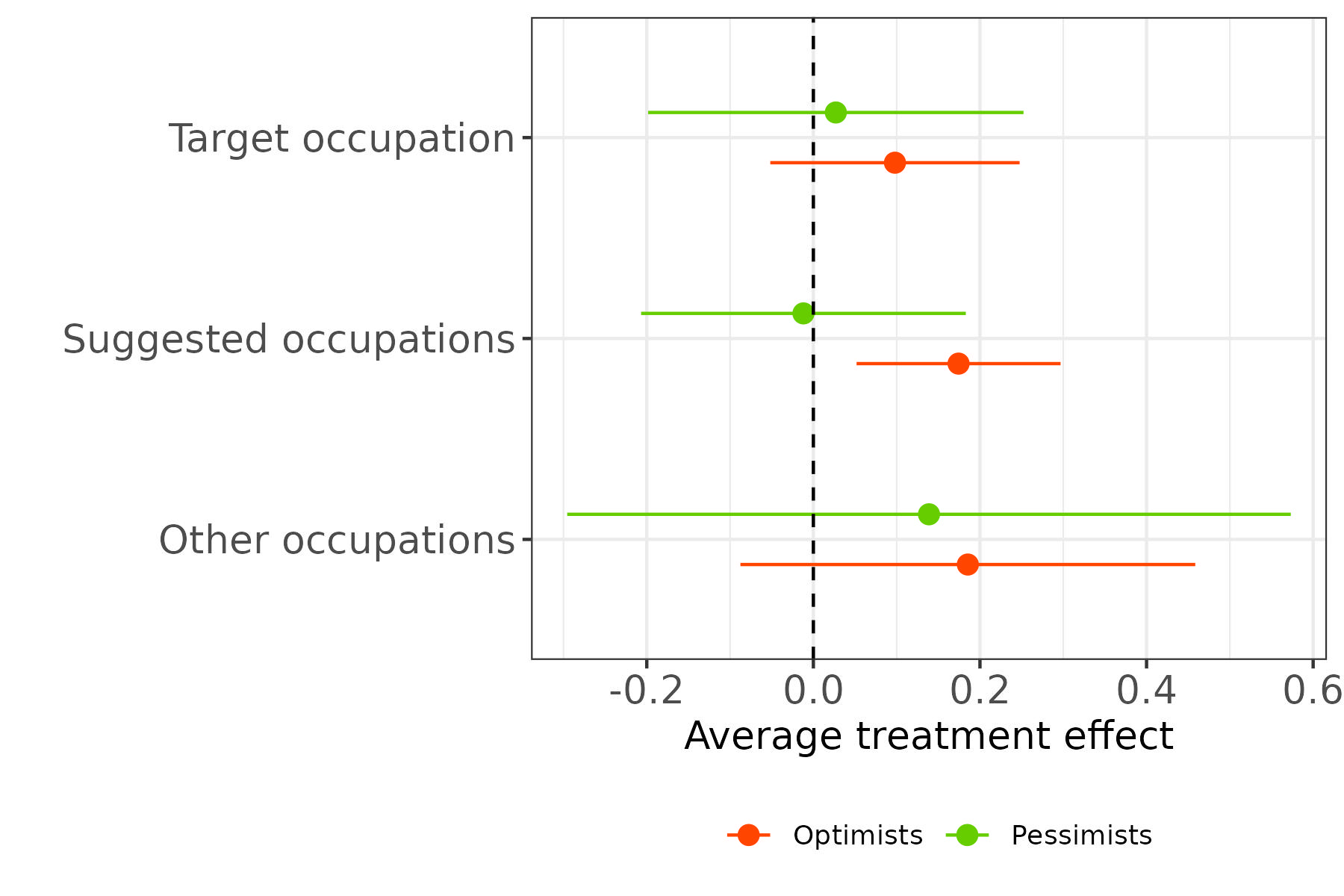}
        \caption{\scriptsize Number of applications by occupation}
        \label{fig:ATE_occupations}
    \end{subfigure}

    \caption*{\footnotesize
    Notes: This figure reports average treatment effects (ATEs) of the occupational recommendation on job applications in the two months following the intervention. Outcomes include the total number of applications, as well as a decomposition by occupation type (target, suggested, or other). ATEs are estimated separately by bias group and control for pre-intervention perceived three-month reemployment probability and the number of applications made in the month preceding the intervention. Heteroskedasticity-robust standard errors are reported. Application counts are winsorized at the 98th percentile of strictly positive values.}
\end{figure}

% ---- relocated from main text (length reduction) ----
\begin{figure}[ht!]
\centering
\includegraphics[width=0.9\textwidth]{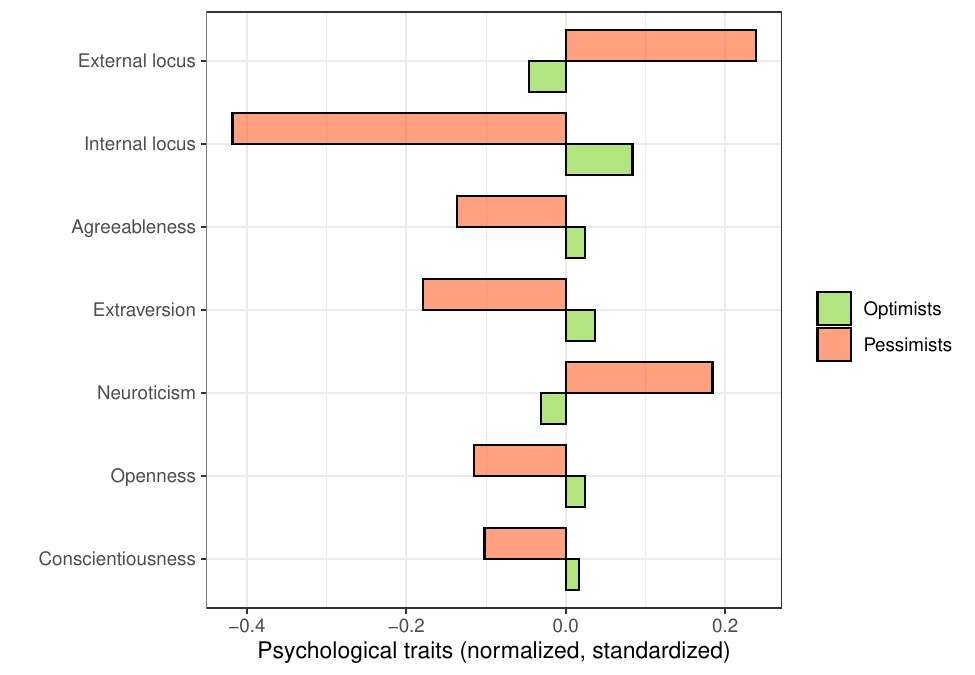}
\caption{Psychological characteristics of pessimistic job seekers}
\label{fig:psychological_profiles}
\begin{flushleft}
\footnotesize
\textit{Notes:} The figure reports standardized differences in psychological characteristics between pessimistic and non-pessimistic job seekers. All variables are standardized to have mean zero and unit variance in the estimation sample. Pessimistic job seekers are identified using the same prediction algorithm as in the main experiment. Since these psychological variables are not available in the current experimental sample, the figure is estimated using earlier survey waves of the panel for which these measures are available. To avoid overfitting, the prediction algorithm is trained on one half of the sample and predictions are computed on the other half.
\end{flushleft}
\end{figure}

\begin{figure}[htpb]
    \centering
    \caption{Initial beliefs for reemployment chances across occupations and \textit{belief narrowness}}
    \includegraphics[width=0.75\linewidth]{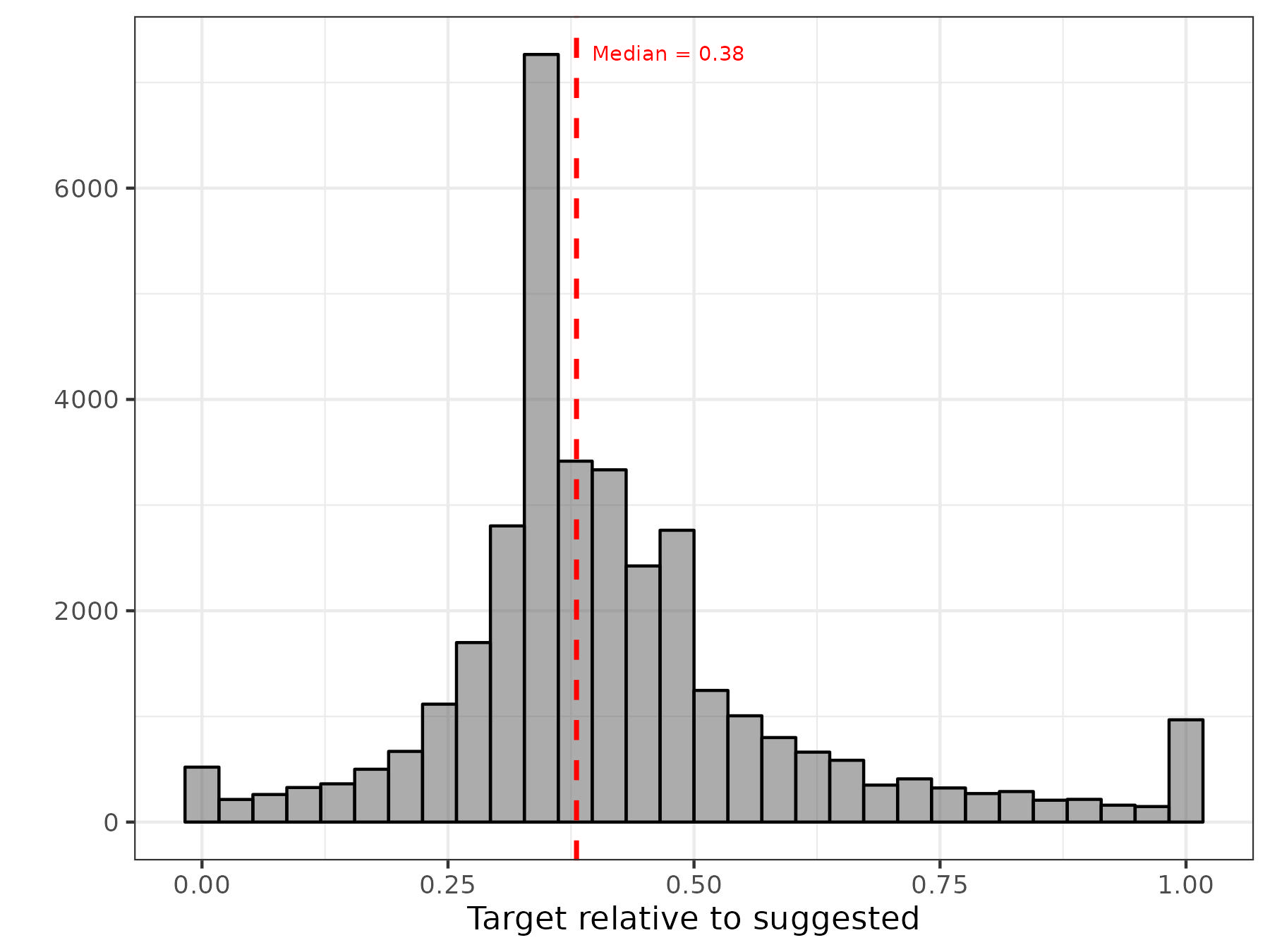}
    \label{fig:beliefs_relocc}
    \caption*{\scriptsize \textit{Notes}: The figure shows the empirical distribution of the \textit{belief narrowness}: ratio of the perceived success probability in the target occupation to the sum of perceived success probabilities across the three occupations.}
\end{figure}

\begin{table}[htbp]
\centering
\caption{Reemployment Outcomes of Optimists}
\label{tab:ate_reemployment_optimists}

\scalebox{0.75}{
\begin{tabular}{lcccccccc}
\toprule
& (1) & (2) & (3) & (4) & (5) & (6) & (7) & (8) \\
\cmidrule(lr){2-9}
Dependent Variable
& \multicolumn{2}{c}{Reemployment}
& Contract
& \multicolumn{5}{c}{Hiring Wage Relative to Reservation Wage} \\
\cmidrule(lr){2-3} \cmidrule(lr){4-4} \cmidrule(lr){5-9}
& 6 Months & 12 Months & Permanent
& Mean & $\geq 1$
& QTE 25 & QTE 50 & QTE 75 \\
\midrule

\multicolumn{9}{l}{\textbf{Panel A. Full Sample}} \\

Occ. treatment
& \begin{tabular}[c]{@{}c@{}}-0.001\\(0.006)\end{tabular}
& \begin{tabular}[c]{@{}c@{}}-0.001\\(0.006)\end{tabular}
& \begin{tabular}[c]{@{}c@{}}0.000\\(0.007)\end{tabular}
& \begin{tabular}[c]{@{}c@{}}0.007\\(0.005)\end{tabular}
& \begin{tabular}[c]{@{}c@{}}0.008\\(0.011)\end{tabular}
& \begin{tabular}[c]{@{}c@{}}0.001\\(0.005)\end{tabular}
& \begin{tabular}[c]{@{}c@{}}0.003\\(0.004)\end{tabular}
& \begin{tabular}[c]{@{}c@{}}0.005\\(0.007)\end{tabular}
\\

\midrule

\multicolumn{9}{l}{\textbf{Panel B. Heterogeneity by Initial Belief Dispersion}} \\

Occ. treatment $\times$ Wide beliefs
& \begin{tabular}[c]{@{}c@{}}0.005\\(0.011)\end{tabular}
& \begin{tabular}[c]{@{}c@{}}0.003\\(0.011)\end{tabular}
& \begin{tabular}[c]{@{}c@{}}-0.018\\(0.012)\end{tabular}
& \begin{tabular}[c]{@{}c@{}}0.004\\(0.007)\end{tabular}
& \begin{tabular}[c]{@{}c@{}}0.011\\(0.016)\end{tabular}
& \begin{tabular}[c]{@{}c@{}}0.004\\(0.009)\end{tabular}
& \begin{tabular}[c]{@{}c@{}}0.005\\(0.005)\end{tabular}
& \begin{tabular}[c]{@{}c@{}}-0.003\\(0.010)\end{tabular}
\\[1.5ex]

Occ. treatment $\times$ Narrow beliefs
& \begin{tabular}[c]{@{}c@{}}-0.010\\(0.011)\end{tabular}
& \begin{tabular}[c]{@{}c@{}}-0.007\\(0.011)\end{tabular}
& \begin{tabular}[c]{@{}c@{}}-0.007\\(0.013)\end{tabular}
& \begin{tabular}[c]{@{}c@{}}0.015**\\(0.007)\end{tabular}
& \begin{tabular}[c]{@{}c@{}}0.025\\(0.016)\end{tabular}
& \begin{tabular}[c]{@{}c@{}}0.006\\(0.008)\end{tabular}
& \begin{tabular}[c]{@{}c@{}}0.006\\(0.004)\end{tabular}
& \begin{tabular}[c]{@{}c@{}}0.017*\\(0.010)\end{tabular}
\\

\midrule

Target Occupation FE
& \checkmark & \checkmark & \checkmark &  &  &  &  &  \\

Initial Reemployment Beliefs
& \checkmark & \checkmark & \checkmark &  &  &  &  &  \\

Initial Wage Beliefs
&  &  &  & \checkmark & \checkmark & \checkmark & \checkmark & \checkmark \\

Applications Prior to Intervention
& \checkmark & \checkmark & \checkmark & \checkmark & \checkmark & \checkmark & \checkmark & \checkmark \\

\midrule

Control Mean
& 0.415 & 0.569 & 0.215 & 1.050 & 0.558 & 0.950 & 1.030 & 1.170 \\

Observations
& 35,388 & 35,388 & 20,038 & 14,167 & 14,167 & 14,167 & 14,167 & 14,167 \\

\bottomrule
\end{tabular}
}

\vspace{0.75em}

\parbox{\textwidth}{\footnotesize
\textit{Notes:} ``Occ.'' denotes occupation. This table reports treatment effects on reemployment outcomes for job seekers classified as optimists. Each column corresponds to a different outcome. Permanent contract status and hiring wages are observed conditional on reemployment within twelve months. Hiring wages are expressed relative to the reservation wage. Hourly wages are winsorized between \euro~11.52 and \euro~26, while reservation wages are winsorized between \euro~1,400 and \euro~5,000. Panel A reports treatment effects for the full sample. Panel B reports heterogeneity by initial belief dispersion, where wide and narrow beliefs are defined relative to the median perceived reemployment probability. Standard errors clustered at the job-seeker level are reported in parentheses. $^{*}p<0.10$, $^{**}p<0.05$, and $^{***}p<0.01$.
}

\end{table}

\begin{figure}[htpb]
    \centering
    \caption{Initial beliefs for reemployment chances across occupations and \textit{belief narrowness}}
    \includegraphics[width=0.75\linewidth]{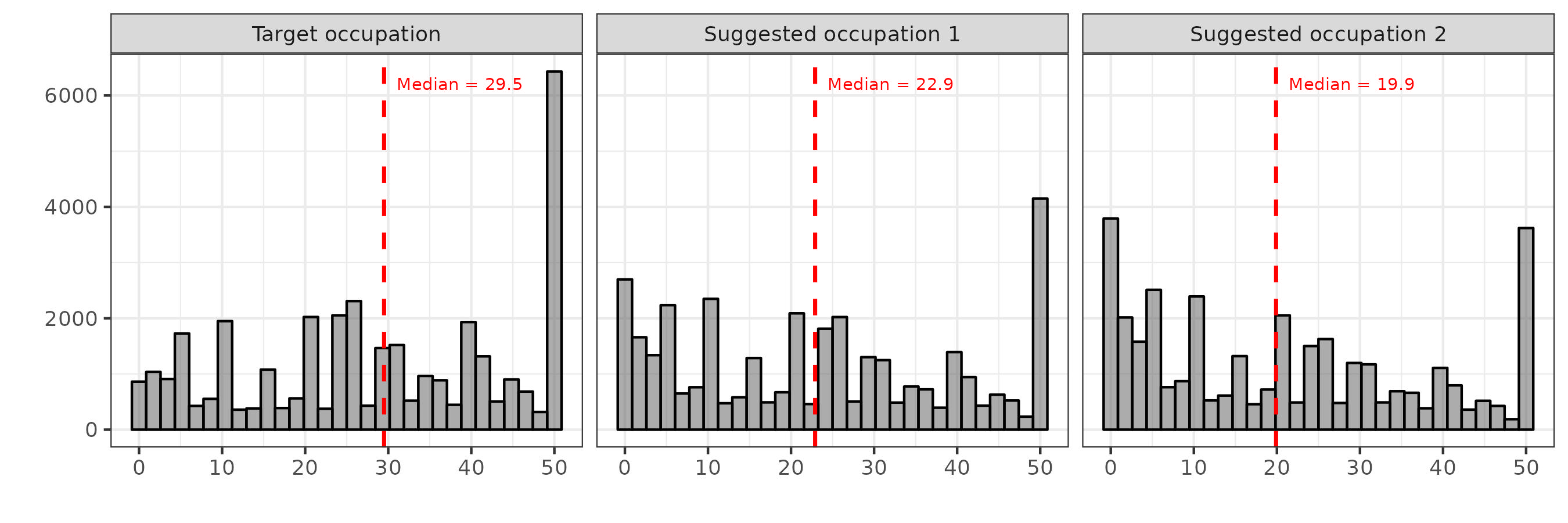}
    \label{fig:beliefs_raw}
    \caption*{\scriptsize \textit{Notes}: the three panels display the empirical distributions of job seekers’ perceived probabilities of application success across three occupations. The first panel corresponds to the target occupation, while the second and third panels refer to two alternative occupations selected for their proximity to the target occupation. }
\end{figure}

\begin{table}[!htbp]
    \caption{Occupational treatment effect on online applications by bias group - full sample}
    \label{tab:ate_full_sample}
\begin{adjustbox}{width = \textwidth}
\begin{threeparttable}
\begin{tabular}{lllllll}
            \hline
            \hline
            & \multicolumn{6}{c}{\textit{Number of applications in the two months following intervention}} \\
            & \multicolumn{3}{c}{Across streams} & \multicolumn{3}{c}{Across occupations} \\  
           \cmidrule(lr){2-4}\cmidrule(lr){5-7} \\
            & On a posted vacancy & Spontaneous & Caseworker suggestion  & Preferred & Proposed & Others \\
            \hline
            & (1) & (2) & (3) & (4) & (5) & (6) \\
            \textbf{Pessimists} & & & & & & \\
            Occupational treatment ($\widehat{\beta}_P$) & 0.019 & 0.067$^{***}$ & 0.006 & -0.0005 & 0.0002 & 0.019\\   
                                      & (0.069) & (0.020) & (0.021) & (0.027) & (0.022)                & (0.050)\\
            \cline{2-7}
            Control mean  ($\widehat{\alpha}_P$) & 1.39 & 0.088 & 0.299 & 0.339 & 0.271 & 0.781 \\
            & & & & & & \\
            \hline
            \textbf{Optimists} & & & & & & \\
            Occupational treatment ($\widehat{\beta}_O$) & 0.081$^{*}$   & 0.012 & 0.019 & 0.027$^{*}$       & 0.027$^{**}$           & 0.027\\   
                                      & (0.044) & (0.011) & (0.014) & (0.016)           & (0.014)                & (0.030)\\
            \cline{2-7} 
            Control mean ($\widehat{\alpha}_O$) & 1.22 & 0.086 & 0.272 & 0.285 & 0.232 & 0.699\\   
            & & & & & & \\
            \hline
            $\beta_P = \beta_O$ pvalue  & 0.45 & 0.02 & 0.59 & 0.37 & 0.30 & 0.90 \\
            Number Obs. & 45,612        & 45,612        & 45,612                 & 45,612        & 45,612             & 45,612\\
            \hline
        \end{tabular}
\begin{tablenotes}[flushleft]
    \small
    \item \textit{Notes:} This table reports estimated average treatment effects of the occupational recommendation on the number of applications submitted in the two months following the intervention, separately by bias group. Wording variations are pooled together. All individuals participating to the initial survey are included. Several application outcomes are considered: (1) the number of applications to posted vacancies, (2) the number of spontaneous applications (i.e., applications not linked to a posted vacancy), and (3) the number of applications submitted in response to vacancies suggested by caseworkers. Applications to posted vacancies are further disaggregated by the occupation of the vacancy: (4) preferred (target) occupation, (5) suggested occupations, and (6) other occupations. All regressions control for perceived three-month reemployment probability and the number of applications made in the month before the intervention. Application counts are winsorized at the 98th percentile of strictly positive values. P-values of Fisher tests for equal effect of the treatment are also reported. Significance levels: $^{*}$p$<$0.1; $^{**}$p$<$0.05; $^{***}$p$<$0.01.
\end{tablenotes}
\end{threeparttable}
\end{adjustbox}
\end{table}

\begin{table}[!htbp]
\caption{Applications and belief narrowness measures}
\label{tab:desc_model}
\centering
\begin{threeparttable}

\begin{tabular}{lllll}
            \hline
            \hline
            & \multicolumn{4}{c}{\textit{Nb. of app. in the month before the intervention}} \\
            \cline{2-5}
            & To a vacancy &  \multicolumn{3}{c}{By occupation } \\
            & & Target & Suggested & Others \\ 
            \hline
            
            Subjective beliefs narrowness & -0.439$^{***}$ & 0.009 & -0.157$^{***}$ & -0.288$^{***}$\\   
                                  & (0.134) & (0.057) & (0.033) & (0.087)\\   
            \hline
            Number obs. & 35,282         & 35,282         & 35,282             & 35,282\\  
            \hline 
            \hline
        \end{tabular}
\begin{tablenotes}[flushleft]
    \small
    \item \textit{Notes:} This table reports estimates of the correlation between the beliefs narrowness measure and the number of applications made in the month prior to the intervention. We first consider the total number of applications to posted vacancies, and then decompose this total by occupation type. All regressions control for gender, age, unemployment duration, education level and include fixed effects for the job seeker’s target occupation. The sample is restricted to individuals provided valid responses to the relevant items in the initial survey. Heteroskedasticity-robust standard errors are reported. Significance levels: $^{*}$p$<$0.1; $^{**}$p$<$0.05; $^{***}$p$<$0.01.
\end{tablenotes}
\end{threeparttable}
\end{table}

\begin{table}[htbp]
\centering
\caption{Occupational Treatment Effects on Beliefs: Heterogeneity by Initial Perceived Outside Options}
\label{tab:ate_by_candidature_avant_full}

\scalebox{0.75}{
\begin{tabular}{lcccc}
\toprule
& \multicolumn{4}{c}{Subjective Beliefs Regarding Application Success} \\
\cmidrule(lr){2-5}
& Target Occ. & Suggested Occ. 1 & Suggested Occ. 2 & Target vs.\ Suggested \\
& $p_{t,post}$ & $p_{1,post}$ & $p_{2,post}$ & $h_{post}$ \\
\midrule

\multicolumn{5}{l}{\textbf{Optimists}} \\

High perceived outside options
& \begin{tabular}[c]{@{}c@{}}0.613\\(0.763)\end{tabular}
& \begin{tabular}[c]{@{}c@{}}1.010\\(0.743)\end{tabular}
& \begin{tabular}[c]{@{}c@{}}2.170***\\(0.757)\end{tabular}
& \begin{tabular}[c]{@{}c@{}}-0.016*\\(0.009)\end{tabular}
\\[1.5ex]

Occ. treatment
& \begin{tabular}[c]{@{}c@{}}0.696\\(0.556)\end{tabular}
& \begin{tabular}[c]{@{}c@{}}0.167\\(0.531)\end{tabular}
& \begin{tabular}[c]{@{}c@{}}1.070**\\(0.534)\end{tabular}
& \begin{tabular}[c]{@{}c@{}}-0.007\\(0.007)\end{tabular}
\\[1.5ex]

High perceived outside options $\times$ Occ. treatment
& \begin{tabular}[c]{@{}c@{}}-0.934\\(0.763)\end{tabular}
& \begin{tabular}[c]{@{}c@{}}-0.801\\(0.752)\end{tabular}
& \begin{tabular}[c]{@{}c@{}}-1.980***\\(0.766)\end{tabular}
& \begin{tabular}[c]{@{}c@{}}0.008\\(0.009)\end{tabular}
\\

\midrule

Control Mean
& 27.7 & 22.7 & 19.9 & 0.424 \\

\midrule

\multicolumn{5}{l}{\textbf{Pessimists}} \\

High perceived outside options
& \begin{tabular}[c]{@{}c@{}}0.503\\(1.140)\end{tabular}
& \begin{tabular}[c]{@{}c@{}}1.800*\\(0.996)\end{tabular}
& \begin{tabular}[c]{@{}c@{}}-0.227\\(1.050)\end{tabular}
& \begin{tabular}[c]{@{}c@{}}-0.063***\\(0.017)\end{tabular}
\\[1.5ex]

Occ. treatment
& \begin{tabular}[c]{@{}c@{}}-0.157\\(0.614)\end{tabular}
& \begin{tabular}[c]{@{}c@{}}-0.016\\(0.561)\end{tabular}
& \begin{tabular}[c]{@{}c@{}}-0.373\\(0.577)\end{tabular}
& \begin{tabular}[c]{@{}c@{}}0.001\\(0.009)\end{tabular}
\\[1.5ex]

High perceived outside options $\times$ Occ. treatment
& \begin{tabular}[c]{@{}c@{}}-1.140\\(1.260)\end{tabular}
& \begin{tabular}[c]{@{}c@{}}-1.130\\(1.150)\end{tabular}
& \begin{tabular}[c]{@{}c@{}}0.864\\(1.200)\end{tabular}
& \begin{tabular}[c]{@{}c@{}}-0.008\\(0.018)\end{tabular}
\\

\midrule

Control Mean
& 17.9 & 13.6 & 13.7 & 0.444 \\

\midrule

Observations
& 8,573 & 8,386 & 8,154 & 7,786 \\

\bottomrule
\end{tabular}
}

\vspace{0.75em}

\parbox{\textwidth}{\footnotesize
\textit{Notes:} ``Occ.'' denotes occupation. This table reports treatment effects of the occupational recommendation on follow-up beliefs, allowing for heterogeneity by initial perceived outside options. Outcomes are the perceived probability of success in the target occupation ($p_{t,post}$), in the first suggested occupation ($p_{1,post}$), in the second suggested occupation ($p_{2,post}$), and the relative weight placed on the target occupation ($h_{post}$). All specifications control for baseline beliefs. The sample is restricted to individuals observed in the follow-up survey. Robust standard errors are reported in parentheses. $^{*}p<0.10$, $^{**}p<0.05$, and $^{***}p<0.01$.
}

\end{table}

\begin{table}[htbp]
\centering
\caption{Predictions, tests, and evidence}
\label{tab:pred_map}
\footnotesize
\renewcommand{\arraystretch}{1.25}
\setlength{\tabcolsep}{4pt}
\begin{adjustbox}{max width=\textwidth}
\begin{tabular}{@{}p{2.6cm}p{1.4cm}p{3.9cm}p{3.0cm}p{2.4cm}@{}}
\toprule
Prediction & Channel & Empirical test & Finding & Evidence \\
\midrule
P.1 No belief updating (occ.) & AA vs.\ BU & Do post-treatment beliefs about success across occupations move? & No detectable shift, including $h_{\text{post}}$ & Tables~\ref{tab:ate_by_candidature_avant_full}, \ref{tab:ols_beliefs_occ_all}, \ref{tab:dml_beliefs_selected} \\
P.2 Heterogeneity by belief narrowness (occ.) & AA vs.\ BU & Are effects concentrated among workers with dispersed priors (high outside options)? & Yes for optimists; absent for narrow-belief workers & Table~\ref{tab:ate_by_initial_beliefs} \\
P.3 No wording effect on beliefs (occ.) & AA vs.\ BU & Does the raise-awareness wording revise beliefs more than direct advice? & No & Appendix~\ref{sec:app-wordings} \\
P.4 Threshold heterogeneity in $w^*$ (mot.) & AS & Is the reservation-wage effect confined below the peer benchmark? & Yes: $+2.4\%$ below, $\approx\!0$ above & Table~\ref{tab:ate_search_behavior} \\
P.5 No shift in $\tilde F$ beliefs (mot.) & AS vs.\ BU & Does the perceived wage-offer distribution move? & No & Table~\ref{tab:ols_beliefs} \\
P.6 Effort gain among low-effort workers (mot.) & AS & Is the effort effect concentrated below the $10$h norm? & Yes: $\approx\!1.6$h below, $\approx\!0$ above & Table~\ref{tab:ate_search_behavior} \\
P.7 Upward reallocation of applications (mot.) & AS & Does application composition shift toward higher-wage vacancies? & Yes: count flat, quality up & Table~\ref{tab:ate_applications_pessimists_id_avg} \\
P.8 Ambiguous employment effects (mot.) & --- & What is the net reemployment effect? & $+2.2$pp (low-effort $+2.9$; job quality up) & Table~\ref{tab:ate_reemployment_pessimists} \\
\bottomrule
\end{tabular}
\end{adjustbox}

\vspace{0.5em}
\parbox{\textwidth}{\footnotesize \textit{Notes:} ``occ.'' and ``mot.'' denote the occupational and motivational arms. BU, AA, and AS denote the Bayesian-updating, attention-activation, and aspiration-shift channels of Section~\ref{subsec:mechanisms}. Each row states a prediction, the channel it discriminates, the test, the finding, and the evidence. Belief responses are read as suggestive rather than dispositive (see text).}
\end{table}

\clearpage

\section{Experimental Design: Additional Tables and Figures}\label{sec:app-design}

\begin{figure}[!htbp]
    \centering
    \includegraphics[width=\linewidth]{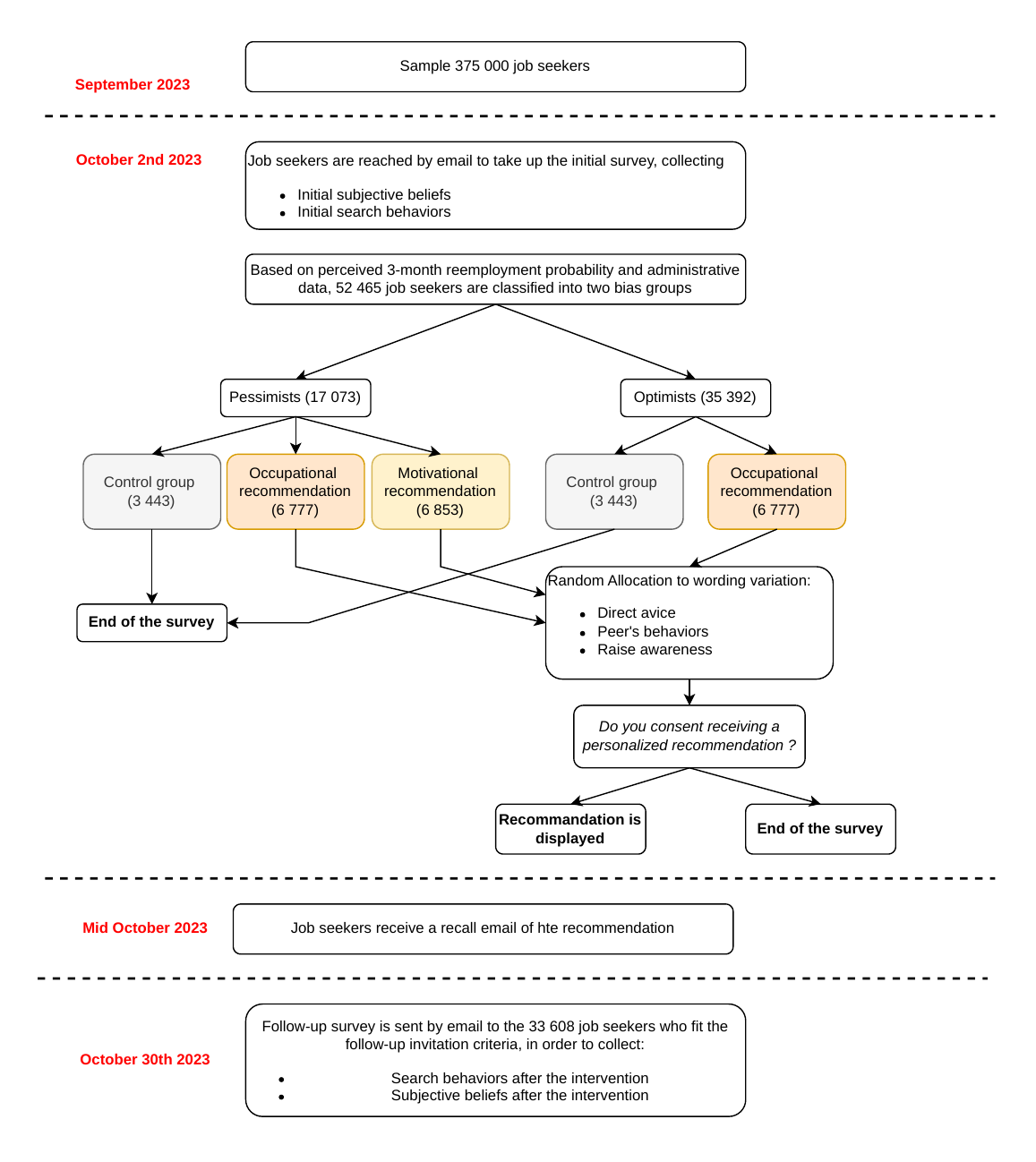}
    \caption{RCT procedure}
    \label{fig:rct_schema}
\end{figure}

\begin{table}[!htbp]
    \caption{List of variables and the databases used to obtain them}
    \label{tab:data_source}
    \centering
    \resizebox{\textwidth}{!}{\begin{tabular}{lll}
        Variable & Database used & Origin  \\
        \hline
        Demographics  & \textit{Fichiers historiques (FH)} & French PES \\
        Initial search parameters & \textit{Fichiers historiques (FH)} & French PES \\
        \multirow{2}*{3-month reemployment indicator}  & \textit{Indices de retour à l'emploi (IRE)} & \multirow{2}*{French PES} \\
          & \textit{Déclarations préalables à l'embauche (DPAE)} & \\
        Subjective beliefs before the intervention & Initial Survey & Self collected \\
        Declared search parameters before the intervention & Initial Survey & Self collected \\
        Applications made on the PES platform & \textit{Actes métiers d'intermédiation (AMI)} & French PES \\
        Hires after the intervention & \textit{Déclaration sociales nominatives (DSN)}& French administration \\
        Search behaviors/beliefs after the intervention & Follow-up survey & Self collected \\
        \hline
    \end{tabular}}
    
\end{table}

\begin{table}[!htbp]
\centering
\caption{Sample Size, Take-up, and Follow-up Sample}
\label{tab:design}
\scalebox{0.82}{
\begin{threeparttable}
\begin{tabular}{lccccc}
\toprule
& \multicolumn{2}{c}{Control} & \multicolumn{2}{c}{Treatment} & Total \\
\cmidrule(lr){2-3} \cmidrule(lr){4-5}
& Pessimists & Optimists & Pessimists & Optimists & Total \\
\midrule

\textit{Panel A. Survey invitations and enrollment} \\

Received main survey & \multicolumn{2}{c}{75,208} & \multicolumn{2}{c}{299,792} & 375,000 \\
Opened main survey & \multicolumn{2}{c}{14,784} & \multicolumn{2}{c}{58,747} & 73,531 \\
Answered first two questions & \multicolumn{2}{c}{10,491} & \multicolumn{2}{c}{41,974} & 52,465 \\
Proposed a recommendation & \multicolumn{2}{c}{--} & \multicolumn{2}{c}{34,391} & 34,391 \\

\addlinespace
\textit{Panel B. Experimental sample} \\

Enrolled in experiment & 3,443 & 7,048 & 13,630 & 28,344 & 52,465 \\
Motivational treatment & -- & -- & 6,853 & -- & 6,853 \\
Occupational treatment & -- & -- & 6,777 & 28,344 & 35,121 \\
Proposed a recommendation & -- & -- & 10,966 & 23,425 & 34,391 \\
Accepted recommendation & -- & -- & 8,643 & 18,695 & 27,338 \\
Take-up rate & -- & -- & 0.634 & 0.660 & -- \\

\addlinespace
\textit{Panel C. Follow-up sample} \\

Invited to follow-up & 1,948 & 4,322 & 8,643 (6,283) & 18,695 (14,262) & 33,608 \\
\quad Share among enrolled & 0.566 & 0.613 & 0.634 (0.461) & 0.659 (0.503) & 0.641 \\

Participated in follow-up & 915 & 2,012 & 3,797 (2,874) & 8,071 (6,418) & 14,795 \\
\quad Share among enrolled & 0.266 & 0.285 & 0.279 & 0.285 & 0.282 \\
\quad Share among follow-up invitees & 0.470 & 0.465 & 0.439 & 0.432 & 0.440 \\

Usable observations & 782 & 1,722 & 3,229 (2,470) & 6,915 (5,552) & 12,648 \\
\quad Share among enrolled & 0.227 & 0.244 & 0.237 (0.181) & 0.244 (0.196) & 0.241 \\
\quad Share among follow-up invitees & 0.401 & 0.398 & 0.374 & 0.370 & 0.376 \\

\bottomrule
\end{tabular}

\begin{tablenotes}[flushleft]
\footnotesize
\item \textit{Notes:} This table reports sample sizes and participation rates at each stage of the experimental design. Panel A reports the number of individuals who received and opened the main survey, answered the first two questions, and were proposed a recommendation. Panel B reports the enrolled experimental sample by treatment status and bias group, as well as recommendation take-up. Panel C reports the construction of the follow-up sample. Shares are computed relative to the enrolled sample or to the follow-up invitee sample, as indicated. Values in parentheses report the corresponding counts or shares among individuals who were proposed a recommendation. Dashes indicate cells that are not applicable by design.
\end{tablenotes}

\end{threeparttable}
}
\end{table}

\begin{table}[htbp]
\centering
\caption{Differences in Demographic Characteristics by Bias Group (Control Group)}
\label{tab:bias_grp_diff}
\begin{threeparttable}

\begin{tabular}{lcc}
\toprule
 & \textbf{Pessimists} & \textbf{Optimists} \\
\midrule

\textit{Panel A. Demographics} \\

Female & 0.67 & 0.61 \\
Age (mean) & 46.7 & 44.4 \\
Age $\leq$ 25 & 0.05 & 0.06 \\
Age 25--35 & 0.14 & 0.19 \\
Age 35--45 & 0.19 & 0.22 \\
Age 45--55 & 0.31 & 0.30 \\
Age $\geq$ 55 & 0.31 & 0.23 \\

\addlinespace
\textit{Panel B. Education} \\

High school & 0.22 & 0.23 \\
HS +2 years & 0.17 & 0.16 \\
Bachelor & 0.09 & 0.10 \\
Master & 0.08 & 0.08 \\
Vocational & 0.31 & 0.31 \\
Other & 0.13 & 0.13 \\

\addlinespace
\textit{Panel C. Unemployment Duration} \\

$\leq$ 6 months & 0.19 & 0.23 \\
6--12 months & 0.21 & 0.25 \\
12--24 months & 0.25 & 0.24 \\
$\geq$ 24 months & 0.35 & 0.28 \\

\addlinespace
\textit{Panel D. Household Characteristics} \\

Married & 0.47 & 0.48 \\
Has children & 0.46 & 0.49 \\

\bottomrule
\end{tabular}

\begin{tablenotes}[flushleft]
\footnotesize
\item \textit{Notes:} This table reports mean baseline characteristics by bias group. The sample is restricted to individuals in the control group. All variables are expressed as proportions, except for age, which is reported in years.
\end{tablenotes}

\end{threeparttable}
\end{table}

\begin{sidewaystable}[ht]
\caption{Balance Tests}
\label{tab:balance}
\centering
\scalebox{0.7}{
\begin{threeparttable}
\begin{tabular}{lccccccccccc}
\toprule
& \multicolumn{7}{c}{Pessimists} & \multicolumn{4}{c}{Optimists} \\
\cmidrule(lr){2-8} \cmidrule(lr){9-12}
& Ctrl. & \multicolumn{3}{c}{Motivational} & \multicolumn{3}{c}{Occupational}
& Ctrl. & \multicolumn{3}{c}{Occupational} \\
\cmidrule(lr){3-5} \cmidrule(lr){6-8} \cmidrule(lr){10-12}
& & Advice & Peers & Awareness & Advice & Peers & Awareness
& & Advice & Peers & Awareness \\
\midrule

\textit{Panel A. Demographics} \\

Female (\%) & 0.67 & 0.677 (0.46) & 0.652 (0.20) & 0.660 (0.49) & 0.648 (0.11) & 0.661 (0.59) & 0.661 (0.55) & 0.61 & 0.617 (0.66) & 0.620 (0.39) & 0.614 (0.96) \\

Age (mean) & 46.70 & 46.8 (0.82) & 47.0 (0.29) & 47.2 (0.08) & 46.9 (0.48) & 47.1 (0.21) & 47.0 (0.32) & 44.40 & 44.6 (0.42) & 44.6 (0.55) & 44.5 (0.69) \\

Age $\leq$ 25 & 0.05 & 0.047 (0.38) & 0.053 (0.46) & 0.047 (0.11) & 0.047 (0.49) & 0.038 (0.23) & 0.052 (0.43) & 0.06 & 0.061 (0.68) & 0.059 (0.67) & 0.059 (0.61) \\
Age 25--35 & 0.14 & 0.152 (0.38) & 0.130 (0.46) & 0.124 (0.11) & 0.140 (0.49) & 0.144 (0.23) & 0.129 (0.43) & 0.19 & 0.179 (0.68) & 0.180 (0.67) & 0.182 (0.61) \\
Age 35--45 & 0.19 & 0.175 (0.38) & 0.179 (0.46) & 0.189 (0.11) & 0.180 (0.49) & 0.192 (0.23) & 0.186 (0.43) & 0.23 & 0.232 (0.68) & 0.230 (0.67) & 0.236 (0.61) \\
Age 45--55 & 0.30 & 0.321 (0.38) & 0.321 (0.46) & 0.335 (0.11) & 0.327 (0.49) & 0.321 (0.23) & 0.327 (0.43) & 0.30 & 0.298 (0.68) & 0.305 (0.67) & 0.298 (0.61) \\
Age $\geq$ 55 & 0.31 & 0.305 (0.38) & 0.316 (0.46) & 0.305 (0.11) & 0.306 (0.49) & 0.305 (0.23) & 0.306 (0.43) & 0.23 & 0.230 (0.68) & 0.225 (0.67) & 0.226 (0.61) \\

\addlinespace
\textit{Panel B. Education} \\

High school & 0.22 & 0.236 (0.84) & 0.245 (0.63) & 0.237 (0.05) & 0.249 (0.05) & 0.236 (0.71) & 0.255 (0.05) & 0.23 & 0.234 (0.96) & 0.240 (0.21) & 0.232 (0.90) \\
HS +2 years & 0.17 & 0.165 (0.84) & 0.167 (0.63) & 0.148 (0.05) & 0.149 (0.05) & 0.169 (0.71) & 0.162 (0.05) & 0.16 & 0.160 (0.96) & 0.157 (0.21) & 0.163 (0.90) \\
Bachelor & 0.09 & 0.096 (0.84) & 0.094 (0.63) & 0.107 (0.05) & 0.081 (0.05) & 0.098 (0.71) & 0.074 (0.05) & 0.10 & 0.098 (0.96) & 0.088 (0.21) & 0.093 (0.90) \\
Master & 0.08 & 0.078 (0.84) & 0.074 (0.63) & 0.077 (0.05) & 0.086 (0.05) & 0.067 (0.71) & 0.068 (0.05) & 0.07 & 0.078 (0.96) & 0.084 (0.21) & 0.080 (0.90) \\
Vocational & 0.31 & 0.298 (0.84) & 0.292 (0.63) & 0.314 (0.05) & 0.297 (0.05) & 0.300 (0.71) & 0.312 (0.05) & 0.31 & 0.300 (0.96) & 0.297 (0.21) & 0.302 (0.90) \\
Other & 0.13 & 0.128 (0.84) & 0.128 (0.63) & 0.117 (0.05) & 0.138 (0.05) & 0.130 (0.71) & 0.129 (0.05) & 0.13 & 0.129 (0.96) & 0.134 (0.21) & 0.130 (0.90) \\

\addlinespace
\textit{Panel C. Labor Market Character.} \\

Children & 0.46 & 0.463 (0.67) & 0.450 (0.55) & 0.478 (0.14) & 0.466 (0.53) & 0.457 (0.94) & 0.457 (0.93) & 0.49 & 0.498 (0.24) & 0.499 (0.18) & 0.496 (0.34) \\
Married & 0.47 & 0.466 (0.72) & 0.478 (0.56) & 0.475 (0.75) & 0.466 (0.71) & 0.467 (0.76) & 0.464 (0.62) & 0.48 & 0.477 (0.58) & 0.483 (0.81) & 0.482 (0.90) \\

\addlinespace
Search effort & 11.40 & 11.7 (0.14) & 11.7 (0.22) & 11.6 (0.30) & 11.4 (0.86) & 11.8 (0.06) & 11.4 (0.98) & 13.80 & 13.7 (0.50) & 13.8 (0.99) & 13.8 (0.81) \\

\addlinespace
\textit{Panel D. Beliefs} \\

3-month reemployment prob. & 15.80 & 15.7 (0.85) & 16.2 (0.17) & 15.7 (0.77) & 16.0 (0.44) & 15.6 (0.55) & 15.6 (0.63) & 59.00 & 59.3 (0.35) & 59.1 (0.69) & 59.8 (0.02) \\
6-month reemployment prob. & 26.50 & 25.6 (0.08) & 26.5 (0.91) & 25.8 (0.17) & 26.3 (0.69) & 25.8 (0.18) & 25.3 (0.02) & 67.60 & 67.7 (0.67) & 67.8 (0.54) & 68.1 (0.15) \\

\midrule
Sample size & 3,443 & 2,352 & 2,242 & 2,259 & 2,289 & 2,265 & 2,223 & 70,748 & 9,425 & 9,408 & 9,511 \\
\bottomrule
\end{tabular}

\begin{tablenotes}
\footnotesize
\item \textit{Notes:} This table reports baseline characteristics by treatment arm and belief group. Values are means. P-values from tests of equality with the control group are reported in parentheses. For binary and continuous variables, p-values are based on two-sided t-tests; for categorical variables, they are based on chi-squared tests. “Motivational” and “Occupational” denote treatment types. “Advice” refers to direct recommendations, “Peers” to information on peers’ behavior, and “Awareness” to informational nudges. “Prob.” denotes probability.
\end{tablenotes}

\end{threeparttable}
}
\end{sidewaystable}

\begin{sidewaystable}[ht]
\caption{Balance Tests: Follow-up Subsample}
\label{tab:balance_follow}
\centering
\scalebox{0.7}{
\begin{threeparttable}
\begin{tabular}{lccccccccccc}
\toprule
& \multicolumn{7}{c}{Pessimists} & \multicolumn{4}{c}{Optimists} \\
\cmidrule(lr){2-8} \cmidrule(lr){9-12}
& Ctrl. & \multicolumn{3}{c}{Motivational} & \multicolumn{3}{c}{Occupational}
& Ctrl. & \multicolumn{3}{c}{Occupational} \\
\cmidrule(lr){3-5} \cmidrule(lr){6-8} \cmidrule(lr){10-12}
& & Advice & Peers & Awareness & Advice & Peers & Awareness
& & Advice & Peers & Awareness \\
\midrule

\textit{Panel A. Demographics} \\

Female (\%) & 0.62 & 0.631 (0.58) & 0.599 (0.50) & 0.599 (0.51) & 0.612 (0.88) & 0.647 (0.25) & 0.658 (0.12) & 0.57 & 0.577 (0.87) & 0.589 (0.35) & 0.580 (0.75) \\

Age (mean) & 49.00 & 49.1 (0.89) & 49.9 (0.14) & 50.0 (0.09) & 49.5 (0.39) & 49.8 (0.17) & 49.3 (0.70) & 47.40 & 47.0 (0.18) & 46.9 (0.08) & 47.1 (0.32) \\

Age $\leq$ 25 & 0.04 & 0.028 (0.58) & 0.026 (0.81) & 0.020 (0.54) & 0.025 (0.90) & 0.019 (0.57) & 0.036 (0.96) & 0.03 & 0.034 (0.80) & 0.038 (0.49) & 0.028 (0.91) \\
Age 25--35 & 0.114 & 0.076 (0.81) & 0.074 (0.54) & 0.088 (0.90) & 0.096 (0.57) & 0.093 (0.96) & 0.130 & 0.13 & 0.134 (0.49) & 0.131 (0.91) \\
Age 35--45 & 0.142 & 0.157 (0.81) & 0.160 (0.54) & 0.155 (0.90) & 0.153 (0.57) & 0.146 (0.96) & 0.210 & 0.21 & 0.218 (0.80) & 0.212 (0.49) & 0.221 (0.91) \\
Age 45--55 & 0.343 & 0.352 (0.81) & 0.357 (0.54) & 0.353 (0.90) & 0.345 (0.57) & 0.353 (0.96) & 0.340 & 0.34 & 0.333 (0.80) & 0.340 (0.49) & 0.333 (0.91) \\
Age $\geq$ 55 & 0.37 & 0.373 (0.58) & 0.389 (0.81) & 0.388 (0.54) & 0.378 (0.90) & 0.387 (0.57) & 0.372 (0.96) & 0.29 & 0.285 (0.80) & 0.275 (0.49) & 0.287 (0.91) \\

\addlinespace
\textit{Panel B. Education} \\

High school & 0.21 & 0.205 (0.97) & 0.220 (0.98) & 0.234 (0.09) & 0.255 (0.61) & 0.249 (0.34) & 0.254 (0.50) & 0.24 & 0.228 (0.55) & 0.234 (0.68) & 0.219 (0.46) \\
HS +2 years & 0.19 & 0.207 (0.97) & 0.192 (0.98) & 0.138 (0.09) & 0.177 (0.61) & 0.187 (0.34) & 0.173 (0.50) & 0.18 & 0.172 (0.55) & 0.170 (0.68) & 0.179 (0.46) \\
Bachelor & 0.09 & 0.094 (0.97) & 0.100 (0.98) & 0.125 (0.09) & 0.073 (0.61) & 0.077 (0.34) & 0.066 (0.50) & 0.10 & 0.097 (0.55) & 0.099 (0.68) & 0.105 (0.46) \\
Master & 0.10 & 0.089 (0.97) & 0.095 (0.98) & 0.095 (0.09) & 0.084 (0.61) & 0.074 (0.34) & 0.087 (0.50) & 0.09 & 0.086 (0.55) & 0.097 (0.68) & 0.089 (0.46) \\
Vocational & 0.29 & 0.292 (0.97) & 0.290 (0.98) & 0.299 (0.09) & 0.303 (0.61) & 0.277 (0.34) & 0.304 (0.50) & 0.28 & 0.291 (0.55) & 0.278 (0.68) & 0.283 (0.46) \\
Other & 0.11 & 0.114 (0.97) & 0.102 (0.98) & 0.110 (0.09) & 0.107 (0.61) & 0.136 (0.34) & 0.116 (0.50) & 0.11 & 0.126 (0.55) & 0.123 (0.68) & 0.126 (0.46) \\

\addlinespace
\textit{Panel C. Labor Market Characteristics} \\

Children & 0.42 & 0.474 (0.07) & 0.458 (0.22) & 0.413 (0.66) & 0.445 (0.45) & 0.419 (0.83) & 0.457 (0.24) & 0.48 & 0.468 (0.48) & 0.494 (0.33) & 0.487 (0.59) \\
Married & 0.48 & 0.467 (0.58) & 0.521 (0.15) & 0.442 (0.15) & 0.466 (0.57) & 0.464 (0.52) & 0.482 (0.99) & 0.51 & 0.471 (0.02) & 0.505 (0.88) & 0.513 (0.68) \\

\addlinespace
Search effort & 11.90 & 12.9 (0.04) & 12.6 (0.12) & 12.6 (0.13) & 11.7 (0.70) & 12.4 (0.24) & 12.0 (0.76) & 14.60 & 14.4 (0.56) & 14.4 (0.48) & 14.5 (0.62) \\

\addlinespace
\textit{Panel D. Beliefs} \\

3-month reemployment prob. & 16.40 & 16.7 (0.62) & 17.0 (0.32) & 15.9 (0.37) & 16.8 (0.55) & 16.7 (0.67) & 16.8 (0.57) & 57.60 & 57.3 (0.64) & 57.2 (0.59) & 58.3 (0.36) \\
6-month reemployment prob. & 27.30 & 26.6 (0.55) & 26.6 (0.55) & 25.8 (0.17) & 26.6 (0.53) & 27.4 (0.89) & 25.3 (0.07) & 66.70 & 65.9 (0.31) & 66.6 (0.87) & 66.7 (0.99) \\

\midrule
Sample size & 782 & 542 & 569 & 540 & 521 & 530 & 527 & 1,722 & 2,269 & 2,266 & 2,353 \\
\bottomrule
\end{tabular}

\begin{tablenotes}
\footnotesize
\item \textit{Notes:} This table reports baseline characteristics for the follow-up subsample (individuals who completed the follow-up survey and remained actively searching). Values are means. P-values from tests of equality with the control group are reported in parentheses. For binary and continuous variables, p-values are based on two-sided t-tests; for categorical variables, they are based on chi-squared tests. “Motivational” and “Occupational” denote treatment types. “Advice” refers to direct recommendations, “Peers” to information on peers’ behavior, and “Awareness” to informational messages. “Prob.” denotes probability.
\end{tablenotes}

\end{threeparttable}
}
\end{sidewaystable}

\clearpage

\section{Correction for invitation-stage selection: Theory}\label{sec:invitation_only_selection}\label{sec:selectedsurvey}

In this section, we consider estimators that correct only for selection at the invitation stage of the follow-up survey. The key idea is that treatment may affect invitation probabilities, while latent response behavior remains invariant across treatment arms. In that case, it is natural to target the average treatment effect among individuals who would respond if invited.

\subsection{Setup and target parameter}

Let $T_i \in \{0,1\}$ denote treatment assignment and $X_i$ baseline covariates. Selection into the observed follow-up sample occurs in two stages. Let $E_i$ denote invitation, $R_i$ response conditional on invitation, and
\[
S_i = E_i R_i
\]
the final observation indicator. We define potential variables
\[
Y_i(t), \quad E_i(t), \quad R_i(t), \qquad t \in \{0,1\},
\]
with consistency: $Y_i = Y_i(T_i)$, $E_i = E_i(T_i)$, and $R_i = R_i(T_i)$. Denote by
\[
e_t(X) = \Pr(E_i = 1 \mid X_i = x, T_i = t)
\]
the probability of being invited conditional on treatment and covariates. In our experiment, the invitation mechanism in the control group is deterministic: conditional on $X_i$, individuals are invited if and only if they answered the last two questions of the initial survey. Let
\[
\mathcal{X}_0 = \{x : e_0(x) = 1\}
\]
denote the subset of covariate values for which invitation under control occurs with probability one. We restrict attention to this subpopulation throughout. For notational simplicity, we do not explicitly condition on $X_i \in \mathcal{X}_0$ in what follows, and all expectations and probabilities are understood to be taken with respect to this restricted population.

We target the average treatment effect conditional on response (given invitation), defined as
\[
\theta_R = \mathbb{E}[Y_i(1) - Y_i(0) \mid R_i = 1].
\]

\subsection{Assumptions}

We impose the following assumptions.

\medskip

\begin{hyp}\label{ass:selection}
    Assume that the following conditions hold
    \begin{enumerate}
        \item[(A1)] Random assignment
        $$(Y_i(1),Y_i(0),E_i(1),E_i(0),R_i(1),R_i(0)) \perp T_i \mid X_i.$$
        \item[(A2)] Response invariance
        $$ R_i(1) = R_i(0) \equiv R_i \quad \text{a.s.}$$
         \item[(A3)] Invitation ignorability: for all $t\in\{0,1\}$,
         $$ (Y_i(t),R_i(t)) \perp E_i(t) \mid X_i.$$
         \item[(A4)] Overlap: for all $t\in\{0,1\}$ and all $x$ in the support of $X_i$,
    \begin{align*}
    & 0 < p_t = P(T_i=t) < 1 \\
     &  0 < P(R_i=1|X_i=x) \\
     &   0 < e_t(x) = P(E_i=1 \mid X_i=x, T_i=t).
    \end{align*}
    \end{enumerate}
\end{hyp}

We describe in the following two section the IPW and Neyman-orthogonal estimator for the average treatment effect conditional on the responding conditional on invitation $\theta_R$.

\subsection{Identification on the selected sample}

The key observation is that under response invariance, the conditional expectation of potential outcomes among responders is identified directly from the selected sample.

\begin{prop}\label{prop:selected_identification}
Under Assumptions (A1)--(A4), for $t \in \{0,1\}$:
\[
\mathbb{E}[Y_i \mid S_i = 1, X_i, T_i = t] = \mathbb{E}[Y_i(t) \mid R_i = 1, X_i].
\]
\end{prop}

\textbf{Proof of Proposition \ref{prop:selected_identification}.}
Since $S_i = E_i R_i$, the event $\{S_i = 1\}$ is equivalent to $\{E_i = 1, R_i = 1\}$. Thus:
\begin{align*}
\mathbb{E}[Y_i \mid S_i = 1, X_i, T_i = t] 
&= \mathbb{E}[Y_i \mid E_i = 1, R_i = 1, X_i, T_i = t] \\
&= \mathbb{E}[Y_i(t) \mid E_i(t) = 1, R_i(t) = 1, X_i, T_i = t] \\
&= \mathbb{E}[Y_i(t) \mid E_i(t) = 1, R_i = 1, X_i, T_i = t] && \text{(using (A2))} \\
&= \mathbb{E}[Y_i(t) \mid E_i(t) = 1, R_i = 1, X_i] && \text{(using (A1))} \\
&= \mathbb{E}[Y_i(t) \mid R_i = 1, X_i]. && \text{(using (A3))}
\end{align*}
The last step uses $(Y_i(t), R_i) \perp E_i(t) \mid X_i$, so if $P(R_i=1|X_i=x)>0$ for all $x \in \mathcal{X}_0$, conditioning on $E_i(t) = 1$ does not change the conditional expectation of $Y_i(t)$ given $R_i = 1$ and $X_i$.\hfill $\square$

\subsection{Partially linear regression on the selected sample}

\subsubsection{Model}

We adopt a partially linear model for the conditional expectation among responders:
\[
\mathbb{E}[Y_i(t) \mid R_i = 1, X_i] = \mu_t + h(X_i),
\]
where $\mu_t$ is a treatment-specific intercept and $h(\cdot)$ is a nonparametric function of covariates common to both treatment arms. This implies a constant conditional average treatment effect among responders:
\[
\mathbb{E}[Y_i(1) - Y_i(0) \mid R_i = 1, X_i] = \mu_1 - \mu_0 = \theta_R.
\]

By Proposition \ref{prop:selected_identification}, this model can be written in terms of observables on the selected sample:
\[
\mathbb{E}[Y_i \mid S_i = 1, X_i, T_i] = \theta_R \cdot T_i + h_0(X_i),
\]
where $h_0(X_i) = \mu_0 + h(X_i)$.

\subsubsection{Nuisance functions}\label{sec:nuisance}

Define the nuisance functions on the selected sample:
\begin{align*}
\ell(x) &= \mathbb{E}[Y_i \mid S_i = 1, X_i = x], \\
m(x) &= \mathbb{E}[T_i \mid S_i = 1, X_i = x].
\end{align*}

The function $\ell(x)$ is the conditional expectation of the outcome given covariates among selected individuals. The function $m(x)$ is the treatment propensity in the selected sample, which differs from the population propensity due to differential invitation:
\[
m(x) = \Pr(T_i = 1 \mid S_i = 1, X_i = x) = \frac{p_1 \cdot e_1(x)}{p_1 \cdot e_1(x) + p_0 \cdot e_0(x)}.
\]

In our setting where $e_0(x) = 1$, this simplifies to:
\[
m(x) = \frac{p_1 \cdot e_1(x)}{p_1 \cdot e_1(x) + p_0}.
\]

Even though the population treatment propensity $\Pr(T_i = 1 \mid X_i)$ is constant (equal to $p_1$ under randomization), the selected-sample propensity $m(x)$ varies with $x$ through the invitation probability $e_1(x)$. Regions where treated individuals are less likely to be invited have lower values of $m(x)$ in the selected sample.

\subsubsection{Estimation}

Under the partially linear model, $\ell(x) = \theta_R \cdot m(x) + h_0(x)$, which implies $Y_i - \ell(X_i) = \theta_R (T_i - m(X_i)) + \varepsilon_i$, 
where $\mathbb{E}[\varepsilon_i \mid X_i, S_i = 1] = 0$. The parameter $\theta_R$ is identified from the regression of $(Y_i - \ell(X_i))$ on $(T_i - m(X_i))$ among selected individuals.

\medskip

The partially linear regression (PLR) estimator is:
\[
\hat{\theta}_R = \frac{\sum_{i: S_i = 1} (T_i - \hat{m}(X_i))(Y_i - \hat{\ell}(X_i))}{\sum_{i: S_i = 1} (T_i - \hat{m}(X_i))^2},
\]
where $\hat{\ell}$ and $\hat{m}$ are nonparametric estimators of the nuisance functions. The estimator is identical in form to the standard Double/Debiased Machine Learning (DML) partially linear regression estimator of \citet{chernozhukov2018double}, applied to the selected sample. The key insight is that under response invariance and the partially linear model, applying DML to the selected sample identifies the treatment effect among responders $\theta_R$, not merely the treatment effect among the selected.

\subsubsection{Neyman-orthogonal score}

The PLR estimator solves the sample moment condition $\sum_{i: S_i = 1} \psi_i(\theta_R, \eta) = 0$, where the score function is:
\[
\psi_i(\theta, \eta) = (T_i - m(X_i))\left[Y_i - \ell(X_i) - \theta(T_i - m(X_i))\right],
\]
and $\eta = (\ell, m)$ collects the nuisance functions. The score $\psi_i(\theta_R, \eta)$ is Neyman-orthogonal with respect to $\eta$ at the true parameter values. Neyman orthogonality ensures that first-order errors in estimating $\ell$ and $m$ do not affect the asymptotic distribution of $\hat{\theta}_R$, allowing the use of machine learning methods for nuisance estimation.

\subsection{Implementation}

\subsubsection{Cross-fitting}

To avoid regularization bias when using machine learning estimators, we implement the estimator using cross-fitting:

\begin{enumerate}
\item Partition the selected sample $\{i : S_i = 1\}$ into $K$ folds.
\item For each fold $k$:
\begin{enumerate}
    \item Estimate $\hat{\ell}^{(-k)}$ and $\hat{m}^{(-k)}$ using observations outside fold $k$.
    \item Compute residuals for observations in fold $k$ using out-of-fold predictions.
\end{enumerate}
\item Pool across folds:
\[
\hat{\theta}_R = \frac{\sum_{k=1}^K \sum_{i \in I_k} (T_i - \hat{m}^{(-k)}(X_i))(Y_i - \hat{\ell}^{(-k)}(X_i))}{\sum_{k=1}^K \sum_{i \in I_k} (T_i - \hat{m}^{(-k)}(X_i))^2}.
\]
\end{enumerate}

We repeat this procedure across $J$ independent random partitions and aggregate using median aggregation following Definition~4.3 in \citet{chernozhukov2018generic}.

\subsubsection{Nuisance estimation}

The nuisance functions are estimated using machine learning on the selected sample:

\begin{itemize}
\item $\hat{\ell}(x)$: Estimated by regressing $Y_i$ on $X_i$ among selected individuals ($S_i = 1$).

\item $\hat{m}(x)$: Estimated by regressing $T_i$ on $X_i$ among selected individuals ($S_i = 1$), or equivalently by classification with probability predictions.
\end{itemize}

We use random forests, though other learners (gradient boosting, elastic net, or ensemble methods) may also be used.

\subsubsection{Variance estimation}

The variance is estimated using the influence function. Define:
\[
\hat{V} = \frac{1}{n_S} \sum_{i: S_i = 1} (T_i - \hat{m}(X_i))^2,
\]
where $n_S = \sum_i S_i$ is the selected sample size. The influence function for observation $i$ (with $S_i = 1$) is:
\[
\text{IF}_i = \frac{1}{\hat{V}} (T_i - \hat{m}(X_i))\left[Y_i - \hat{\ell}(X_i) - \hat{\theta}_R(T_i - \hat{m}(X_i))\right].
\]

The variance is estimated as:
\[
\widehat{\text{Var}}(\hat{\theta}_R) = \frac{1}{n_S^2} \sum_{i: S_i = 1} \text{IF}_i^2.
\]

\subsubsection{Aggregation across sample splits}

Following \citet{chernozhukov2018generic}, we aggregate across $J$ independent sample splits. For each split $j$, we obtain $\hat{\theta}_R^{(j)}$ and $\hat{\sigma}^{(j)}$, yielding confidence interval $[L^{(j)}, U^{(j)}]$. The final estimate and confidence interval are:
\[
\hat{\theta}_R = \text{median}\{\hat{\theta}_R^{(1)}, \ldots, \hat{\theta}_R^{(J)}\},
\]
\[
\text{CI} = \left[ \text{median}\{L^{(1)}, \ldots, L^{(J)}\}, \; \text{median}\{U^{(1)}, \ldots, U^{(J)}\} \right].
\]

\FloatBarrier

\section{Correction for invitation-stage: Results}\label{sec:results_correction}

\begin{table}[htbp]
\centering
\caption{Permutation Test Results for Survey Response Conditional on Invitation}
\label{tab:perm_test_results}

\scalebox{0.8}{
\begin{tabular}{l r r r r r r}
\toprule
Treatment
& Invited
& Treated
& Control
& Difference
& Permutation p-value
& Covariate-adjusted p-value \\
\midrule

Motivational
& 4,773 & 0.400 & 0.398 & 0.003 & 0.884 & 0.398 \\

Occupational
& 4,721 & 0.396 & 0.398 & -0.002 & 0.907 & 0.700 \\

M. Direct Advice
& 2,837 & 0.392 & 0.398 & -0.006 & 0.779 & 0.267 \\

M. Peers' Behaviors
& 2,796 & 0.412 & 0.398 & 0.015 & 0.470 & 0.581 \\

M. Raise Awareness
& 2,792 & 0.396 & 0.398 & -0.001 & 0.963 & 0.161 \\

O. Direct Advice
& 2,821 & 0.391 & 0.398 & -0.007 & 0.754 & 0.746 \\

O. Peers' Behaviors
& 2,784 & 0.403 & 0.398 & 0.005 & 0.812 & 0.883 \\

O. Raise Awareness
& 2,768 & 0.393 & 0.398 & -0.005 & 0.835 & 0.625 \\

\bottomrule
\end{tabular}
}

\vspace{0.75em}

\parbox{\textwidth}{\footnotesize
\textit{Notes:} This table reports permutation tests of Assumption~\ref{ass:selection}-(A2), comparing survey response rates between treatment and control groups conditional on invitation. The first p-value column reports the standard permutation test. The second reports the permutation test after residualizing outcomes with respect to baseline covariates. Response rates are reported as proportions.
}

\end{table}

\begin{table}[htbp]
\centering
\caption{Treatment Effects on Search Effort: DML on Selected Sample}
\label{tab:dml_eff_selected}

\scalebox{0.9}{
\begin{tabular}{l c c c c c c}
\toprule
& \multicolumn{6}{c}{Estimation Method} \\
\cmidrule(lr){2-7}
Treatment
& Naive
& Stacked
& Random Forest
& Elastic Net
& XGBoost
& Naive IPW \\
\midrule

Motivational
& \begin{tabular}[c]{@{}c@{}}1.030***\\(0.323)\end{tabular}
& \begin{tabular}[c]{@{}c@{}}1.080***\\(0.342)\end{tabular}
& \begin{tabular}[c]{@{}c@{}}1.110***\\(0.343)\end{tabular}
& \begin{tabular}[c]{@{}c@{}}0.980***\\(0.341)\end{tabular}
& \begin{tabular}[c]{@{}c@{}}1.060***\\(0.349)\end{tabular}
& \begin{tabular}[c]{@{}c@{}}1.310***\\(0.382)\end{tabular}
\\[1.5ex]

Occupational
& \begin{tabular}[c]{@{}c@{}}0.230\\(0.320)\end{tabular}
& \begin{tabular}[c]{@{}c@{}}0.220\\(0.340)\end{tabular}
& \begin{tabular}[c]{@{}c@{}}0.290\\(0.337)\end{tabular}
& \begin{tabular}[c]{@{}c@{}}0.220\\(0.339)\end{tabular}
& \begin{tabular}[c]{@{}c@{}}0.160\\(0.344)\end{tabular}
& \begin{tabular}[c]{@{}c@{}}0.600\\(0.382)\end{tabular}
\\

\bottomrule
\end{tabular}
}

\vspace{0.75em}

\parbox{\textwidth}{\footnotesize
\textit{Notes:} This table reports estimates of the average treatment effect on search effort among survey respondents, $\theta_R=\mathbb{E}[Y(1)-Y(0)\mid R=1]$. Under response invariance, this parameter is identified using Double Machine Learning (DML) on the selected sample consisting of follow-up survey respondents satisfying $e_0(X)=1$. The naive specification regresses search effort on treatment status while controlling for baseline search effort and arrival-rate beliefs. DML specifications use a rich set of administrative characteristics, pre-intervention behavior, and baseline survey responses to estimate nuisance functions. Random Forest, Elastic Net (GlmNet), XGBoost, and a stacked ensemble are considered. The Naive IPW specification reweights observations using the inverse predicted probability of survey participation. Robust standard errors are reported in parentheses. $^{*}p<0.10$, $^{**}p<0.05$, and $^{***}p<0.01$.
}

\end{table}

\begin{table}[htbp]
\centering
\caption{Treatment Effects on Reservation Wage (log): DML on Selected Sample}
\label{tab:dml_wage_selected}

\scalebox{0.9}{
\begin{tabular}{l c c c c c c}
\toprule
& \multicolumn{6}{c}{Estimation Method} \\
\cmidrule(lr){2-7}
Treatment
& Naive
& Stacked
& Random Forest
& Elastic Net
& XGBoost
& Naive IPW \\
\midrule

Motivational
& \begin{tabular}[c]{@{}c@{}}0.021**\\(0.009)\end{tabular}
& \begin{tabular}[c]{@{}c@{}}0.019**\\(0.009)\end{tabular}
& \begin{tabular}[c]{@{}c@{}}0.019**\\(0.009)\end{tabular}
& \begin{tabular}[c]{@{}c@{}}0.021**\\(0.009)\end{tabular}
& \begin{tabular}[c]{@{}c@{}}0.012\\(0.010)\end{tabular}
& \begin{tabular}[c]{@{}c@{}}0.021**\\(0.009)\end{tabular}
\\[1.5ex]

Occupational
& \begin{tabular}[c]{@{}c@{}}0.009\\(0.009)\end{tabular}
& \begin{tabular}[c]{@{}c@{}}0.009\\(0.009)\end{tabular}
& \begin{tabular}[c]{@{}c@{}}0.009\\(0.009)\end{tabular}
& \begin{tabular}[c]{@{}c@{}}0.008\\(0.009)\end{tabular}
& \begin{tabular}[c]{@{}c@{}}0.008\\(0.009)\end{tabular}
& \begin{tabular}[c]{@{}c@{}}0.009\\(0.008)\end{tabular}
\\

\bottomrule
\end{tabular}
}

\vspace{0.75em}

\parbox{\textwidth}{\footnotesize
\textit{Notes:} This table reports estimates of the average treatment effect on log reservation wages among survey respondents. Reservation wages are winsorized at the interval [1,400, 5,000] before taking logarithms. Under response invariance, Double Machine Learning (DML) on the selected sample identifies $\theta_R=\mathbb{E}[Y(1)-Y(0)\mid R=1]$. The naive specification controls for baseline log reservation wage, wage-offer beliefs, and occupation fixed effects. The stacked specification combines Elastic Net, Random Forest, and XGBoost learners for nuisance-function estimation. Robust standard errors are reported in parentheses. $^{*}p<0.10$, $^{**}p<0.05$, and $^{***}p<0.01$.
}

\end{table}

\begin{table}[htbp]
\centering
\caption{Treatment Effects on Beliefs: DML on Selected Sample}
\label{tab:dml_beliefs_selected}

\scalebox{0.8}{
\begin{tabular}{l cc}
\toprule
Outcome & Motivational (S1) & Occupational (S2) \\
\midrule

Arrival rate (10 applications)
& \begin{tabular}[c]{@{}c@{}}-0.560\\(0.974)\end{tabular}
& \begin{tabular}[c]{@{}c@{}}0.590\\(0.973)\end{tabular}
\\[1ex]

Arrival rate (20 applications)
& \begin{tabular}[c]{@{}c@{}}-0.980\\(1.114)\end{tabular}
& \begin{tabular}[c]{@{}c@{}}-0.980\\(1.107)\end{tabular}
\\[1ex]

Arrival rate (30 applications)
& \begin{tabular}[c]{@{}c@{}}-0.520\\(1.231)\end{tabular}
& \begin{tabular}[c]{@{}c@{}}0.250\\(1.247)\end{tabular}
\\[1ex]

Perceived return to effort
& \begin{tabular}[c]{@{}c@{}}-0.030\\(0.056)\end{tabular}
& \begin{tabular}[c]{@{}c@{}}-0.040\\(0.057)\end{tabular}
\\[1ex]

Wage offer belief (low)
& \begin{tabular}[c]{@{}c@{}}-0.930\\(1.356)\end{tabular}
& \begin{tabular}[c]{@{}c@{}}-0.350\\(1.380)\end{tabular}
\\[1ex]

Wage offer belief (mid)
& \begin{tabular}[c]{@{}c@{}}0.450\\(1.180)\end{tabular}
& \begin{tabular}[c]{@{}c@{}}0.770\\(1.177)\end{tabular}
\\[1ex]

Wage offer belief (high)
& \begin{tabular}[c]{@{}c@{}}0.690\\(1.154)\end{tabular}
& \begin{tabular}[c]{@{}c@{}}1.090\\(1.172)\end{tabular}
\\

\bottomrule
\end{tabular}
}

\vspace{0.75em}

\parbox{\textwidth}{\footnotesize
\textit{Notes:} This table reports Double Machine Learning (DML) estimates of treatment effects on beliefs among survey respondents. Under response invariance, DML on the selected sample identifies $\theta_R = \mathbb{E}[Y(1)-Y(0)\mid R=1]$. A stacked ensemble consisting of Elastic Net, Random Forest, and XGBoost learners is used for nuisance-function estimation. Robust standard errors are reported in parentheses. $^{*}p<0.10$, $^{**}p<0.05$, and $^{***}p<0.01$.
}

\end{table}

\FloatBarrier

\section{Complementary Results on Heterogeneity of the Treatment Effects}\label{sec:Heterogeneity_GML_employment}

\begin{table}[htbp]
\centering
\caption{Method Comparison: Occupational --- All}
\label{tab:mcomp_occ_all}

\textbf{Panel A: Standard Learners}
\vspace{1em}

\scalebox{0.8}{
\begin{tabular}{l cc cc cc}
\toprule
& \multicolumn{2}{c}{ENet+CB}
& \multicolumn{2}{c}{NNet+CB}
& \multicolumn{2}{c}{Ranger+CB} \\
\cmidrule(lr){2-3}\cmidrule(lr){4-5}\cmidrule(lr){6-7}
Outcome & ATE & HET & ATE & HET & ATE & HET \\
\midrule

12-month reemployment
& \begin{tabular}[c]{@{}c@{}}\textbf{0.004}\\{\footnotesize(-0.005,\,0.013)}\\{\footnotesize[0.475]}\end{tabular}
& \begin{tabular}[c]{@{}c@{}}\textbf{0.224}\\{\footnotesize(-0.042,\,0.516)}\\{\footnotesize[0.079]}\end{tabular}
& \begin{tabular}[c]{@{}c@{}}0.004\\{\footnotesize(-0.004,\,0.013)}\\{\footnotesize[0.404]}\end{tabular}
& \begin{tabular}[c]{@{}c@{}}0.023\\{\footnotesize(-0.048,\,0.095)}\\{\footnotesize[0.294]}\end{tabular}
& \begin{tabular}[c]{@{}c@{}}0.005\\{\footnotesize(-0.003,\,0.014)}\\{\footnotesize[0.316]}\end{tabular}
& \begin{tabular}[c]{@{}c@{}}0.027\\{\footnotesize(-0.052,\,0.105)}\\{\footnotesize[0.289]}\end{tabular}
\\[1.5ex]

Relative wage
& \begin{tabular}[c]{@{}c@{}}0.004\\{\footnotesize(-0.003,\,0.011)}\\{\footnotesize[0.391]}\end{tabular}
& \begin{tabular}[c]{@{}c@{}}-0.387\\{\footnotesize(-0.927,\,0.181)}\\{\footnotesize[0.869]}\end{tabular}
& \begin{tabular}[c]{@{}c@{}}\textbf{0.004}\\{\footnotesize(-0.003,\,0.011)}\\{\footnotesize[0.380]}\end{tabular}
& \begin{tabular}[c]{@{}c@{}}\textbf{-0.027}\\{\footnotesize(-0.187,\,0.128)}\\{\footnotesize[0.614]}\end{tabular}
& \begin{tabular}[c]{@{}c@{}}0.004\\{\footnotesize(-0.003,\,0.011)}\\{\footnotesize[0.361]}\end{tabular}
& \begin{tabular}[c]{@{}c@{}}-0.185\\{\footnotesize(-0.345,\,-0.025)}\\{\footnotesize[0.973]}\end{tabular}
\\

\bottomrule
\end{tabular}
}

\vspace{1em}

\textbf{Panel B: X-Learners}
\vspace{1em}

\scalebox{0.8}{
\begin{tabular}{l cc cc}
\toprule
& \multicolumn{2}{c}{X-Learn (Ranger)}
& \multicolumn{2}{c}{X-Learn (ENet)} \\
\cmidrule(lr){2-3}\cmidrule(lr){4-5}
Outcome & ATE & HET & ATE & HET \\
\midrule

12-month reemployment
& \begin{tabular}[c]{@{}c@{}}0.005\\{\footnotesize(-0.004,\,0.014)}\\{\footnotesize[0.333]}\end{tabular}
& \begin{tabular}[c]{@{}c@{}}0.062\\{\footnotesize(-0.057,\,0.181)}\\{\footnotesize[0.196]}\end{tabular}
& \begin{tabular}[c]{@{}c@{}}0.004\\{\footnotesize(-0.005,\,0.013)}\\{\footnotesize[0.434]}\end{tabular}
& \begin{tabular}[c]{@{}c@{}}0.069\\{\footnotesize(-0.040,\,0.179)}\\{\footnotesize[0.148]}\end{tabular}
\\[1.5ex]

Relative wage
& \begin{tabular}[c]{@{}c@{}}0.004\\{\footnotesize(-0.003,\,0.011)}\\{\footnotesize[0.344]}\end{tabular}
& \begin{tabular}[c]{@{}c@{}}-0.198\\{\footnotesize(-0.431,\,0.035)}\\{\footnotesize[0.919]}\end{tabular}
& \begin{tabular}[c]{@{}c@{}}0.004\\{\footnotesize(-0.003,\,0.011)}\\{\footnotesize[0.389]}\end{tabular}
& \begin{tabular}[c]{@{}c@{}}-0.188\\{\footnotesize(-0.466,\,0.100)}\\{\footnotesize[0.858]}\end{tabular}
\\

\bottomrule
\end{tabular}
}

\vspace{0.75em}

\parbox{\textwidth}{\footnotesize
\textit{Notes:} Cross-fitting with 50 sample splits. +CB denotes the causal boosting procedure (Algorithm~6.2 in \cite{chernozhukov2018generic}). X-Learn denotes the X-learner of \citeauthor{kunzel2019metalearners} (\citeyear{kunzel2019metalearners}). HET reports one-sided heterogeneity tests. Numbers in parentheses are 95\% confidence intervals; numbers in brackets are p-values. Boldface indicates the best-performing specification.
}

\end{table}
\begin{table}[htbp]
\centering
\caption{Method Comparison: Occupational --- Non-Pessimists}
\label{tab:mcomp_occ_opt}

\textbf{Panel A: Standard Learners}
\vspace{1em}

\scalebox{0.8}{
\begin{tabular}{l cc cc cc}
\toprule
& \multicolumn{2}{c}{ENet+CB}
& \multicolumn{2}{c}{NNet+CB}
& \multicolumn{2}{c}{Ranger+CB} \\
\cmidrule(lr){2-3}\cmidrule(lr){4-5}\cmidrule(lr){6-7}
Outcome & ATE & HET & ATE & HET & ATE & HET \\
\midrule

12-month reemployment
& \begin{tabular}[c]{@{}c@{}}\textbf{-0.002}\\{\footnotesize(-0.013,\,0.008)}\\{\footnotesize[0.716]}\end{tabular}
& \begin{tabular}[c]{@{}c@{}}\textbf{0.076}\\{\footnotesize(-0.285,\,0.449)}\\{\footnotesize[0.374]}\end{tabular}
& \begin{tabular}[c]{@{}c@{}}-0.003\\{\footnotesize(-0.013,\,0.007)}\\{\footnotesize[0.641]}\end{tabular}
& \begin{tabular}[c]{@{}c@{}}0.037\\{\footnotesize(-0.047,\,0.124)}\\{\footnotesize[0.235]}\end{tabular}
& \begin{tabular}[c]{@{}c@{}}-0.002\\{\footnotesize(-0.013,\,0.008)}\\{\footnotesize[0.705]}\end{tabular}
& \begin{tabular}[c]{@{}c@{}}0.015\\{\footnotesize(-0.079,\,0.110)}\\{\footnotesize[0.396]}\end{tabular}
\\[1.5ex]

Relative wage
& \begin{tabular}[c]{@{}c@{}}0.005\\{\footnotesize(-0.003,\,0.013)}\\{\footnotesize[0.342]}\end{tabular}
& \begin{tabular}[c]{@{}c@{}}-0.972\\{\footnotesize(-1.581,\,-0.364)}\\{\footnotesize[0.996]}\end{tabular}
& \begin{tabular}[c]{@{}c@{}}0.005\\{\footnotesize(-0.003,\,0.013)}\\{\footnotesize[0.334]}\end{tabular}
& \begin{tabular}[c]{@{}c@{}}-0.064\\{\footnotesize(-0.258,\,0.141)}\\{\footnotesize[0.701]}\end{tabular}
& \begin{tabular}[c]{@{}c@{}}0.004\\{\footnotesize(-0.004,\,0.012)}\\{\footnotesize[0.368]}\end{tabular}
& \begin{tabular}[c]{@{}c@{}}-0.076\\{\footnotesize(-0.271,\,0.129)}\\{\footnotesize[0.733]}\end{tabular}
\\

\bottomrule
\end{tabular}
}

\vspace{1em}

\textbf{Panel B: X-Learners}
\vspace{1em}

\scalebox{0.8}{
\begin{tabular}{l cc cc}
\toprule
& \multicolumn{2}{c}{X-Learn (Ranger)}
& \multicolumn{2}{c}{X-Learn (ENet)} \\
\cmidrule(lr){2-3}\cmidrule(lr){4-5}
Outcome & ATE & HET & ATE & HET \\
\midrule

12-month reemployment
& \begin{tabular}[c]{@{}c@{}}-0.003\\{\footnotesize(-0.013,\,0.008)}\\{\footnotesize[0.691]}\end{tabular}
& \begin{tabular}[c]{@{}c@{}}0.027\\{\footnotesize(-0.111,\,0.164)}\\{\footnotesize[0.376]}\end{tabular}
& \begin{tabular}[c]{@{}c@{}}-0.002\\{\footnotesize(-0.013,\,0.008)}\\{\footnotesize[0.730]}\end{tabular}
& \begin{tabular}[c]{@{}c@{}}0.069\\{\footnotesize(-0.052,\,0.189)}\\{\footnotesize[0.175]}\end{tabular}
\\[1.5ex]

Relative wage
& \begin{tabular}[c]{@{}c@{}}\textbf{0.004}\\{\footnotesize(-0.004,\,0.012)}\\{\footnotesize[0.380]}\end{tabular}
& \begin{tabular}[c]{@{}c@{}}\textbf{0.031}\\{\footnotesize(-0.237,\,0.294)}\\{\footnotesize[0.422]}\end{tabular}
& \begin{tabular}[c]{@{}c@{}}0.005\\{\footnotesize(-0.004,\,0.013)}\\{\footnotesize[0.354]}\end{tabular}
& \begin{tabular}[c]{@{}c@{}}-0.022\\{\footnotesize(-0.308,\,0.277)}\\{\footnotesize[0.547]}\end{tabular}
\\

\bottomrule
\end{tabular}
}

\vspace{0.75em}

\parbox{\textwidth}{\footnotesize
\textit{Notes:} Cross-fitting with 50 sample splits. +CB denotes the causal boosting procedure (Algorithm~6.2 in \cite{chernozhukov2018generic}). X-Learn denotes the X-learner of \citeauthor{kunzel2019metalearners} (\citeyear{kunzel2019metalearners}). HET reports one-sided heterogeneity tests. Numbers in parentheses are 95\% confidence intervals; numbers in brackets are p-values. Boldface indicates the best-performing specification.
}

\end{table}
\begin{table}[htbp]
\centering
\caption{Method Comparison: Occupational --- Pessimists}
\label{tab:mcomp_occ_pess}

\textbf{Panel A: Standard Learners}
\vspace{1em}

\scalebox{0.8}{
\begin{tabular}{l cc cc cc}
\toprule
& \multicolumn{2}{c}{ENet+CB}
& \multicolumn{2}{c}{NNet+CB}
& \multicolumn{2}{c}{Ranger+CB} \\
\cmidrule(lr){2-3}\cmidrule(lr){4-5}\cmidrule(lr){6-7}
Outcome & ATE & HET & ATE & HET & ATE & HET \\
\midrule

12-month reemployment
& \begin{tabular}[c]{@{}c@{}}0.021\\{\footnotesize(0.005,\,0.037)}\\{\footnotesize[0.035]}\end{tabular}
& \begin{tabular}[c]{@{}c@{}}-0.355\\{\footnotesize(-0.892,\,0.186)}\\{\footnotesize[0.863]}\end{tabular}
& \begin{tabular}[c]{@{}c@{}}0.018\\{\footnotesize(0.002,\,0.034)}\\{\footnotesize[0.064]}\end{tabular}
& \begin{tabular}[c]{@{}c@{}}0.059\\{\footnotesize(-0.043,\,0.160)}\\{\footnotesize[0.171]}\end{tabular}
& \begin{tabular}[c]{@{}c@{}}0.017\\{\footnotesize(0.001,\,0.033)}\\{\footnotesize[0.084]}\end{tabular}
& \begin{tabular}[c]{@{}c@{}}-0.029\\{\footnotesize(-0.169,\,0.110)}\\{\footnotesize[0.635]}\end{tabular}
\\[1.5ex]

Relative wage
& \begin{tabular}[c]{@{}c@{}}0.001\\{\footnotesize(-0.016,\,0.017)}\\{\footnotesize[0.957]}\end{tabular}
& \begin{tabular}[c]{@{}c@{}}-0.349\\{\footnotesize(-0.898,\,0.168)}\\{\footnotesize[0.863]}\end{tabular}
& \begin{tabular}[c]{@{}c@{}}0.000\\{\footnotesize(-0.016,\,0.017)}\\{\footnotesize[0.964]}\end{tabular}
& \begin{tabular}[c]{@{}c@{}}-0.132\\{\footnotesize(-0.441,\,0.169)}\\{\footnotesize[0.768]}\end{tabular}
& \begin{tabular}[c]{@{}c@{}}0.001\\{\footnotesize(-0.016,\,0.017)}\\{\footnotesize[0.961]}\end{tabular}
& \begin{tabular}[c]{@{}c@{}}-0.222\\{\footnotesize(-0.476,\,0.034)}\\{\footnotesize[0.923]}\end{tabular}
\\

\bottomrule
\end{tabular}
}

\vspace{1em}

\textbf{Panel B: X-Learners}
\vspace{1em}

\scalebox{0.8}{
\begin{tabular}{l cc cc}
\toprule
& \multicolumn{2}{c}{X-Learn (Ranger)}
& \multicolumn{2}{c}{X-Learn (ENet)} \\
\cmidrule(lr){2-3}\cmidrule(lr){4-5}
Outcome & ATE & HET & ATE & HET \\
\midrule

12-month reemployment
& \begin{tabular}[c]{@{}c@{}}\textbf{0.017}\\{\footnotesize(0.001,\,0.033)}\\{\footnotesize[0.089]}\end{tabular}
& \begin{tabular}[c]{@{}c@{}}\textbf{0.113}\\{\footnotesize(-0.097,\,0.322)}\\{\footnotesize[0.188]}\end{tabular}
& \begin{tabular}[c]{@{}c@{}}0.020\\{\footnotesize(0.004,\,0.036)}\\{\footnotesize[0.042]}\end{tabular}
& \begin{tabular}[c]{@{}c@{}}0.066\\{\footnotesize(-0.175,\,0.307)}\\{\footnotesize[0.327]}\end{tabular}
\\[1.5ex]

Relative wage
& \begin{tabular}[c]{@{}c@{}}\textbf{0.000}\\{\footnotesize(-0.016,\,0.017)}\\{\footnotesize[0.962]}\end{tabular}
& \begin{tabular}[c]{@{}c@{}}\textbf{-0.105}\\{\footnotesize(-0.454,\,0.238)}\\{\footnotesize[0.692]}\end{tabular}
& \begin{tabular}[c]{@{}c@{}}-0.000\\{\footnotesize(-0.017,\,0.016)}\\{\footnotesize[0.970]}\end{tabular}
& \begin{tabular}[c]{@{}c@{}}-0.243\\{\footnotesize(-0.663,\,0.163)}\\{\footnotesize[0.833]}\end{tabular}
\\

\bottomrule
\end{tabular}
}

\vspace{0.75em}

\parbox{\textwidth}{\footnotesize
\textit{Notes:} Cross-fitting with 50 sample splits. +CB denotes the causal boosting procedure (Algorithm~6.2 in \cite{chernozhukov2018generic}). X-Learn denotes the X-learner of \citeauthor{kunzel2019metalearners} (\citeyear{kunzel2019metalearners}). HET reports one-sided heterogeneity tests. Numbers in parentheses are 95\% confidence intervals; numbers in brackets are p-values. Boldface indicates the best-performing specification.
}

\end{table}
\begin{table}[htbp]
\centering
\caption{Method Comparison: Motivational --- Pessimists}
\label{tab:mcomp_mot}

\textbf{Panel A: Standard Learners}
\vspace{1em}

\scalebox{0.8}{
\begin{tabular}{l cc cc cc}
\toprule
& \multicolumn{2}{c}{ENet+CB}
& \multicolumn{2}{c}{NNet+CB}
& \multicolumn{2}{c}{Ranger+CB} \\
\cmidrule(lr){2-3}\cmidrule(lr){4-5}\cmidrule(lr){6-7}
Outcome & ATE & HET & ATE & HET & ATE & HET \\
\midrule

12-month reemployment
& \begin{tabular}[c]{@{}c@{}}\textbf{0.022}\\{\footnotesize(0.006,\,0.038)}\\{\footnotesize[0.027]}\end{tabular}
& \begin{tabular}[c]{@{}c@{}}\textbf{0.102}\\{\footnotesize(-0.350,\,0.548)}\\{\footnotesize[0.354]}\end{tabular}
& \begin{tabular}[c]{@{}c@{}}0.018\\{\footnotesize(0.002,\,0.034)}\\{\footnotesize[0.065]}\end{tabular}
& \begin{tabular}[c]{@{}c@{}}0.026\\{\footnotesize(-0.079,\,0.130)}\\{\footnotesize[0.340]}\end{tabular}
& \begin{tabular}[c]{@{}c@{}}0.020\\{\footnotesize(0.004,\,0.036)}\\{\footnotesize[0.037]}\end{tabular}
& \begin{tabular}[c]{@{}c@{}}-0.086\\{\footnotesize(-0.221,\,0.048)}\\{\footnotesize[0.854]}\end{tabular}
\\[1.5ex]

Relative wage
& \begin{tabular}[c]{@{}c@{}}\textbf{0.010}\\{\footnotesize(-0.007,\,0.027)}\\{\footnotesize[0.330]}\end{tabular}
& \begin{tabular}[c]{@{}c@{}}\textbf{0.114}\\{\footnotesize(-0.331,\,0.523)}\\{\footnotesize[0.331]}\end{tabular}
& \begin{tabular}[c]{@{}c@{}}0.009\\{\footnotesize(-0.008,\,0.026)}\\{\footnotesize[0.384]}\end{tabular}
& \begin{tabular}[c]{@{}c@{}}-0.007\\{\footnotesize(-0.261,\,0.239)}\\{\footnotesize[0.517]}\end{tabular}
& \begin{tabular}[c]{@{}c@{}}0.009\\{\footnotesize(-0.008,\,0.026)}\\{\footnotesize[0.406]}\end{tabular}
& \begin{tabular}[c]{@{}c@{}}-0.059\\{\footnotesize(-0.292,\,0.168)}\\{\footnotesize[0.660]}\end{tabular}
\\

\bottomrule
\end{tabular}
}

\vspace{1em}

\textbf{Panel B: X-Learners}
\vspace{1em}

\scalebox{0.8}{
\begin{tabular}{l cc cc}
\toprule
& \multicolumn{2}{c}{X-Learn (Ranger)}
& \multicolumn{2}{c}{X-Learn (ENet)} \\
\cmidrule(lr){2-3}\cmidrule(lr){4-5}
Outcome & ATE & HET & ATE & HET \\
\midrule

12-month reemployment
& \begin{tabular}[c]{@{}c@{}}0.020\\{\footnotesize(0.004,\,0.036)}\\{\footnotesize[0.042]}\end{tabular}
& \begin{tabular}[c]{@{}c@{}}-0.095\\{\footnotesize(-0.300,\,0.108)}\\{\footnotesize[0.778]}\end{tabular}
& \begin{tabular}[c]{@{}c@{}}0.021\\{\footnotesize(0.005,\,0.037)}\\{\footnotesize[0.035]}\end{tabular}
& \begin{tabular}[c]{@{}c@{}}-0.013\\{\footnotesize(-0.245,\,0.219)}\\{\footnotesize[0.538]}\end{tabular}
\\[1.5ex]

Relative wage
& \begin{tabular}[c]{@{}c@{}}0.009\\{\footnotesize(-0.008,\,0.026)}\\{\footnotesize[0.384]}\end{tabular}
& \begin{tabular}[c]{@{}c@{}}-0.119\\{\footnotesize(-0.444,\,0.219)}\\{\footnotesize[0.714]}\end{tabular}
& \begin{tabular}[c]{@{}c@{}}0.009\\{\footnotesize(-0.008,\,0.026)}\\{\footnotesize[0.373]}\end{tabular}
& \begin{tabular}[c]{@{}c@{}}-0.026\\{\footnotesize(-0.386,\,0.344)}\\{\footnotesize[0.547]}\end{tabular}
\\

\bottomrule
\end{tabular}
}

\vspace{0.75em}

\parbox{\textwidth}{\footnotesize
\textit{Notes:} Cross-fitting with 50 sample splits. +CB denotes the causal boosting procedure (Algorithm~6.2 in \cite{chernozhukov2018generic}). X-Learn denotes the X-learner of \citeauthor{kunzel2019metalearners} (\citeyear{kunzel2019metalearners}). HET reports one-sided heterogeneity tests. Numbers in parentheses are 95\% confidence intervals; numbers in brackets are p-values. Boldface indicates the best-performing specification.
}

\end{table}

\FloatBarrier
\subsection{On the Occupational Treatment - Optimists}

\begin{table}[htbp]
\centering
\caption{GATES (Sorted): Occupational --- Non-Pessimists}
\label{tab:gates_occ_opt}

\scalebox{0.8}{
\begin{tabular}{l ccccc cc}
\toprule
Outcome & G1 & G2 & G3 & G4 & G5 & G5$-$G1 & Equality p-value \\
\midrule

12-month reemployment
& \begin{tabular}[c]{@{}c@{}}-0.017\\{\footnotesize[-0.050,\;0.016]}\end{tabular}
& \begin{tabular}[c]{@{}c@{}}-0.008\\{\footnotesize[-0.041,\;0.025]}\end{tabular}
& \begin{tabular}[c]{@{}c@{}}-0.003\\{\footnotesize[-0.036,\;0.031]}\end{tabular}
& \begin{tabular}[c]{@{}c@{}}0.003\\{\footnotesize[-0.030,\;0.037]}\end{tabular}
& \begin{tabular}[c]{@{}c@{}}0.013\\{\footnotesize[-0.021,\;0.046]}\end{tabular}
& 0.031
& 0.439
\\[1.5ex]

Relative wage
& \begin{tabular}[c]{@{}c@{}}-0.008\\{\footnotesize[-0.033,\;0.018]}\end{tabular}
& \begin{tabular}[c]{@{}c@{}}-0.001\\{\footnotesize[-0.026,\;0.024]}\end{tabular}
& \begin{tabular}[c]{@{}c@{}}0.004\\{\footnotesize[-0.022,\;0.029]}\end{tabular}
& \begin{tabular}[c]{@{}c@{}}0.008\\{\footnotesize[-0.016,\;0.033]}\end{tabular}
& \begin{tabular}[c]{@{}c@{}}0.016\\{\footnotesize[-0.008,\;0.042]}\end{tabular}
& 0.026
& 0.553
\\

\bottomrule
\end{tabular}
}

\vspace{0.75em}

\parbox{\textwidth}{\footnotesize
\textit{Notes:} Estimates are obtained using 5-fold cross-fitting with 50 sample splits. GATES groups are sorted from the least affected group (G1) to the most affected group (G5). Numbers in brackets report 90\% confidence intervals. The final column reports the p-value for the equality test across groups.
}

\end{table}
\begin{table}[htbp]
\centering
\caption{CLAN: 12-Month Reemployment --- Non-Pessimists (ENet+CB)}
\label{tab:clan_occ_opt_retour_12m}

\scalebox{0.8}{
\begin{tabular}{l ccccc}
\toprule
Variable & Most Affected & Least Affected & Difference & 90\% CI & p-value \\
\midrule

Unemp. dur.
& 21.675
& 22.367
& -0.692
& [-1.412,\;0.028]
& 0.114 \\

Job belief
& 56.204
& 61.482
& -5.278***
& [-5.940,\;-4.616]
& $<$0.001 \\

Hiring occ.1
& 30.493
& 31.122
& -0.629***
& [-1.023,\;-0.234]
& 0.009 \\

Hiring occ.2
& 27.890
& 20.651
& 7.239***
& [6.838,\;7.641]
& $<$0.001 \\

Hiring occ.3
& 25.610
& 18.351
& 7.259***
& [6.863,\;7.654]
& $<$0.001 \\

Rel. occ.1
& 0.354
& 0.486
& -0.133***
& [-0.137,\;-0.128]
& $<$0.001 \\

Search effort
& 10.228
& 16.214
& -5.986***
& [-6.228,\;-5.744]
& $<$0.001 \\

Low effort
& 0.693
& 0.294
& 0.399***
& [0.386,\;0.412]
& $<$0.001 \\

Low wage prior
& 0.510
& 0.395
& 0.115***
& [0.101,\;0.129]
& $<$0.001 \\

Reserv. wage (k)
& 2.014
& 2.025
& -0.010
& [-0.028,\;0.007]
& 0.314 \\

Tension
& -0.001
& 0.492
& -0.493***
& [-0.509,\;-0.477]
& $<$0.001 \\

\bottomrule
\end{tabular}
}

\vspace{0.75em}

\parbox{\textwidth}{\footnotesize
\textit{Notes:} CLAN estimates compare the 20\% of individuals with the highest predicted treatment effects (Most Affected) to the 20\% with the lowest predicted treatment effects (Least Affected). Estimates are obtained using ENet+CB with cross-fitting. Numbers in brackets report 90\% confidence intervals. Significance levels are denoted by *** $p<0.01$, ** $p<0.05$, and * $p<0.10$.
}

\end{table}
\begin{table}[htbp]
\centering
\caption{Belief ATEs by 12-Month Reemployment GATES: Non-Pessimists}
\label{tab:bel_gates_occ_opt}

\scalebox{0.8}{
\begin{tabular}{l c c c c c c}
\toprule
Belief & G1 & G2 & G3 & G4 & G5 & Equality p-value \\
\midrule

$\Delta$ Perc. return
& \begin{tabular}[c]{@{}c@{}}0.110\\{\footnotesize[-0.1,\;0.3]}\end{tabular}
& \begin{tabular}[c]{@{}c@{}}0.120\\{\footnotesize[-0.1,\;0.3]}\end{tabular}
& \begin{tabular}[c]{@{}c@{}}0.090\\{\footnotesize[-0.1,\;0.3]}\end{tabular}
& \begin{tabular}[c]{@{}c@{}}0.070\\{\footnotesize[-0.1,\;0.2]}\end{tabular}
& \begin{tabular}[c]{@{}c@{}}0.040\\{\footnotesize[-0.1,\;0.2]}\end{tabular}
& 0.600
\\[1.5ex]

$\Delta$ Wage1
& \begin{tabular}[c]{@{}c@{}}-1.910\\{\footnotesize[-5.7,\;1.9]}\end{tabular}
& \begin{tabular}[c]{@{}c@{}}0.160\\{\footnotesize[-3.4,\;3.8]}\end{tabular}
& \begin{tabular}[c]{@{}c@{}}0.570\\{\footnotesize[-3.1,\;4.0]}\end{tabular}
& \begin{tabular}[c]{@{}c@{}}1.460\\{\footnotesize[-2.0,\;4.9]}\end{tabular}
& \begin{tabular}[c]{@{}c@{}}0.180\\{\footnotesize[-3.2,\;3.5]}\end{tabular}
& 0.462
\\[1.5ex]

$\Delta$ Wage2
& \begin{tabular}[c]{@{}c@{}}-1.210\\{\footnotesize[-4.5,\;2.0]}\end{tabular}
& \begin{tabular}[c]{@{}c@{}}-0.430\\{\footnotesize[-3.7,\;2.9]}\end{tabular}
& \begin{tabular}[c]{@{}c@{}}-0.080\\{\footnotesize[-3.2,\;3.0]}\end{tabular}
& \begin{tabular}[c]{@{}c@{}}1.020\\{\footnotesize[-2.0,\;4.0]}\end{tabular}
& \begin{tabular}[c]{@{}c@{}}0.090\\{\footnotesize[-2.7,\;2.9]}\end{tabular}
& 0.502
\\[1.5ex]

$\Delta$ Wage3
& \begin{tabular}[c]{@{}c@{}}0.150\\{\footnotesize[-3.6,\;3.7]}\end{tabular}
& \begin{tabular}[c]{@{}c@{}}0.110\\{\footnotesize[-3.5,\;3.6]}\end{tabular}
& \begin{tabular}[c]{@{}c@{}}-0.700\\{\footnotesize[-4.3,\;2.6]}\end{tabular}
& \begin{tabular}[c]{@{}c@{}}-1.300\\{\footnotesize[-4.7,\;2.1]}\end{tabular}
& \begin{tabular}[c]{@{}c@{}}-1.680\\{\footnotesize[-4.6,\;1.2]}\end{tabular}
& 0.621
\\[1.5ex]

$\Delta$ Marg. ret.
& \begin{tabular}[c]{@{}c@{}}-1.150\\{\footnotesize[-4.5,\;2.1]}\end{tabular}
& \begin{tabular}[c]{@{}c@{}}0.260\\{\footnotesize[-2.9,\;3.3]}\end{tabular}
& \begin{tabular}[c]{@{}c@{}}1.220\\{\footnotesize[-1.7,\;4.1]}\end{tabular}
& \begin{tabular}[c]{@{}c@{}}1.240\\{\footnotesize[-1.7,\;4.1]}\end{tabular}
& \begin{tabular}[c]{@{}c@{}}-0.290\\{\footnotesize[-3.1,\;2.6]}\end{tabular}
& 0.458
\\[1.5ex]

$\Delta$ Hiring1
& \begin{tabular}[c]{@{}c@{}}-0.590\\{\footnotesize[-2.1,\;1.0]}\end{tabular}
& \begin{tabular}[c]{@{}c@{}}-0.260\\{\footnotesize[-1.8,\;1.3]}\end{tabular}
& \begin{tabular}[c]{@{}c@{}}0.330\\{\footnotesize[-1.2,\;1.8]}\end{tabular}
& \begin{tabular}[c]{@{}c@{}}0.460\\{\footnotesize[-1.0,\;1.9]}\end{tabular}
& \begin{tabular}[c]{@{}c@{}}0.150\\{\footnotesize[-1.3,\;1.6]}\end{tabular}
& 0.500
\\[1.5ex]

$\Delta$ Hiring2
& \begin{tabular}[c]{@{}c@{}}0.320\\{\footnotesize[-1.3,\;2.0]}\end{tabular}
& \begin{tabular}[c]{@{}c@{}}-0.430\\{\footnotesize[-2.0,\;1.2]}\end{tabular}
& \begin{tabular}[c]{@{}c@{}}-0.240\\{\footnotesize[-1.8,\;1.3]}\end{tabular}
& \begin{tabular}[c]{@{}c@{}}0.410\\{\footnotesize[-1.1,\;1.9]}\end{tabular}
& \begin{tabular}[c]{@{}c@{}}0.520\\{\footnotesize[-0.9,\;1.9]}\end{tabular}
& 0.405
\\[1.5ex]

$\Delta$ Rel. occ.
& \begin{tabular}[c]{@{}c@{}}-0.020\\{\footnotesize[-0.0,\;0.0]}\end{tabular}
& \begin{tabular}[c]{@{}c@{}}-0.010\\{\footnotesize[-0.0,\;0.0]}\end{tabular}
& \begin{tabular}[c]{@{}c@{}}-0.000\\{\footnotesize[-0.0,\;0.0]}\end{tabular}
& \begin{tabular}[c]{@{}c@{}}-0.010\\{\footnotesize[-0.0,\;0.0]}\end{tabular}
& \begin{tabular}[c]{@{}c@{}}-0.000\\{\footnotesize[-0.0,\;0.0]}\end{tabular}
& 0.312
\\

\bottomrule
\end{tabular}
}

\vspace{0.75em}

\parbox{\textwidth}{\footnotesize
\textit{Notes:} Average treatment effects on belief outcomes by GATES group, where groups are formed using predicted treatment effects on 12-month reemployment. Estimates are obtained using cross-fitting with 50 sample splits. G1 denotes the least affected group and G5 the most affected group. Numbers in brackets report confidence intervals. The final column reports the median p-value from the interaction test between treatment status and GATES group.
}

\end{table}

\FloatBarrier

\subsection{On the Occupational Treatment - Pessimists}

\begin{table}[htbp]
\centering
\caption{GATES (Sorted): Occupational --- Pessimists}
\label{tab:gates_occ_pess}

\scalebox{0.8}{
\begin{tabular}{l ccccc cc}
\toprule
Outcome & G1 & G2 & G3 & G4 & G5 & G5$-$G1 & Equality p-value \\
\midrule

12-month reemployment
& \begin{tabular}[c]{@{}c@{}}-0.010\\{\footnotesize[-0.061,\;0.042]}\end{tabular}
& \begin{tabular}[c]{@{}c@{}}0.004\\{\footnotesize[-0.046,\;0.055]}\end{tabular}
& \begin{tabular}[c]{@{}c@{}}0.017\\{\footnotesize[-0.034,\;0.068]}\end{tabular}
& \begin{tabular}[c]{@{}c@{}}0.029\\{\footnotesize[-0.022,\;0.079]}\end{tabular}
& \begin{tabular}[c]{@{}c@{}}0.042\\{\footnotesize[-0.009,\;0.093]}\end{tabular}
& 0.050
& 0.298
\\[1.5ex]

Relative wage
& \begin{tabular}[c]{@{}c@{}}-0.030\\{\footnotesize[-0.085,\;0.024]}\end{tabular}
& \begin{tabular}[c]{@{}c@{}}-0.012\\{\footnotesize[-0.064,\;0.038]}\end{tabular}
& \begin{tabular}[c]{@{}c@{}}0.003\\{\footnotesize[-0.050,\;0.056]}\end{tabular}
& \begin{tabular}[c]{@{}c@{}}0.017\\{\footnotesize[-0.035,\;0.066]}\end{tabular}
& \begin{tabular}[c]{@{}c@{}}0.030\\{\footnotesize[-0.020,\;0.080]}\end{tabular}
& 0.060
& 0.608
\\

\bottomrule
\end{tabular}
}

\vspace{0.75em}

\parbox{\textwidth}{\footnotesize
\textit{Notes:} Estimates are obtained using 5-fold cross-fitting with 50 sample splits. GATES groups are sorted from the least affected group (G1) to the most affected group (G5). Numbers in brackets report 90\% confidence intervals. The final column reports the p-value for the equality test across groups.
}

\end{table}
\begin{table}[htbp]
\centering
\caption{CLAN: 12-Month Reemployment --- Pessimists (X-Learn (Ranger))}
\label{tab:clan_occ_pess_retour_12m}

\scalebox{0.8}{
\begin{tabular}{l c c c c c}
\toprule
Variable & Most Affected & Least Affected & Difference & 90\% CI & p-value \\
\midrule

Unemp. dur.
& 26.659 & 26.262 & 0.397 & [-1.144,\;1.938] & 0.672 \\

Job belief
& 15.302 & 16.135 & -0.833** & [-1.420,\;-0.247] & 0.019 \\

Hiring occ.1
& 21.797 & 19.658 & 2.139*** & [1.404,\;2.875] & $<0.001$ \\

Hiring occ.2
& 15.437 & 13.989 & 1.448*** & [0.757,\;2.139] & $<0.001$ \\

Hiring occ.3
& 13.396 & 12.772 & 0.624 & [-0.029,\;1.277] & 0.116 \\

Rel. occ.1
& 0.427 & 0.388 & 0.040*** & [0.031,\;0.048] & $<0.001$ \\

Search effort
& 11.825 & 11.454 & 0.371 & [-0.096,\;0.838] & 0.191 \\

Low effort
& 0.569 & 0.593 & -0.024 & [-0.049,\;0.001] & 0.120 \\

Low wage prior
& 0.448 & 0.479 & -0.031** & [-0.056,\;-0.005] & 0.048 \\

Reserv. wage (k)
& 1.942 & 1.925 & 0.018 & [-0.014,\;0.049] & 0.359 \\

Tension
& 0.025 & 0.356 & -0.331*** & [-0.362,\;-0.301] & $<0.001$ \\

\bottomrule
\end{tabular}
}

\vspace{0.75em}

\parbox{\textwidth}{\footnotesize
\textit{Notes:} Characteristics of individuals in the top and bottom 20\% of the predicted treatment-effect distribution. “Most Affected” refers to individuals with the highest predicted treatment effects and “Least Affected” to those with the lowest predicted treatment effects. Differences are computed as Most Affected minus Least Affected. Confidence intervals are 90\%. Estimates are obtained using X-Learn (Ranger) with cross-fitting.
}

\end{table}

\begin{table}[htbp]
\centering
\caption{Belief ATEs by 12-Month Reemployment GATES: Pessimists}
\label{tab:bel_gates_occ_pess}

\scalebox{0.8}{
\begin{tabular}{l c c c c c c}
\toprule
Belief & G1 & G2 & G3 & G4 & G5 & Equality p-value \\
\midrule

$\Delta$ Perc. return
& \begin{tabular}[c]{@{}c@{}}-0.040\\{\footnotesize[-0.3,\;0.2]}\end{tabular}
& \begin{tabular}[c]{@{}c@{}}0.010\\{\footnotesize[-0.2,\;0.3]}\end{tabular}
& \begin{tabular}[c]{@{}c@{}}-0.110\\{\footnotesize[-0.4,\;0.1]}\end{tabular}
& \begin{tabular}[c]{@{}c@{}}-0.190\\{\footnotesize[-0.4,\;0.0]}\end{tabular}
& \begin{tabular}[c]{@{}c@{}}-0.290*\\{\footnotesize[-0.5,\;-0.0]}\end{tabular}
& 0.324
\\[1.5ex]

$\Delta$ Wage1
& \begin{tabular}[c]{@{}c@{}}3.690\\{\footnotesize[-3.0,\;10.5]}\end{tabular}
& \begin{tabular}[c]{@{}c@{}}1.410\\{\footnotesize[-4.3,\;7.1]}\end{tabular}
& \begin{tabular}[c]{@{}c@{}}-0.840\\{\footnotesize[-6.3,\;4.5]}\end{tabular}
& \begin{tabular}[c]{@{}c@{}}-2.770\\{\footnotesize[-8.1,\;2.6]}\end{tabular}
& \begin{tabular}[c]{@{}c@{}}-4.110\\{\footnotesize[-9.9,\;1.6]}\end{tabular}
& 0.319
\\[1.5ex]

$\Delta$ Wage2
& \begin{tabular}[c]{@{}c@{}}0.910\\{\footnotesize[-4.9,\;6.8]}\end{tabular}
& \begin{tabular}[c]{@{}c@{}}1.140\\{\footnotesize[-4.3,\;6.5]}\end{tabular}
& \begin{tabular}[c]{@{}c@{}}0.630\\{\footnotesize[-4.2,\;5.4]}\end{tabular}
& \begin{tabular}[c]{@{}c@{}}-0.290\\{\footnotesize[-5.1,\;4.3]}\end{tabular}
& \begin{tabular}[c]{@{}c@{}}-2.150\\{\footnotesize[-7.0,\;2.8]}\end{tabular}
& 0.592
\\[1.5ex]

$\Delta$ Wage3
& \begin{tabular}[c]{@{}c@{}}2.370\\{\footnotesize[-3.6,\;8.1]}\end{tabular}
& \begin{tabular}[c]{@{}c@{}}1.510\\{\footnotesize[-4.1,\;7.2]}\end{tabular}
& \begin{tabular}[c]{@{}c@{}}1.410\\{\footnotesize[-3.8,\;6.3]}\end{tabular}
& \begin{tabular}[c]{@{}c@{}}0.640\\{\footnotesize[-4.0,\;5.4]}\end{tabular}
& \begin{tabular}[c]{@{}c@{}}0.800\\{\footnotesize[-4.5,\;6.1]}\end{tabular}
& 0.762
\\[1.5ex]

$\Delta$ Marg. ret.
& \begin{tabular}[c]{@{}c@{}}2.350\\{\footnotesize[-3.1,\;7.8]}\end{tabular}
& \begin{tabular}[c]{@{}c@{}}-0.140\\{\footnotesize[-5.2,\;4.8]}\end{tabular}
& \begin{tabular}[c]{@{}c@{}}-0.810\\{\footnotesize[-5.7,\;4.1]}\end{tabular}
& \begin{tabular}[c]{@{}c@{}}-0.500\\{\footnotesize[-5.1,\;3.5]}\end{tabular}
& \begin{tabular}[c]{@{}c@{}}-3.810\\{\footnotesize[-8.2,\;0.8]}\end{tabular}
& 0.345
\\[1.5ex]

$\Delta$ Hiring1
& \begin{tabular}[c]{@{}c@{}}1.000\\{\footnotesize[-1.6,\;3.6]}\end{tabular}
& \begin{tabular}[c]{@{}c@{}}-1.180\\{\footnotesize[-3.7,\;1.3]}\end{tabular}
& \begin{tabular}[c]{@{}c@{}}-0.820\\{\footnotesize[-3.3,\;1.6]}\end{tabular}
& \begin{tabular}[c]{@{}c@{}}0.000\\{\footnotesize[-2.4,\;2.4]}\end{tabular}
& \begin{tabular}[c]{@{}c@{}}-0.620\\{\footnotesize[-3.3,\;2.0]}\end{tabular}
& 0.542
\\[1.5ex]

$\Delta$ Hiring2
& \begin{tabular}[c]{@{}c@{}}1.010\\{\footnotesize[-1.8,\;3.7]}\end{tabular}
& \begin{tabular}[c]{@{}c@{}}-0.560\\{\footnotesize[-3.0,\;1.8]}\end{tabular}
& \begin{tabular}[c]{@{}c@{}}-0.510\\{\footnotesize[-2.8,\;1.7]}\end{tabular}
& \begin{tabular}[c]{@{}c@{}}-0.340\\{\footnotesize[-2.6,\;1.9]}\end{tabular}
& \begin{tabular}[c]{@{}c@{}}-1.630\\{\footnotesize[-4.0,\;0.7]}\end{tabular}
& 0.501
\\[1.5ex]

$\Delta$ Rel. occ.
& \begin{tabular}[c]{@{}c@{}}-0.020\\{\footnotesize[-0.1,\;0.0]}\end{tabular}
& \begin{tabular}[c]{@{}c@{}}-0.030\\{\footnotesize[-0.1,\;0.0]}\end{tabular}
& \begin{tabular}[c]{@{}c@{}}-0.020\\{\footnotesize[-0.1,\;0.0]}\end{tabular}
& \begin{tabular}[c]{@{}c@{}}-0.010\\{\footnotesize[-0.0,\;0.0]}\end{tabular}
& \begin{tabular}[c]{@{}c@{}}0.020\\{\footnotesize[-0.0,\;0.1]}\end{tabular}
& 0.296
\\

\bottomrule
\end{tabular}
}

\vspace{0.75em}

\parbox{\textwidth}{\footnotesize
\textit{Notes:} Average treatment effects on belief outcomes by GATES group, where groups are formed using predicted treatment effects on 12-month reemployment. Estimates are obtained using cross-fitting with 50 sample splits. G1 denotes the least affected group and G5 the most affected group. Numbers in brackets report confidence intervals. The final column reports the median p-value from the interaction test between treatment status and GATES group.
}

\end{table}

\FloatBarrier

\subsection{On the Motivational Treatment }

\begin{table}[htbp]
\centering
\caption{GATES (Sorted): Motivational --- Pessimists}
\label{tab:gates_mot}

\scalebox{0.8}{
\begin{tabular}{l c c c c c c c}
\toprule
Outcome & G1 & G2 & G3 & G4 & G5 & G5$-$G1 & p-value \\
\midrule

12-month reemployment
& \begin{tabular}[c]{@{}c@{}}-0.009\\{\footnotesize[-0.060,\;0.042]}\end{tabular}
& \begin{tabular}[c]{@{}c@{}}0.009\\{\footnotesize[-0.042,\;0.060]}\end{tabular}
& \begin{tabular}[c]{@{}c@{}}0.020\\{\footnotesize[-0.031,\;0.071]}\end{tabular}
& \begin{tabular}[c]{@{}c@{}}0.035\\{\footnotesize[-0.015,\;0.086]}\end{tabular}
& \begin{tabular}[c]{@{}c@{}}\textbf{0.053*}\\{\footnotesize[0.002,\;0.103]}\end{tabular}
& 0.062
& 0.441
\\[1.5ex]

Relative wage
& \begin{tabular}[c]{@{}c@{}}-0.013\\{\footnotesize[-0.070,\;0.042]}\end{tabular}
& \begin{tabular}[c]{@{}c@{}}-0.001\\{\footnotesize[-0.055,\;0.052]}\end{tabular}
& \begin{tabular}[c]{@{}c@{}}0.012\\{\footnotesize[-0.041,\;0.065]}\end{tabular}
& \begin{tabular}[c]{@{}c@{}}0.019\\{\footnotesize[-0.036,\;0.073]}\end{tabular}
& \begin{tabular}[c]{@{}c@{}}0.032\\{\footnotesize[-0.018,\;0.084]}\end{tabular}
& 0.046
& 0.645
\\

\bottomrule
\end{tabular}
}

\vspace{0.75em}

\parbox{\textwidth}{\footnotesize
\textit{Notes:} Group Average Treatment Effects (GATES) based on treatment-effect rankings from the preferred learner. Individuals are sorted into five equally sized groups (G1--G5) according to predicted treatment effects, with G5 containing the highest predicted gains. Estimates are averaged over 50 sample splits using 5-fold cross-fitting. Brackets report 90\% confidence intervals. The final two columns report the difference between the top and bottom groups (G5$-$G1) and the associated test p-value.
}

\end{table}
\begin{table}[htbp]
\centering
\caption{CLAN: 12-Month Reemployment --- Pessimists (ENet+CB)}
\label{tab:clan_mot_retour_12m}

\scalebox{0.8}{
\begin{tabular}{l c c c c c}
\toprule
Variable & Most Affected & Least Affected & Difference & 90\% CI & p-value \\
\midrule

Unemp. dur.
& 24.732 & 26.124 & -1.392 & [-2.868,\;0.084] & 0.121 \\

Job belief
& 17.712 & 14.578 & 3.134*** & [2.543,\;3.725] & $<0.001$ \\

Hiring occ.1
& 27.353 & 15.387 & 11.965*** & [11.255,\;12.676] & $<0.001$ \\

Hiring occ.2
& 21.148 & 9.594 & 11.553*** & [10.872,\;12.235] & $<0.001$ \\

Hiring occ.3
& 19.195 & 9.048 & 10.147*** & [9.493,\;10.802] & $<0.001$ \\

Rel. occ.1
& 0.379 & 0.474 & -0.095*** & [-0.105,\;-0.086] & $<0.001$ \\

Search effort
& 7.653 & 15.101 & -7.448*** & [-7.869,\;-7.028] & $<0.001$ \\

Low effort
& 0.876 & 0.333 & 0.543*** & [0.522,\;0.564] & $<0.001$ \\

Low wage prior
& 0.491 & 0.462 & 0.029* & [0.004,\;0.055] & 0.061 \\

Reserv. wage (k)
& 1.887 & 1.979 & -0.092*** & [-0.123,\;-0.060] & $<0.001$ \\

Tension
& -0.128 & 0.579 & -0.707*** & [-0.737,\;-0.678] & $<0.001$ \\

\bottomrule
\end{tabular}
}

\vspace{0.75em}

\parbox{\textwidth}{\footnotesize
\textit{Notes:} Characteristics of individuals in the top and bottom 20\% of the predicted treatment-effect distribution. “Most Affected” refers to individuals with the highest predicted treatment effects and “Least Affected” to those with the lowest predicted treatment effects. Differences are computed as Most Affected minus Least Affected. Confidence intervals are 90\%. Estimates are obtained using ENet+CB with cross-fitting.
}

\end{table}
\begin{table}[htbp]
\centering
\caption{CLAN: Relative Wage --- Pessimists (ENet+CB)}
\label{tab:clan_mot_rel_wage}

\scalebox{0.8}{
\begin{tabular}{l c c c c c}
\toprule
Variable & Most Affected & Least Affected & Difference & 90\% CI & p-value \\
\midrule

Unemp. dur.
& 19.447 & 26.912 & -7.465 & [-10.153,\;-4.778] & $<0.001$ \\

Job belief
& 19.208 & 17.941 & 1.268 & [0.069,\;2.466] & 0.082 \\

Hiring occ.1
& 22.524 & 21.456 & 1.068 & [-0.482,\;2.619] & 0.257 \\

Hiring occ.2
& 17.557 & 17.316 & 0.241 & [-1.243,\;1.725] & 0.789 \\

Hiring occ.3
& 16.098 & 15.718 & 0.380 & [-1.047,\;1.807] & 0.661 \\

Rel. occ.1
& 0.411 & 0.410 & 0.001 & [-0.017,\;0.019] & 0.936 \\

Search effort
& 19.758 & 6.554 & 13.205 & [12.502,\;13.907] & $<0.001$ \\

Low effort
& 0.101 & 0.950 & -0.849 & [-0.875,\;-0.822] & $<0.001$ \\

Low wage prior
& 0.908 & 0.037 & 0.871 & [0.846,\;0.895] & $<0.001$ \\

Reserv. wage (k)
& 1.878 & 1.956 & -0.078 & [-0.126,\;-0.031] & 0.007 \\

Tension
& 0.459 & 0.020 & 0.439 & [0.380,\;0.498] & $<0.001$ \\

\bottomrule
\end{tabular}
}

\vspace{0.75em}

\parbox{\textwidth}{\footnotesize
\textit{Notes:} Characteristics of individuals in the top and bottom 20\% of the predicted treatment-effect distribution. “Most Affected” refers to individuals with the highest predicted treatment effects and “Least Affected” to those with the lowest predicted treatment effects. Differences are computed as Most Affected minus Least Affected. Confidence intervals are 90\%. Estimates are obtained using cross-fitting.
}

\end{table}
\begin{table}[htbp]
\centering
\caption{Belief ATEs by 12-Month Reemployment GATES: All Pessimists}
\label{tab:bel_gates_mot}

\scalebox{0.8}{
\begin{tabular}{l c c c c c c}
\toprule
Belief & G1 & G2 & G3 & G4 & G5 & Equality p-value \\
\midrule

$\Delta$ Perc. return
& \begin{tabular}[c]{@{}c@{}}-0.060\\{\footnotesize[-0.300,\;0.200]}\end{tabular}
& \begin{tabular}[c]{@{}c@{}}-0.100\\{\footnotesize[-0.300,\;0.100]}\end{tabular}
& \begin{tabular}[c]{@{}c@{}}0.020\\{\footnotesize[-0.200,\;0.200]}\end{tabular}
& \begin{tabular}[c]{@{}c@{}}-0.020\\{\footnotesize[-0.300,\;0.200]}\end{tabular}
& \begin{tabular}[c]{@{}c@{}}-0.260\\{\footnotesize[-0.500,\;0.000]}\end{tabular}
& 0.421
\\[1.5ex]

$\Delta$ Wage1
& \begin{tabular}[c]{@{}c@{}}-2.110\\{\footnotesize[-8.000,\;3.700]}\end{tabular}
& \begin{tabular}[c]{@{}c@{}}-0.670\\{\footnotesize[-6.500,\;5.000]}\end{tabular}
& \begin{tabular}[c]{@{}c@{}}-1.150\\{\footnotesize[-6.900,\;4.700]}\end{tabular}
& \begin{tabular}[c]{@{}c@{}}-0.760\\{\footnotesize[-6.700,\;5.400]}\end{tabular}
& \begin{tabular}[c]{@{}c@{}}-2.700\\{\footnotesize[-8.400,\;3.000]}\end{tabular}
& 0.807
\\[1.5ex]

$\Delta$ Wage2
& \begin{tabular}[c]{@{}c@{}}-1.870\\{\footnotesize[-7.200,\;3.400]}\end{tabular}
& \begin{tabular}[c]{@{}c@{}}-1.760\\{\footnotesize[-6.800,\;3.200]}\end{tabular}
& \begin{tabular}[c]{@{}c@{}}0.620\\{\footnotesize[-4.400,\;5.600]}\end{tabular}
& \begin{tabular}[c]{@{}c@{}}1.610\\{\footnotesize[-3.300,\;6.600]}\end{tabular}
& \begin{tabular}[c]{@{}c@{}}0.760\\{\footnotesize[-4.300,\;5.800]}\end{tabular}
& 0.579
\\[1.5ex]

$\Delta$ Wage3
& \begin{tabular}[c]{@{}c@{}}2.900\\{\footnotesize[-2.900,\;8.600]}\end{tabular}
& \begin{tabular}[c]{@{}c@{}}1.210\\{\footnotesize[-3.900,\;6.300]}\end{tabular}
& \begin{tabular}[c]{@{}c@{}}1.700\\{\footnotesize[-3.400,\;6.800]}\end{tabular}
& \begin{tabular}[c]{@{}c@{}}2.060\\{\footnotesize[-3.100,\;7.300]}\end{tabular}
& \begin{tabular}[c]{@{}c@{}}0.730\\{\footnotesize[-4.500,\;5.700]}\end{tabular}
& 0.741
\\[1.5ex]

$\Delta$ Marg. ret.
& \begin{tabular}[c]{@{}c@{}}2.500\\{\footnotesize[-2.400,\;7.300]}\end{tabular}
& \begin{tabular}[c]{@{}c@{}}0.030\\{\footnotesize[-4.400,\;4.300]}\end{tabular}
& \begin{tabular}[c]{@{}c@{}}-0.830\\{\footnotesize[-5.400,\;3.900]}\end{tabular}
& \begin{tabular}[c]{@{}c@{}}-0.700\\{\footnotesize[-5.400,\;4.200]}\end{tabular}
& \begin{tabular}[c]{@{}c@{}}-4.480\\{\footnotesize[-10.000,\;1.000]}\end{tabular}
& 0.300
\\[1.5ex]

$\Delta$ Hiring1
& \begin{tabular}[c]{@{}c@{}}2.940\\{\footnotesize[0.400,\;5.400]}\end{tabular}
& \begin{tabular}[c]{@{}c@{}}-0.010\\{\footnotesize[-2.500,\;2.400]}\end{tabular}
& \begin{tabular}[c]{@{}c@{}}-0.660\\{\footnotesize[-3.100,\;1.800]}\end{tabular}
& \begin{tabular}[c]{@{}c@{}}0.220\\{\footnotesize[-2.400,\;2.700]}\end{tabular}
& \begin{tabular}[c]{@{}c@{}}-0.520\\{\footnotesize[-3.300,\;2.200]}\end{tabular}
& 0.276
\\[1.5ex]

$\Delta$ Hiring2
& \begin{tabular}[c]{@{}c@{}}0.210\\{\footnotesize[-1.700,\;2.300]}\end{tabular}
& \begin{tabular}[c]{@{}c@{}}-0.320\\{\footnotesize[-2.600,\;2.100]}\end{tabular}
& \begin{tabular}[c]{@{}c@{}}0.150\\{\footnotesize[-2.100,\;2.400]}\end{tabular}
& \begin{tabular}[c]{@{}c@{}}0.160\\{\footnotesize[-2.200,\;2.700]}\end{tabular}
& \begin{tabular}[c]{@{}c@{}}-0.990\\{\footnotesize[-3.500,\;1.500]}\end{tabular}
& 0.496
\\[1.5ex]

%$\Delta$ Rel. occ.
%& \begin{tabular}[c]{@{}c@{}}0.030\\{\footnotesize[0.000,\;0.100]}\end{tabular}
%& \begin{tabular}[c]{@{}c@{}}0.010\\{\footnotesize[0.000,\;0.000]}\end{tabular}
%& \begin{tabular}[c]{@{}c@{}}0.000\\{\footnotesize[0.000,\;0.000]}\end{tabular}
%& \begin{tabular}[c]{@{}c@{}}0.010\\{\footnotesize[0.000,\;0.000]}\end{tabular}
%& \begin{tabular}[c]{@{}c@{}}0.010\\{\footnotesize[0.000,\;0.000]}\end{tabular}
%& 0.544
%\\

\bottomrule
\end{tabular}
}

\vspace{0.75em}

\parbox{\textwidth}{\footnotesize
\textit{Notes:} Average treatment effects on beliefs by GATES group based on predicted treatment effects for 12-month reemployment. Individuals are sorted into five equally sized groups (G1--G5), with G5 containing those with the highest predicted treatment effects. Brackets report confidence intervals. Estimates are averaged over 50 sample splits using cross-fitting. The final column reports the median p-value from the ANOVA test of equality across groups.
}

\end{table}

\FloatBarrier

\section{Complementary Results on Beliefs updating}\label{sec:Beliefs_updating_GML}

\subsection{Average effects}\label{sec:average_beliefs}

\begin{table}[htbp]
\centering
\caption{OLS: Occupational Treatment (S2) on Belief Updates --- All Population}
\label{tab:ols_beliefs_occ_all}

\scalebox{0.8}{
\begin{tabular}{l c c c c}
\toprule
Outcome & Coefficient & Robust SE & N & Control Mean \\
\midrule

$\Delta$ Wage belief (q1)
& -0.137
& (0.663)
& 9,161
& 7.161 \\

$\Delta$ Wage belief (median)
& 0.186
& (0.615)
& 8,830
& 8.395 \\

$\Delta$ Wage belief (q3)
& -0.118
& (0.672)
& 7,924
& 7.345 \\

$\Delta$ Hiring probability (occ.~1)
& 0.545*
& (0.325)
& 8,556
& -3.005 \\

$\Delta$ Hiring probability (occ.~2)
& 0.277
& (0.313)
& 8,368
& -2.399 \\

$\Delta$ Relative occupation belief
& -0.003
& (0.004)
& 7,718
& 0.005 \\

$\Delta$ Marginal return
& 0.669
& (0.617)
& 7,471
& 0.108 \\

$\Delta$ Perceived return
& 0.065*
& (0.033)
& 6,895
& -0.054 \\

\bottomrule
\end{tabular}
}

\vspace{0.75em}

\parbox{\textwidth}{\footnotesize
\textit{Notes:} Ordinary least squares estimates of the effect of the occupational treatment (S2) relative to the control group on belief updating. Belief updates are defined as follow-up minus baseline beliefs and are observed only for survey respondents. All specifications control for the corresponding baseline belief and occupation fixed effects. Robust HC1 standard errors are reported in parentheses. $^{*}p<0.10$, $^{**}p<0.05$, and $^{***}p<0.01$.
}

\end{table}

\begin{table}[htbp]
\centering
\caption{OLS Treatment Effects on Belief Updates --- Pessimists}
\label{tab:ols_beliefs}

\scalebox{0.8}{
\begin{tabular}{l cc cc}
\toprule
& \multicolumn{2}{c}{S1 (Motivational)}
& \multicolumn{2}{c}{S2 (Occupational)} \\
\cmidrule(lr){2-3}\cmidrule(lr){4-5}
Outcome & Coefficient & Robust SE & Coefficient & Robust SE \\
\midrule

$\Delta$ Wage belief (q1)
& -0.785 & (1.354)
& -0.037 & (1.357) \\

$\Delta$ Wage belief (median)
& 0.168 & (1.250)
& 0.702 & (1.238) \\

$\Delta$ Wage belief (q3)
& 1.639 & (1.243)
& 1.960 & (1.256) \\

$\Delta$ Hiring probability (occ.~1)
& 0.675 & (0.643)
& -0.069 & (0.643) \\

$\Delta$ Hiring probability (occ.~2)
& 0.498 & (0.575)
& 0.078 & (0.584) \\

$\Delta$ Relative occupation belief
& 0.012 & (0.009)
& -0.003 & (0.009) \\

$\Delta$ Marginal return
& -0.758 & (1.186)
& -1.040 & (1.213) \\

$\Delta$ Perceived return
& -0.046 & (0.061)
& -0.062 & (0.061) \\

\midrule
N (respondents) & \multicolumn{4}{c}{3,029} \\
\bottomrule
\end{tabular}
}

\vspace{0.75em}

\parbox{\textwidth}{\footnotesize
\textit{Notes:} Ordinary least squares estimates of the effects of the motivational treatment (S1) and occupational treatment (S2) relative to the control group on belief updating among pessimists. Belief updates are defined as follow-up minus baseline beliefs and are observed only for survey respondents. All specifications control for the corresponding baseline belief and occupation fixed effects. Robust HC1 standard errors are reported in parentheses. $^{*}p<0.10$, $^{**}p<0.05$, and $^{***}p<0.01$.
}

\end{table}

\subsection{Heterogeneous effects using Generic Machine Learning}\label{sec:het_beliefs}

\begin{table}[htbp]
\centering
\caption{GML: Treatment Effects on Belief Updates}
\label{tab:gml_beliefs}

\scalebox{0.8}{
\begin{tabular}{l cccc cccc}
\toprule
& \multicolumn{4}{c}{Occupational (S2 vs Control, All)}
& \multicolumn{4}{c}{Motivational (S1 vs Control, Pessimists)} \\
\cmidrule(lr){2-5}\cmidrule(lr){6-9}
Outcome
& ATE & p-value & HET & p-value
& ATE & p-value & HET & p-value \\
\midrule

$\Delta$ Perceived return
& 0.040 & 0.915 & 0.278 & $<0.001$
& -0.082 & 0.316 & 0.495 & $<0.001$ \\

$\Delta$ Wage belief (occ.~1)
& -0.768 & 0.577 & 0.024 & 0.031
& -1.585 & 0.377 & 0.032 & 0.007 \\

$\Delta$ Wage belief (occ.~2)
& -0.284 & 0.990 & -0.007$^{\dagger}$ & 0.029
& -0.226 & 0.913 & 0.098 & 0.106 \\

$\Delta$ Wage belief (occ.~3)
& -0.485 & 0.910 & 0.001 & 0.021
& 1.367 & 0.311 & 0.046 & 0.296 \\

$\Delta$ Hiring probability (occ.~1)
& 0.095 & 0.797 & 0.010 & $<0.001$
& 0.867 & 0.536 & 0.122 & 0.073 \\

$\Delta$ Hiring probability (occ.~2)
& -0.019 & 0.557 & 0.054 & 0.074
& 0.015 & 0.966 & 0.006 & 0.002 \\

$\Delta$ Relative occupation belief
& -0.002** & 0.043 & 0.177 & 0.005
& 0.015* & 0.085 & -0.013$^{\dagger}$ & 0.091 \\

\bottomrule
\end{tabular}
}

\vspace{0.75em}

\parbox{\textwidth}{\footnotesize
\textit{Notes:} Best Linear Predictor (BLP) estimates from the generalized machine learning (GML) procedure. For each outcome, the preferred learner is selected as the specification with the highest median heterogeneity coefficient across sample splits. ATE denotes the median average treatment effect across 100 sample splits. HET denotes the heterogeneity coefficient ($\beta_2$) from the best-performing split. Belief updates are defined as follow-up minus baseline beliefs and are observed only for survey respondents. Statistical significance of ATE estimates is denoted by $^{*}p<0.10$, $^{**}p<0.05$, and $^{***}p<0.01$. The symbol $^{\dagger}$ indicates a negative heterogeneity estimate.
}

\end{table}

\end{document}